\documentclass{amsart}

\RequirePackage[OT1]{fontenc}

\usepackage{amssymb,latexsym,bbm}
\usepackage{geometry}
\usepackage{graphicx}
\usepackage{color}
\usepackage{subfigure}
\usepackage{multirow}
\usepackage{mathptmx}
\usepackage{url,bm}
\usepackage{enumerate}
\usepackage{float}
\usepackage[ruled,vlined,linesnumbered,noresetcount]{algorithm2e}
\makeatletter
\newcommand{\AlgoResetCount}{\renewcommand{\@ResetCounterIfNeeded}{\setcounter{AlgoLine}{0}}}
\newcommand{\AlgoNoResetCount}{\renewcommand{\@ResetCounterIfNeeded}{}}
\newcounter{AlgoSavedLineCount}

\SetKwInOut{Parameter}{Parameters}
\makeatother

\newtheorem{theorem}{Theorem}
\newtheorem{assumption}{Assumption}
\newtheorem{definition}{Definition}
\newtheorem{lemma}{Lemma}
\newtheorem{corollary}{Corollary}
\newtheorem{proposition}{Proposition}
\newtheorem{remark}{Remark}

\title[Graph based Gaussian process]{Graph based Gaussian processes on restricted domains }

\author{David B Dunson}
\address{David B Dunson\\
Department of Statistical Science\\
Duke University}
\email{dunson@duke.edu}

\author{Hau-Tieng~Wu}
\address{Hau-Tieng Wu\\
Departments of Mathematics and Department of Statistical Science\\
Duke University}
\email{hauwu@math.duke.edu}

\author{Nan~WU}
\address{Nan Wu\\
Department of Mathematics\\
Duke University}
\email{nan.wu@duke.edu}

\begin{document}
\maketitle

\begin{abstract}
In nonparametric regression, it is common for the inputs to fall in a restricted subset of Euclidean space.  Typical kernel-based methods that do not take into account the intrinsic geometry of the domain across which observations are collected may produce sub-optimal results.  In this article, we focus on solving this problem in the context of Gaussian process (GP) models, proposing a new class of Graph Laplacian based GPs (GL-GPs), which learn a covariance that respects the geometry of the input domain.  As the heat kernel is intractable computationally, we approximate the covariance using finitely-many eigenpairs of the Graph Laplacian (GL).  The GL is constructed from a kernel which depends only on the Euclidean coordinates of the inputs. Hence, we can benefit from the full knowledge about the kernel to extend the covariance structure to newly arriving samples by a Nystr\"{o}m type extension. We provide substantial theoretical support for the GL-GP methodology, and illustrate performance gains in various applications.
\end{abstract}

\section{Introduction}

We are interested in problems in which data are collected on `inputs' $x = (x_{1},\ldots,x_{D}) \in S \subset\mathbb{R}^D$ and `outputs' $y \in \mathbb{R}$, with 
$S$ a subset of $\mathbb{R}^D$.  Labeled training data $\mathcal{D} = \{ x_i,y_i\}_{i=1}^m$ are available for samples $i=1,\ldots,m$, along with (possibly) unlabelled data $\mathcal{U} = \{x_i\}_{i=m+1}^{m+n}$.  There are many settings in which this problem arises, including nonparametric regression focused on using features $x$ to predict outcome $y$ and computer model emulation in which $x$ corresponds to an input into a computer model that is expensive to run and $y$ corresponds to the output.  Gaussian process (GP) models are routinely used in these settings, but without explicitly taking into account the geometry of $S$ or using the unlabelled data to improve estimation of the unknown input-output function.  It is common for $S$ to be highly restricted and non-linear. For example,  in computer model emulation, the outputs commonly satisfy some physical laws or differential equations constrained over the domain of the inputs with complicted geometry.  

For concreteness, we focus on the following model throughout the paper, while noting that many elaborations are straightforward within the framework we will propose, 
\begin{eqnarray}
y_i = f(x_i) + \eta_i,\quad \eta_i \sim N(0,\sigma_{noise}^2), \label{eq:base}
\end{eqnarray}
where $f: S \to \mathbb{R}$ is an unknown regression function that is assigned a GP prior centred at zero with covariance function $C(x,x')$ and $\sigma_{noise}^2$ is the measurement error variance (this can be set to zero for deterministic computer models).  The choice of $C(x,x')$ has a fundamental impact on the results; the most common choices of covariance function are the squared exponential and its generalization, the Mat\'ern.  Both choices depend critically on the distance between inputs $x$ and $x'$; by far the most common choice in practice is the Euclidean distance.  

We use the following example to illustrate the problems with ignoring the geometry of $S$ and simply using a Euclidean distance-based kernel.  Raynaud's disease is a disorder of the blood vessels in the fingers and toes. When a person is cold or feels emotional stress, it causes the blood vessels to narrow so that the blood can not get to the skin. As a result, the affected parts on the fingers and toes turn white and then blue and there is a significant difference between the temperature over the affected part and the unaffected part.  In figure \ref{introhand}, we plot a simulation of the skin temperature (degree Celsius) of a patient with Raynaud's disease from a 3D scan of a right hand, a surface in $\mathbb{R}^3$. We hold out the temperature values at all but a relatively small number of sensor locations; the top left panel shows the labelled data and the top right all the data.
\begin{figure}[htb!]
\centering
\includegraphics[width=1 \columnwidth]{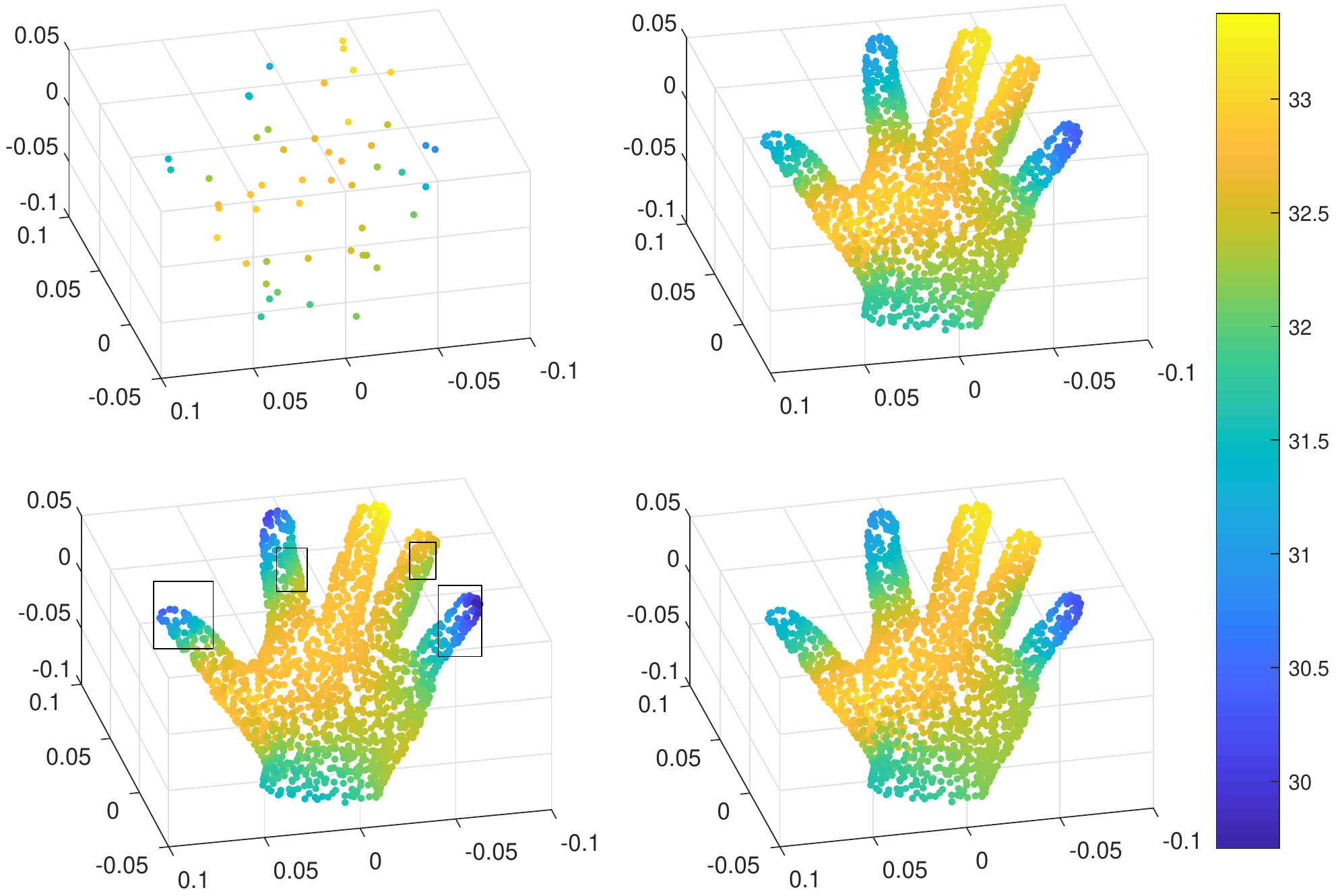}\label{introhand}
\caption{Points are sampled from a 3D scan of a right hand. Each point corresponds to the skin temperature (degree Celsius).  Top left panel shows a subset of labelled data and the top right panel shows all the data. We fit GP models to predict held-out temperature data.  The bottom left panel shows the results for a GP with squared exponential covariance, while the bottom right panel shows the results for our GL-GP.}
\end{figure}
We fit a GP with the square exponential of the Euclidean distance in the covariance, and show the predicted values in the bottom left panel of  Figure \ref{introhand}. There is a substantial change in temperature  between the index finger and the middle finger and between the ring finger and the little finger. For the GP with squared exponential covariance, 
 there is inappropriate smoothing between different fingers, leading to poor predictions over the regions indicated in the boxes.  

There have been attempts to solve related problems in the literature. \cite{cheng2013local} assume $S$ is an unknown submanifold in a Euclidean space and develop a locally linear regression method on the manifold; see also \cite{zhu2003semi,zhu2005semi,belkin2006manifold,nadler2009semi,dunlop2019large,wang2015nonparametric} for  semi-supervised approaches. An alternative is to choose a covariance in a GP prior that respects the geometry of $S$, but it is not clear how to specify such a covariance.  When the subset $S$ is a submanifold in a Euclidean space, the heat kernel of $S$ that characterizes the diffusion of heat flowing out of a point in $S$ is a potentially appealing choice.  \cite{castillo2014thomas} studies theoretical properties of the posterior distribution, such as rates of contraction, for related models.  Unfortunately, the heat kernel is typically intractable to calculate.  For known manifolds, \cite{lin2018extrinsic} propose an extrinsic GP that embeds the manifold in a higher-dimensional Euclidean space, while \cite{niu2019intrinsic} propose an intrinsic GP relying on a Monte Carlo (MC) approximation to the heat kernel.  The intrinsic GP is only applicable to known manifolds and their MC approximation is very expensive computationally, relying on simulating Brownian motion many times and calculating proportions of paths from a starting point ending up close to a target point. {There are also valid covariance structures defined on a submanifold in a Euclidean space other than the heat kernel. For example, the generalized M\'atern covariance is discussed in \cite{li2021fixed, borovitskiy2020mat}. When the manifold is known, such covariances can be approximated by numerically solving stochastic partial differential equations on the manifold.}

In this article, we develop a novel Graph Laplacian-based GP (GL-GP) to solve the above problem with predictors on an unknown subset $S$ of $\mathbb{R}^D$, which is not necessarily a manifold.  The key idea is to construct a covariance that incorporates the intrinsic geometry of $S$.  This is accomplished through taking {\em finitely} many eigenpairs of the graph Laplacian (GL) using the labelled and unlabelled predictor values.  The GL is widely studied in spectral graph theory \cite{Fan:1996}, and corresponds to the infinitesimal generator of a random walk on the sampled data points.  The covariance in the GL-GP approximates a diffusion process on $S$ with respect to intrinsic distances between data points. Tuning parameters in the covariance functions can be estimated by maximizing marginal likelihoods \cite{rasmussen2003gaussian}.
We propose a Nystr\"{o}m extension method to extend the covariance structure determined from an existing dataset to a newly arriving dataset. To provide a teaser illustrating practical advantages of the GL-GP, we fit our GL-GP to the skin temperature example.  The predicted values are shown in the bottom right panel of Figure \ref{introhand}. 

The proposed GL-GP can be used broadly in place of existing GP models. The method is designed to adapt to the support of the sample points, and one does not need to know {\em a priori} that the data have constrained support.  We find in practice that the GL-GP often outperforms GPs with standard off the shelf covariance functions in general applications, particularly when the predictors have a nontrivial geometric or topological structure.  In addition to the novel GL-GP methodology, a major contribution of this paper is the theory we have developed in support of the GL-GP.  We first  associate the GL with an integral operator for any subset $S$ of the Euclidean space. We discuss theoretical properties of the integral operator and define the covariance function of the GL-GP by using finitely many eigenpairs of the integral operator.  When $S$ is an embedded submanifold of $\mathbb{R}^D$, we provide the convergence rate of the GL-GP covariance matrix and its Nystr\"{o}m extension to the GL-GP covariance function. We show the stability of the GL (hence the GL-GP algorithm) when there are measurement errors so that predictors do not fall exactly within $S$.  Finally, when $S$ is an embedded submanifold of $\mathbb{R}^D$, we provide theory on contraction of the posterior for $f$ around the true function $f_0$ under some regularity assumptions.

\section{Graph based Gaussian processes on subsets of Euclidean space}

Focusing on equation (\ref{eq:base}), we propose a new approach for choosing the covariance function $\mathsf{C}$ in the Gaussian process; our proposed approach differs from current standard methods in not treating $\mathsf{C}$ as a pre-specified function having a small number of unknown parameters (e.g, squared exponential or M\'atern) but instead estimates the covariance in a manner that takes into account the structure of the support $S$ as well as information in both the labeled and unlabeled data.  Before introducing the proposed covariance, we provide a review of prediction based on GPs.  

\subsection{Gaussian process review}

Denote $\textbf{f}\in\mathbb{R}^m$ to be the discretization of a continuous function $f$ over $x_1, x_2, \cdots , x_m$ so that $\textbf{f}(i)=f(x_i)$.  A GP prior for $f$ 
implies 
$p(\textbf{f}|x_1,x_2, \cdots, x_m)=\mathcal{N}(0, \Sigma_{\textbf{f}\textbf{f}}),$ 
where $\Sigma_{\textbf{f}\textbf{f}} \in \mathbb{R}^{m\times m}$ is the covariance matrix induced from $\mathsf{C}$, with the $(i,j)$ element of $\Sigma_{\textbf{f}\textbf{f}}$ corresponding to $\mbox{cov}\{f(x_i),f(x_j)\}=\mathsf{C}(x_i,x_j)$, for $1 \leq i,j \leq m$.  The prior distribution $\mathcal{N}(0, \Sigma_{\textbf{f}\textbf{f}})$ can be combined with the information in the likelihood function under model (\ref{eq:base}) to obtain the posterior distribution, which will be used as a basis for inference.

We want to predict $f(x_i)$, where $x_i\in \mathcal{U}$.  Denote $\textbf{f}_*\in\mathbb{R}^n$ with $\textbf{f}_*(i)=f(x_{m+i})$ for $i=1,\ldots,n$. Under a GP prior for $f$, the joint distribution of $\textbf{f}$ and $\textbf{f}_*$ is 
$p(\textbf{f},\textbf{f}_*)=\mathcal{N}(0, \Sigma),$ 
where $\Sigma$ is an $(m+n) \times (m+n)$ covariance matrix that can be expressed as 
$
\Sigma= \begin{bmatrix}
\Sigma_{\textbf{f}\textbf{f}} & \Sigma_{\textbf{f}\textbf{f}_*}  \\
\Sigma_{\textbf{f}_*\textbf{f}} & \Sigma_{\textbf{f}_*\textbf{f}_*}
\end{bmatrix}\,, \nonumber 
$
where $\Sigma_{\textbf{f}\textbf{f}_*} \in \mathbb{R}^{m\times n}$, $\Sigma_{\textbf{f}_*\textbf{f}}\in \mathbb{R}^{n\times m}$, and $\Sigma_{\textbf{f}_*\textbf{f}_*}\in \mathbb{R}^{n\times n}$ are induced from the covariance function $\mathsf{C}$ respectively.
Denote $\textbf{y}\in\mathbb{R}^m$ with $\textbf{y}(i)=y_i$ for $i=1, \cdots, m$ to be the observations over $\{x_1, x_2, \cdots, x_m\}$. Under model \eqref{eq:base} and a GP prior, we have 
$p(\textbf{y}, \textbf{f}_*)=\mathcal{N}(0, \tilde{\Sigma}), $
where
\begin{align}
\tilde{\Sigma}=\Sigma+\begin{bmatrix}
\sigma^2_{noise} I_{m\times m} & 0  \\
0 & 0 \\
\end{bmatrix}= \begin{bmatrix}
\Sigma_{\textbf{f}\textbf{f}} +\sigma^2_{noise} I_{m\times m} & \Sigma_{\textbf{f}\textbf{f}_*}  \\
\Sigma_{\textbf{f}_*\textbf{f}} & \Sigma_{\textbf{f}_*\textbf{f}_*} \\
\end{bmatrix}. \nonumber 
\end{align}
By a direct calculation, the predictive distribution is 
\begin{align}
p(\textbf{f}_*|\textbf{y})=\mathcal{N}(\Sigma_{\textbf{f}_*\textbf{f}}(\Sigma_{\textbf{f}\textbf{f}}+\sigma^2_{noise} I_{m \times m})^{-1}\textbf{y}, \Sigma_{\textbf{f}_*\textbf{f}_*}-\Sigma_{\textbf{f}_*\textbf{f}} (\Sigma_{\textbf{f}\textbf{f}}+\sigma^2_{noise} I_{m \times m})^{-1}\Sigma_{\textbf{f}\textbf{f}_*})\,. \nonumber 
\end{align}
We refer the readers to \cite{rasmussen2003gaussian} for more background on Gaussian processes.

\subsection{Graph Laplacian and graph based Gaussian process}

The GL is a fundamental tool in spectral graph theory \cite{Fan:1996}.  In this section, we first summarize the GL and then introduce the GL-GP, which uses the spectral structure of GL to define a covariance matrix $\Sigma$ to be used as described in the previous subsection.

Given a dataset $\mathcal{X}=\{x_1, \cdots, x_m, x_{m+1} \cdots, x_{m+n}\} \subset \mathbb{R}^D$, we first define a kernel 
\begin{equation}\label{KK matrix}
k_{\epsilon}(x,x')=\exp\Big(-\frac{\|x-x'\|^2_{\mathbb{R}^D}}{4\epsilon^2}\Big),
\end{equation}
where $\epsilon>0$ is a bandwidth parameter.  Although we can choose a more general kernel within our proposed methodology, we focus on this Gaussian-like choice for simplicity.  This kernel is used to define an $(m+n) \times (m+n)$  affinity matrix $W$ as
\begin{align}\label{W matrix}
W_{ij}&\,:=\frac{k_{\epsilon}(x_i,x_j)}{q_{\epsilon}(x_i) q_{\epsilon}(x_j)}\,,
\end{align}
where $i,j=1,\ldots,m+n$ and $q_{\epsilon}(x):=\sum_{i=1}^{m+n} k_{\epsilon}(x,x_i)$. We construct an $(m+n) \times (m+n)$ diagonal matrix $D$ so that its $i$-th diagonal entry is 
\begin{align}\label{D matrix}
D_{ii}=\sum_{j=1}^{m+n} W_{ij},
\end{align}
and define the row stochastic transition matrix as $A=D^{-1}W.$
Our main quantity of interest is the GL matrix, which is defined as 
\begin{equation}\label{L matrix}
L:=\frac{A-I}{\epsilon^2}\,.
\end{equation}

The affinity matrix $W$ in \eqref{W matrix} is symmetric. Hence $W$ can be considered as the affinity matrix of an undirected complete graph $G=(V, E, W)$, where $V=\{x_i\}_{i=1}^{m+n}$, $E$ consists of edges connecting any pair of points in $V$, and $L$ is the GL over the affinity graph $G$. Moreover, the GL can be viewed as an infinitesimal generator of a random walk on $G$  \cite{Fan:1996}. 

\begin{remark}In the graph theory literature,  $D-W$ is typically called the {\em unnormalized} GL and the matrix $I-A=I-D^{-1}W$ is called the {\em normalized} (or random walk) GL. We call the matrix $L$ the {\em kernel normalized} GL since the affinity is normalized by $q_{\epsilon}(x_i) q_{\epsilon}(x_j)$ in \eqref{W matrix}. 
\end{remark}

\begin{remark}\label{remark diffusion map}
The affinity matrix $W$ in \eqref{W matrix} is also considered in many kernel based-machine learning algorithms, e.g., diffusion maps \cite{coifman2006diffusion}. In \cite{coifman2006diffusion}, the term \eqref{W matrix} is called the {\em $\alpha$-normalization} with $\alpha=1$. The kernel $k_{\epsilon}(x_i,x_j)$ is normalized by $q_{\epsilon}$ to adjust for the non-uniform sampling density. 
\end{remark}

A basic spectral property of $L$ is summarized in the following proposition. 
\begin{proposition}\label{positivity of eigenvalue of -L}
Let $\mu$ be an eigenvalue of $-L$. Then, $\mu$ is real and $0 \leq \mu \leq \frac{1}{\epsilon^2}$. In particular, $0$ is the smallest eigenvalue of $-L$.
\end{proposition}
\begin{proof}
It is sufficient to show that the eigenvalues of $A$ are real and between $0$ and $1$ with $1$ the largest eigenvalue of $A$. Let $K_{ij}=k_{\epsilon}(x_i,x_j)$. Then, $K$ is a symmetric matrix generated by the Gaussian. By Bochner's theorem, $K$ is positive semidefinite. If $Q$ is the diagonal matrix with $Q_{ii}=\sum_{j=1}^{m+n} k_{\epsilon}(x_i,x_j)$, then $W=Q^{-1}WQ^{-1}$ is positive semidefinite.  Let $\tilde{A} = D^{-1/2}W D^{-1/2}$, which is a symmetric matrix. Hence, the eigenvalues of $\tilde{A}$ are real and $\tilde{A}$ is positive semidefinite. Since $A$ is similar to $\tilde{A}$ through $A=D^{-1/2}\tilde{A} D^{1/2}$, the eigenvalues of $A$ and $\tilde{A}$ are the same. Finally, since $A$ is row stochastic, the largest eigenvalue of $A$ is $1$ by the fact that the spectral norm is bounded by the $L^\infty$ norm of $A$. 
\qed
\end{proof}

Suppose the dataset  $\mathcal{X}=\{x_1, \cdots, x_m, x_{m+1} \cdots, x_{m+n}\}$ lies within a subset $S$ of $\mathbb{R}^D$ and we construct the GL matrix $L$ based on $\mathcal{X}$ as in \eqref{L matrix}. 
Suppose the eigenpairs of the GL over the affinity graph $G$ are 
$(-\mu_{i,m+n,\epsilon},v_{i,m+n,\epsilon})$.
Denote $\tilde{v}_{i,m+n,\epsilon}$ to be the eigenvector associated with the eigenvalue $\mu_{i,m+n,\epsilon}$ of $-L$ normalized in the $\ell^2$ norm, where $i=0,\ldots, m+n-1$. 
By Proposition \ref{positivity of eigenvalue of -L}, we order $\mu_{i,m+n,\epsilon}$ so that $0 = \mu_{0 ,m+n,\epsilon} < \mu_{1,m+n,\epsilon}\leq \ldots\leq \mu_{m+n-1 ,m+n,\epsilon}$, where $\mu_{0 ,m+n,\epsilon}<\mu_{1 ,m+n,\epsilon}$ since the graph is connected. 
Then, we define 
\begin{align}\label{GL-GP matrix}
\mathsf H_{\epsilon,K,t}=(m+n) \sum_{i=0}^{K-1} e^{-\mu_{i,m+n,\epsilon} t} \tilde{v}_{i,m+n,\epsilon} \tilde{v}^\top_{i,m+n,\epsilon}\in \mathbb{R}^{(m+n)\times (m+n)}\,,
\end{align}
to be the covariance matrix for  GP regression over $\mathcal{X}$.  

Since $\mathsf H_{\epsilon,K,t}$ is constructed from the GL, we refer to a GP with the covariance (\ref{GL-GP matrix}) as the GL-GP. We will show later that the GL-GP covariance matrix in \eqref{GL-GP matrix} is an approximation to the GL-GP covariance function on $S \times S$ and it is associated with the kernel of a compact self-adjoint integral operator over $L^2(S)$. The total dimension of the eigenspaces corresponding to the non-zero eigenvalues of the intergal operator is $K$. We will also show that the eigenvalue $e^{-\mu_{i,m+n,\epsilon} t}$ is an approximation to the $i$-th largest eigenvalue of the integral operator and $\sqrt{m+n}\tilde{v}_{i,m+n,\epsilon}$  is an approximation to the corresponding eigenfunctions normalized in the $L^2$ norm. We will discuss more details of formulation \eqref{GL-GP matrix} after we define the GL-GP covariance function on $S \times S$. 

For illustration, we consider a two balloons example in $\mathbb{R}^3$.  In this case, $S$ is a set with singularities consisting of $2$ dimensional spheres and $1$ dimensional line segments. We sample $2530$ points on $S$ as the dataset $\mathcal{X}$. Through this toy example, we discuss how the parameters $\epsilon$, $K$ and $t$ in the GL-GP covariance matrix $\mathsf H_{\epsilon,K,t}$ are related to the geometric structure of  set $S$. We also compare the GL-GP covariance and the GP with squared exponential covariance:
\begin{align}\label{normal GP}
\mathsf{C}(x, x')=A \exp\Big(-\frac{\|x-x'\|^2_{\mathbb{R}^D}}{\rho^2}\Big).
\end{align}

\begin{figure}[htb!]
\centering
\includegraphics[width=1 \textwidth]{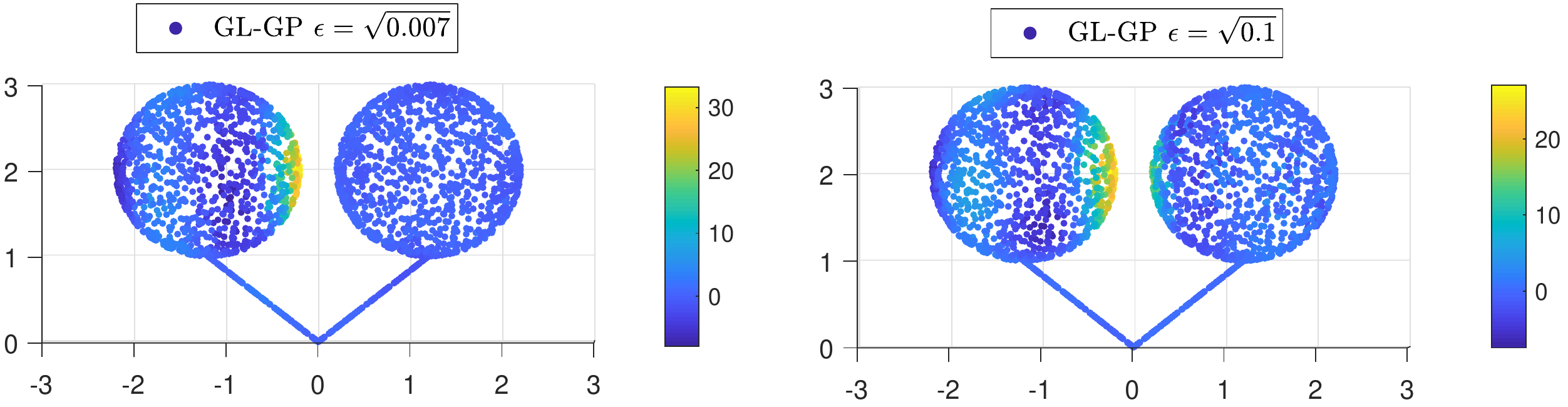}
\caption{When $K=38$, $t=0.01$ in $\mathsf H_{\epsilon,K,t}$,  we plot the row of the GL-GP covariance matrix corresponding to a point in the dataset over the $2530$ sampled points on $S \subset \mathbb{R}^3$ for $\epsilon=\sqrt{0.007}$ and $\epsilon=\sqrt{0.1}$ respectively.}\label{DBGP_differenteps}
\end{figure}

The parameter $\epsilon$ in the GL should be large enough so that there are sufficient numbers of points in an $\epsilon$-sized Euclidean neighborhood around $x$ to obtain information about the local geometry but not so large as to include points that are not close to $x$ in intrinsic distance within $S$.  
In the two balloons example, we focus on a point $x_i$ located on the right edge of the left balloon; this point is close in Euclidean distance to points on the left edge of the right balloon.  We plot the covariance $\mathsf H_{\epsilon,K,t}$ between this point and the other points in Figure \ref{DBGP_differenteps} letting $K=38$ and $t=0.01$.  In the left panel, $\epsilon= \sqrt{0.007}$, which is an appropriate size to define an intrinsic neighborhood on the set, while in the right panel $\epsilon=\sqrt{0.1}$, which bridges the gap and leads to inappropriate covariance across balloons.
The parameter $K$ controls fluctuations in the GL-GP covariance. When $K$ is larger, higher frequency oscillations are considered in constructing the GL-GP covariance. However, due to the factor $e^{-\mu_{i,m+n,\epsilon} t}$ in the covariance matrix, the amplitudes of those high frequency oscillations are relatively small. In Figure \ref{DBGP_differentK}, we plot the covariance relative to the same $x$ as in the previous figure but for different choices of $K$.  When $K=8$ the covariance decreases monotonically with increasing intrinsic distance from $x_i$, while for $K=38$ there are oscillations with small amplitudes contributing to non-monotonicity.  Finally, the diffusion time $t$ controls the rate of decrease in the covariance as the intrinsic distance increases.  Figure \ref{comparisonDBGPnGP} shows the impact of varying $t$ on the covariance relative to the impact of varying the bandwidth in the squared exponential covariance.


\begin{figure}[htb!]
\centering
\includegraphics[width=1 \textwidth]{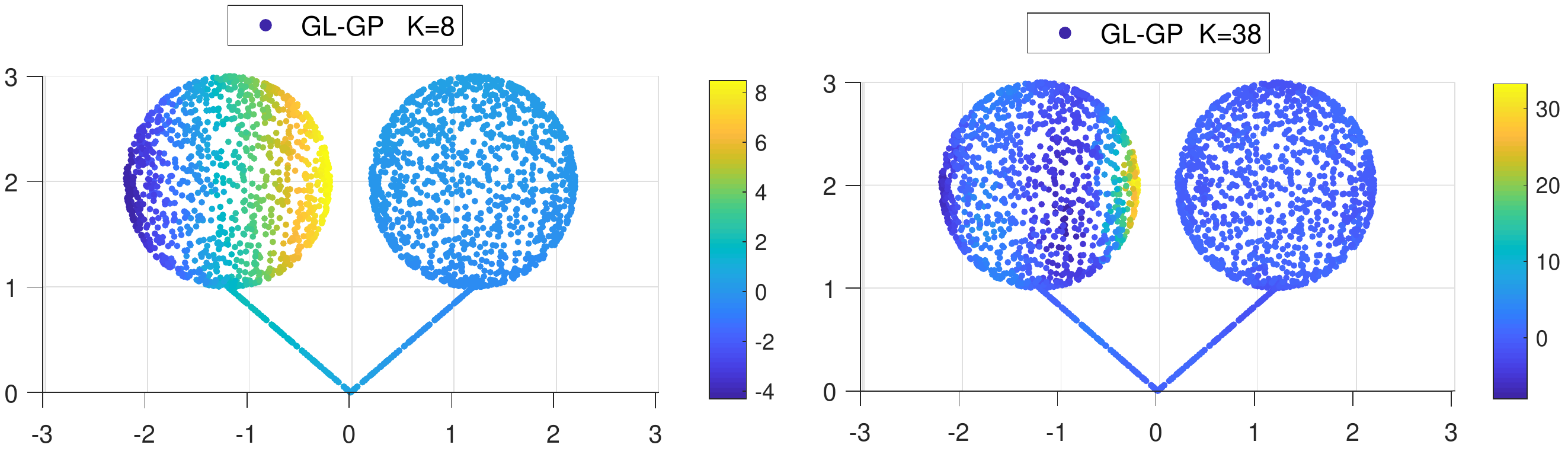}
\caption{When $\epsilon=\sqrt{0.007}$, $t=0.01$ in $\mathsf H_{\epsilon,K,t}$,  we plot the row of GL-GP covariance matrix corresponding to a point in the dataset over the $2530$ sampled points on $S \subset \mathbb{R}^3$ for $K=8$ and $K=38$ respectively.}\label{DBGP_differentK}
\end{figure}

\begin{figure}[htb!]
\centering
\includegraphics[width=1 \textwidth]{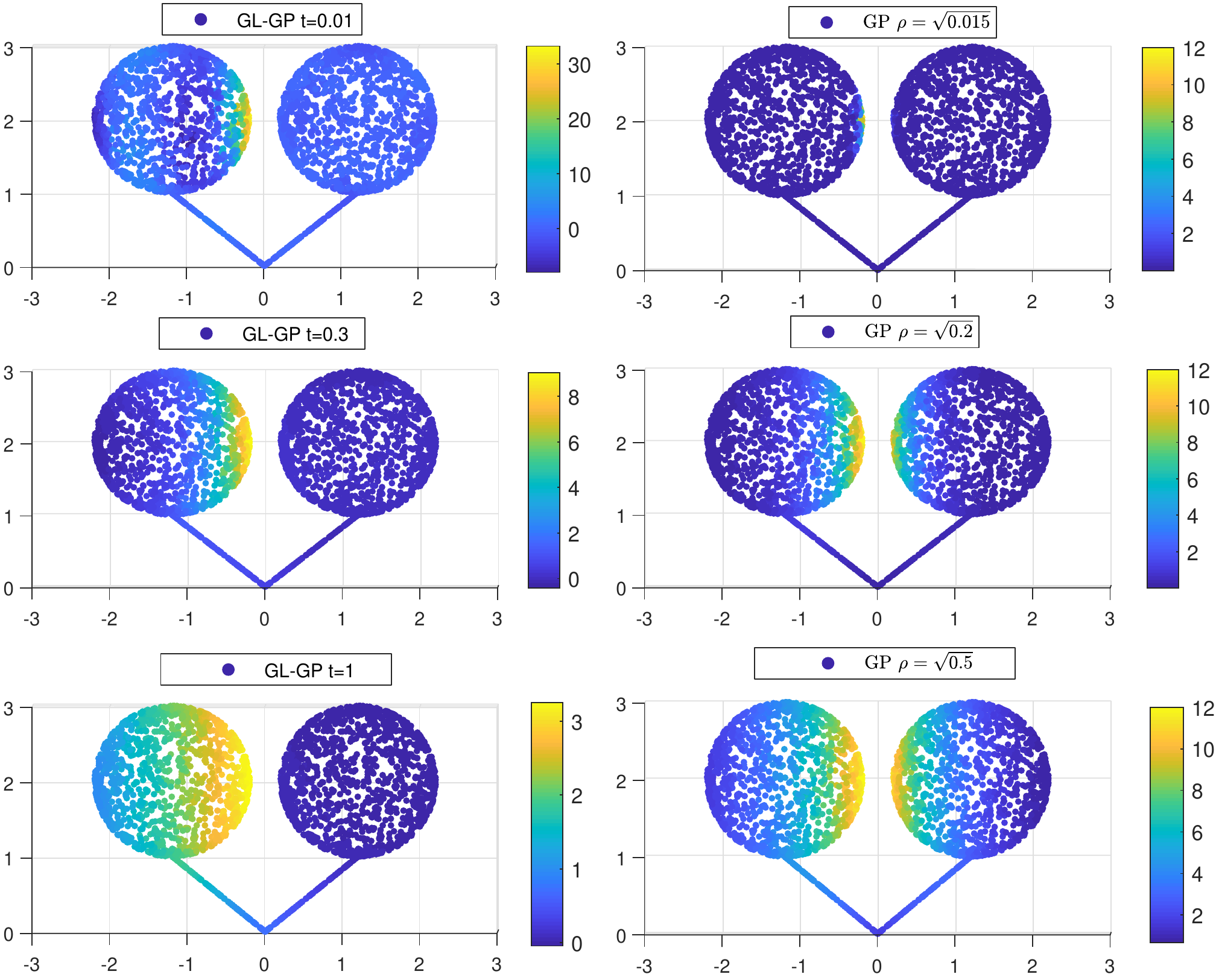}
\caption{When $\epsilon=\sqrt{0.007}$, $K=38$ in $\mathsf H_{\epsilon,K,t}$,  we plot a row of GL-GP covariance matrix corresponding to a point in the dataset over the $2530$ sampled points on $S \subset \mathbb{R}^3$ for $t=0.01$, $t=0.3$ and $t=1$ respectively in the three panels on the left. In comparison, we plot the same row of the sqaure exponential covariance matrix over the $2530$ sampled points on $S \subset \mathbb{R}^3$ induced by \eqref{normal GP} for $A=12$ and $\rho=\sqrt{0.015}$, $\rho=\sqrt{0.2}$ and $\rho=\sqrt{0.5}$ respectively in the three panels on the right. }\label{comparisonDBGPnGP}
\end{figure}

From the above discussion, the parameters $t$, $\epsilon$ and $K$ in the GL-GP covariance have an important impact on the performance.  Let 
$
\mathsf H_{\epsilon,K,t}=\begin{bmatrix}
\mathsf A & \mathsf B \\
\mathsf C & \mathsf D \\
\end{bmatrix},
$
where $\mathsf A$ is an $m \times m$ matrix. We propose to estimate these parameters by maximizing the log of the marginal likelihood, 
\begin{align}\label{marginal likelihood}
\log p(\textbf{y}|\epsilon, K, t, \sigma_{noise}) = -\textbf{y}^\top (\mathsf A+\sigma^2_{noise} I_{m \times m})^{-1}\textbf{y}-\log(\det(\mathsf A+\sigma^2_{noise} I_{m \times m}))-\frac{m}{2}\log(2\pi).
\end{align}
An identical strategy is used routinely in the literature for estimating GP covariance parameters \cite{rasmussen2003gaussian}. The parameters $\epsilon$ and $K$ are related to the GL over the dataset, while the parameters $t$ and $\sigma_{noise}$ are not.  To maximize \eqref{marginal likelihood}, we propose to alternate between a grid search for $\epsilon$ and $K$ and gradient descent for $t$ and $\sigma_{noise}$.

\subsection{Nystr\"{o}m extension of the GL-GP covariance matrix}\label{Nystrom subsection}
Suppose the dataset  $\mathcal{X}=\{x_1, \cdots, x_m, x_{m+1} \cdots, x_{m+n}\}\subset S\subset\mathbb{R}^D$, $L$ is the GL constructed from $\mathcal{X}$ following \eqref{L matrix}, and $(\mu_{i,m+n,\epsilon}, \tilde{v}_{i,m+n,\epsilon})$ is the $i$-th eigenpair of $-L$ with $\tilde{v}_{i,m+n,\epsilon}$ normalized in $\ell^2$.   Then, construct the GL-GP covariance matrix  $\mathsf H_{\epsilon,K,t}$ as in \eqref{GL-GP matrix}. Now, suppose we have $\ell$ additional samples $\{x_{m+n+1}, \cdots, x_{m+n+\ell}\}$ on $S$ and set $\mathcal{X}^*=\{x_1, \cdots, x_{m+n}, x_{m+n+1}, \cdots, x_{m+n+\ell}\}$. In this subsection, we propose a computationally efficient Nystr\"{o}m extension algorithm to find an extension of $\mathsf H_{\epsilon,K,t}$ over $\mathcal{X}^*$ without needing to rerun the whole GL-GP algorithm.

The idea is using the full knowledge of the kernel to construct an extension of the eigenvector $\tilde{v}_{i,m+n,\epsilon}$ to $\mathcal{X}^*$. Recall that $q_{\epsilon}(x):=\sum_{i=1}^{m+n} k_{\epsilon}(x,x_i)$. Define the extension matrix $E \in \mathbb{R}^{(m+n+\ell) \times (m+n)}$ by
\begin{align}\label{extension matrix}
E_{ij}:=\frac{ \frac{k_{\epsilon}(x_i,x_j)}{q_{\epsilon}(x_i) q_{\epsilon}(x_j)}}{\sum_{j=1}^{m+n} \frac{k_{\epsilon}(x_i,x_j)}{q_{\epsilon}(x_i)q_{\epsilon}(x_j)}}, 
\end{align}
where $x_i \in \mathcal{X}^*$ and $x_j \in \mathcal{X}$. We have the following immediate proposition.

\begin{proposition}\label{extension of the eigenvector}
For $x_j \in \mathcal{X}$, we have $\frac{1}{1-\epsilon^2\mu_{i,m+n,\epsilon}}E \tilde{v}_{i,m+n,\epsilon}(j)=\tilde{v}_{i,m+n,\epsilon}(j)$.
\end{proposition}
\begin{proof}
Observe that $A_{ij}=E_{ij}$ for $1 \leq i,j \leq m+n$. Since $\tilde{v}_{i,m+n,\epsilon}$ is the eigenvector of $A$ corresponding to the eigenvalue $1-\epsilon^2\mu_{i,m+n,\epsilon}$, for $1 \leq j \leq m+n$, we have 
$\frac{1}{1-\epsilon^2\mu_{i,m+n,\epsilon}}E \tilde{v}_{i,m+n,\epsilon}(j)=\frac{1}{1-\epsilon^2\mu_{i,m+n,\epsilon}}A \tilde{v}_{i,m+n,\epsilon}(j)=\tilde{v}_{i,m+n,\epsilon}(j)$. \qed
\end{proof}

This proposition says that $\frac{1}{1-\epsilon^2\mu_{i,m+n,\epsilon}}E \tilde{v}_{i,m+n,\epsilon}$ is an extension of the eigenvector $\tilde{v}_{i,m+n,\epsilon}$ from $\mathcal{X}$ to $\mathcal{X}^*$. With the parameters $\epsilon$, $K$, and $t$, the Nystrom extension gives
\begin{align}\label{DBGP covariance matrix extension}
\mathsf H^*_{\epsilon,K,t}:=E \left(\sum_{i=0}^{K-1} \frac{e^{-\mu_{i,m+n,\epsilon} t}}{(1-\epsilon^2\mu_{i,m+n,\epsilon})^2} \tilde{v}_{i,m+n,\epsilon} \tilde{v}^\top_{i,m+n,\epsilon}\right) E^\top\in \mathbb{R}^{(m+n+\ell)\times (m+n+\ell)},
\end{align}
which is an extension of the GL-GP covariance matrix $\mathsf H_{\epsilon,K,t}$ over $\mathcal{X}^*$. Based on the definition of the extension matrix $E$, for $1 \leq i,j \leq l$, $\mathsf H^*_{\epsilon,K,t}(m+n+i, m+n+j)$ only depends on  $\mathcal{X}$, $x_{m+n+i}$ and $x_{m+n+j}$ and not on the remaining samples.  Hence, it will not change when further samples are added.

We will justify the error in the Nystr\"{o}m extension of the GL-GP covariance structure and the error in the prediction after we introduce the GL-GP covariance function in Theorem \ref{Nystrom rate main theorem}. A simulation result of the Nystr\"{o}m extension is provided in Section \ref{simulation nystrom} of the appendices.

\begin{remark}
The Nystr\"{o}m extension can be used to induce a covariance function for the GL-GP.
 For any $x,y \in S$, let  $\mathcal{X}=\{x_1, \cdots, x_{m+n}\}$ and $\mathcal{X}^*=\{x_1, \cdots, x_{m+n}, x,y\}$. Construct the Nystr\"{o}m extension matrix $\mathsf H^*_{\epsilon,K,t}$ based on $\mathcal{X}^*$ and define $\mathsf{C}_{\epsilon,K}(x, y, t)$ to be the corresponding entry in the matrix. Such $\mathsf{C}_{\epsilon,K}(x, y, t)$ shares the property that the restriction over  $\mathcal{X}$ is equal to the covariance matrix.  However, this definition relies on the samples $\mathcal{X}$. Later, we will introduce a more natural way to define the covariance function that is independent of samples.
\end{remark}

\section{Applications}\label{simulations}
In this section, we apply our GL-GP approach in three different examples.  In all cases, we compare the GL-GP with the GP with the covariance  \eqref{normal GP}. The first case is the two balloons example. The second case is a spiral, which is a compact manifold with boundary. The third case is a complicated $2$ dimensional compact manifold with boundary coming from the 3D scan of a human's right hand.

\subsection{GL-GP on two balloons}
We have two unit spheres $S_1$ and $S_2$ centered at $(1.2, 0, 2)$ and $(-1.2, 0,2)$ respectively. We connect the south poles $(1.2, 0, 1)$ and $(-1.2, 0,1)$ to the origin by two line segments $L_1$ and $L_2$; $S=S_1 \cup S_2 \cup L_1 \cup L_2$. Globally, $S$ does not have a manifold structure, and the Hausdorff dimension of $S$ is $2$. 

We sample $30$ points from a uniform density on one of the spheres and sample $3$ points from a uniform on the attached line segment.
We then find the $33$ points symmetric to the sampled points on the other sphere and line segment. Those $66$ points are the labeled data points $\{x_{1}, \ldots, x_{66}\}$.  $2200$ unlabeled points $\{x_{67}, \ldots, x_{2266}\}$ are sampled in the same way with $2000$ points on the spheres and $200$ points on the line segment. The labels are sampled via \eqref{eq:base} with $\sigma_{noise}=1$ and $f(x) = 5 \times d_S(x,a)$, for $x \in S$, with $a=(-1.2,0,3)$ the north pole of $S_2$ and $d_S$ denoting the distance metric for the space $S$. The distance on $S$ between two points on the same sphere is the length of the shorter part of the great circle connecting those points. The distance between a point $x$ on a sphere and a point $x'$ on the segment attached to the sphere is the sum of the distance on $S$ between $x$ to the attaching point and the Euclidean distance between the attaching point to $x'$.  Hence, the formula for $d_S(x,a)$ is
\begin{align}
d_S(x,a)= \begin{cases}
\arccos \big([x-(-1.2, 0, 2)] \cdot (0,0,1)\big) &\text{if $x \in S_2$,}\\
\pi+\|x-(-1.2, 0, 1)\|_{\mathbb{R}^3} &\text{if $x \in L_2$,} \\
\pi+ \|(1.2, 0, 1)\|_{\mathbb{R}^3}+\|x\|_{\mathbb{R}^3}  &\text{if $x \in L_1$,} \\
\pi+2 \|(1.2, 0, 1)\|_{\mathbb{R}^3} + \arccos \big([x-(1.2, 0, 2)] \cdot (0,0,-1)\big)  &\text{if $x \in S_1$.}
\end{cases}
\end{align}
We plot the labels over $x_i$ for $i=1,\ldots,66$ and the true values in the top two panels in Figure \ref{two spheres comparison}.

We use both the GL-GP and the GP with square exponential covariance to predict the response values for the unlabeled data.  In both cases, covariance parameters are chosen by maximizing the marginal likelihood.  We calculate the root mean square error (RMSE) relative to the true value of the regression function at the unlabeled data points. Maximizing \eqref{marginal likelihood}, the parameters are $K=30$, $\epsilon =\sqrt{0.012}$, $t=0.33$ and $\sigma_{noise}=\sqrt{0.9}$, leading to an RMSE $0.721$. For the GP, we obtain $A=250$, $\rho=\sqrt{0.54}$, $\sigma_{noise}=\sqrt{0.9}$ and an RMSE $1.812$. Figure \ref{two spheres comparison} compares the prediction by the GL-GP and the GP over $x_i$, for $i=67, \cdots, 2266$. For better visualization, in Figure \ref{differencetwospheres}, we compare the difference between the prediction and the ground truth over $x_i$, for $i=67, \cdots, 2266$.  The GL-GP performs better over the regions indicated in the boxes and their symmetric parts on the other sphere. Since the square exponential covariance in the GP tries to smooth the predictive values between region $1$ and its symmetric part, the prediction in region $1$ is lower than the true value. 

\begin{figure}[htb!]
\centering
\includegraphics[width=.7 \columnwidth]{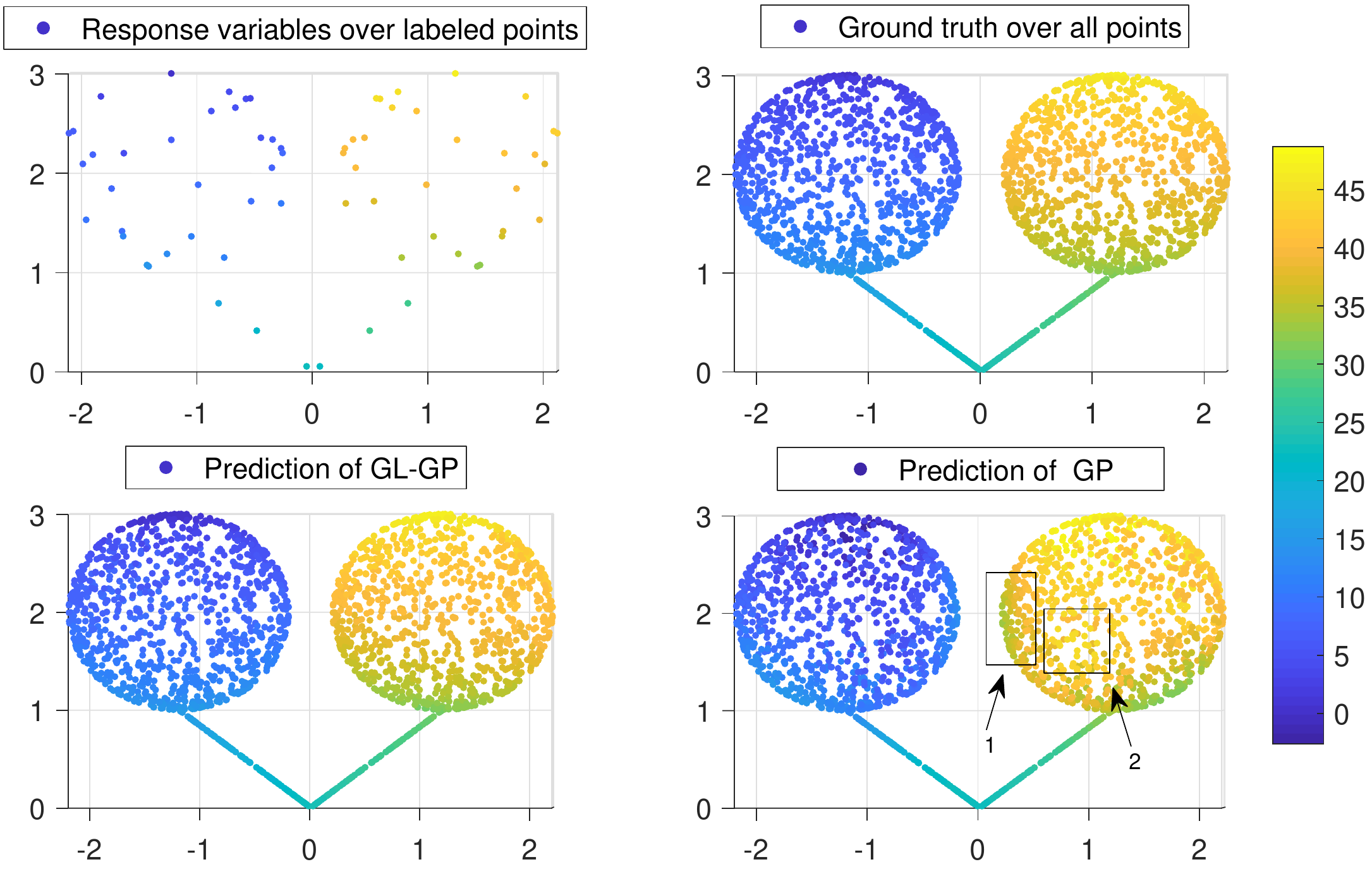}
\caption{The true regression function is $5$ times the distance of a point to the north pole of the sphere on the left. Let $\{x_1, \dots, x_{66}\}=$ labeled points and $\{x_{67}, \dots, x_{2266}\}=$ unlabeled points. Top left:  response variables over labeled points. Top right: ground truth over all points. Bottom left: prediction of GL-GP over unlabeled points with $RMSE=0.721$.  Bottom right:  prediction of (squared exponential Euclidean) GP over unlabeled points with $RMSE=1.812$. }\label{two spheres comparison}
\end{figure}

\begin{figure}[htb!]
\centering
\includegraphics[width=0.7 \columnwidth]{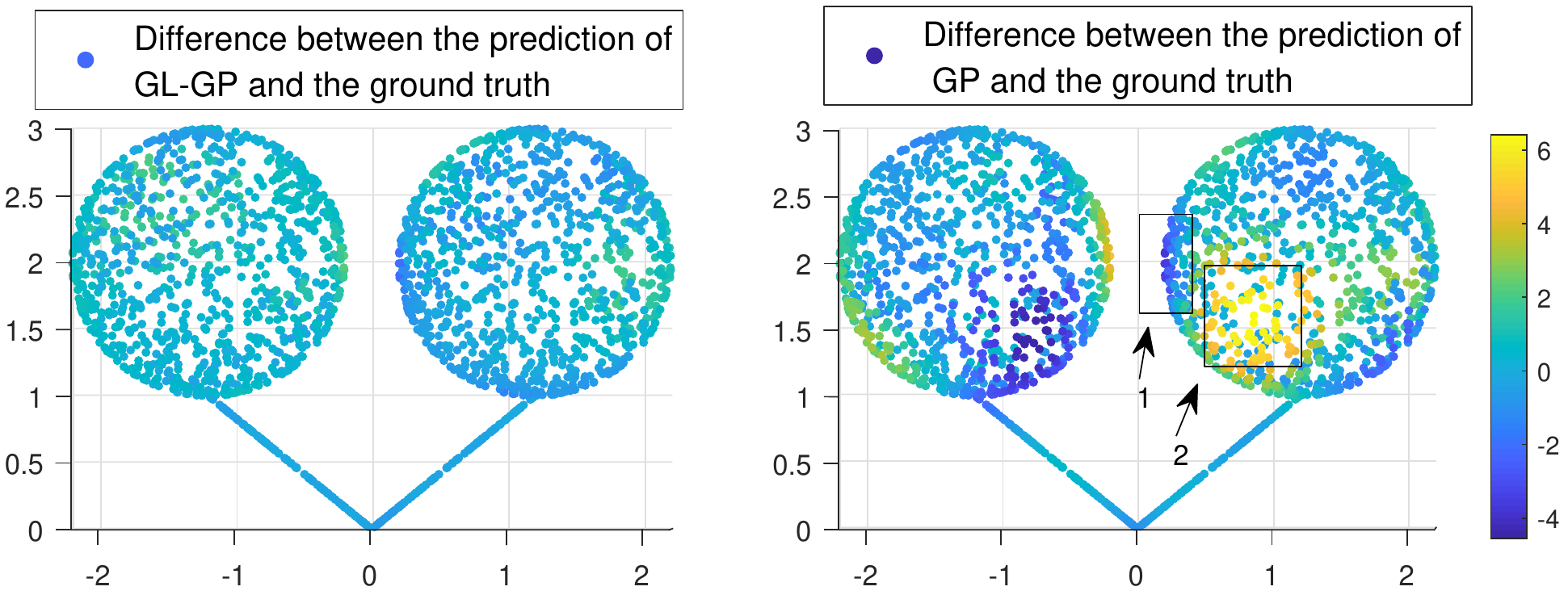}
\caption{Left: difference between prediction of GL-GP and ground truth over $\{x_i\}$, $i=67, \cdots, 2266$. Right: difference between prediction of GP and ground truth over $\{x_i\}$, for $i=67, \cdots, 2266$. }\label{differencetwospheres}
\end{figure}

\subsection{Spiral case}
We consider a spiral $M$ embedded in $\mathbb{R}^{2}$ parametrized by 
$$\gamma(\theta)=((\theta+4)^{0.7}\cos(\theta), (\theta+4)^{0.7}\sin(\theta)) \in \mathbb{R}^{2},\mbox{ where }\theta \in [0, 8\pi).$$ 
We sample $60$ labeled points $\{\theta_1, \ldots,\theta_{60}\}$ and $1500$ unlabeled points $\{\theta_{61}, \ldots, \theta_{1560}\}$ on $[0 ,8\pi)$ from the uniform density.  Let $x_i=\gamma(\theta_i)$ for $i=1,\ldots,1560$. The labels are sampled under \eqref{eq:base} with $\sigma_{noise}=1$ and 
$f(\gamma(\theta)) =3\sin\big(\frac{\theta}{10}\big) +3\cos\big(\frac{\theta}{2}\big)+4\sin\big(\frac{4\theta}{5}\big).$
We plot the labels over $x_i$ for $i=1, \ldots, 60$,  and the ground truth in the top two panels in Figure \ref{spiral comparison}.
Maximizing the marginal likelihoods, we obtain  $K=9$, $\epsilon =\sqrt{0.1}$, $t=0.02$ and $\sigma_{noise}=\sqrt{1.3}$ for the GL-GP, 
leading to an RMSE of $0.853$. For the GP with covariance  \eqref{normal GP}, we obtain $A=17$, $\rho=\sqrt{2.2}$, $\sigma_{noise}=\sqrt{0.7}$ and an RMSE of $1.967$. The bottom two panels in Figure \ref{spiral comparison} show the predictions of the two different approaches. For better visualization, in Figure \ref{spiral comparison theta}, we plot the the predictions over the parameter $\theta_i$. The GL-GP greatly improves predictive performance. As an example, we provide an analysis over region 1.  When we ignore the intrinsic geometry of the spiral, based on the response variables over the labeled data, there is a potential increase along the direction outward over region 1. Moreover, for the points that are close to region 1 in  Euclidean distance on the inner neighbor arc, the labels are negative and large in magnitude. Hence, the prediction by the GP over region 1 is negative.

\begin{figure}[htb!]
\centering
\includegraphics[width=.7\columnwidth]{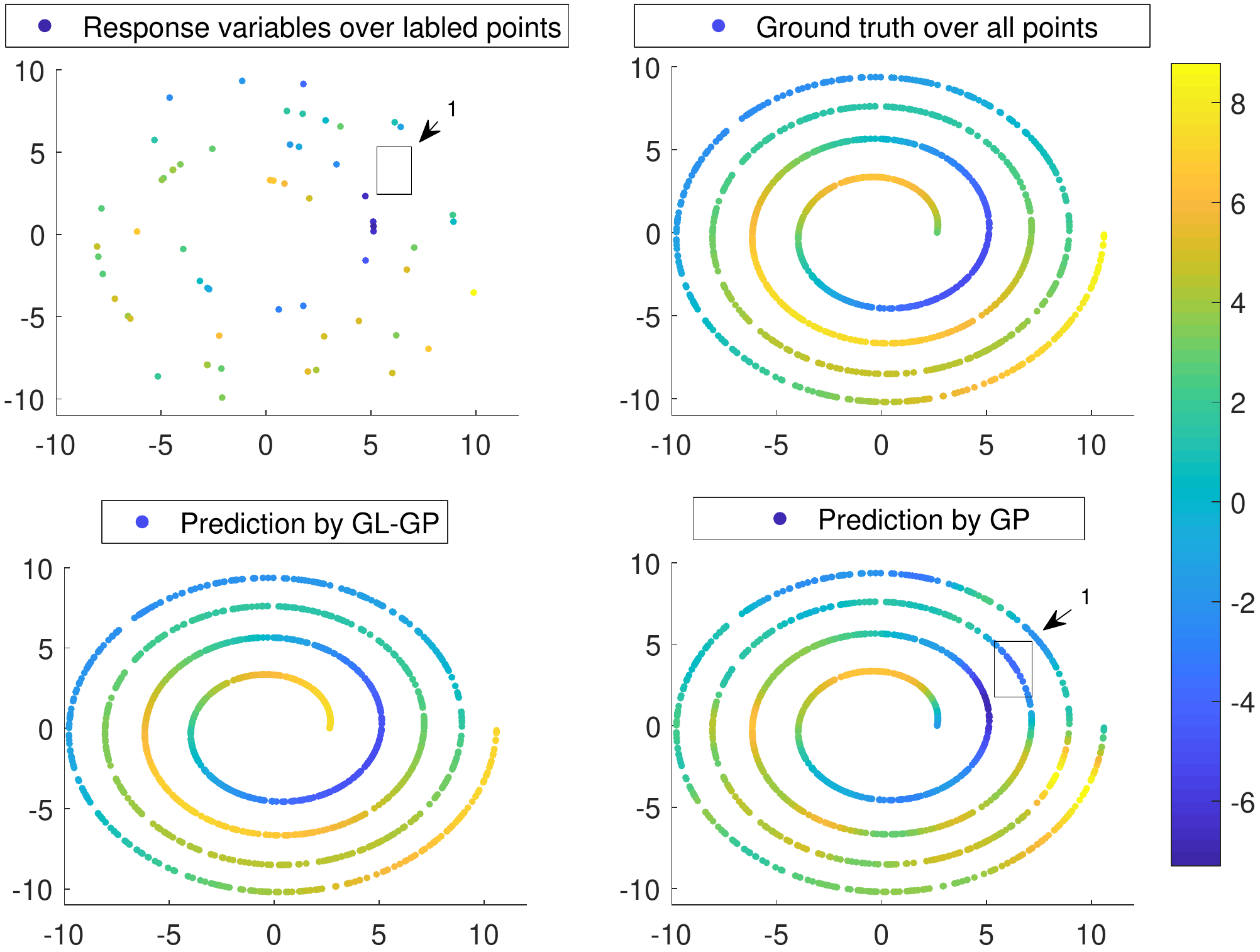}
\caption{ Let $\{x_1, \dots, x_{60}\}=$ labeled points and $\{x_{61}, \dots, x_{1560}\}=$ unlabeled points. Top left: response variables over labeled points. Top right: ground truth over all points. Bottom left: prediction of GL-GP over unlabeled points with $RMSE=0.853$.  Bottom right: prediction of (squared exponential Euclidean) GP over the unlabeled points with $RMSE=1.967$.}. \label{spiral comparison}
\end{figure}

\begin{figure}[htb!]
\centering
\includegraphics[width=1 \columnwidth]{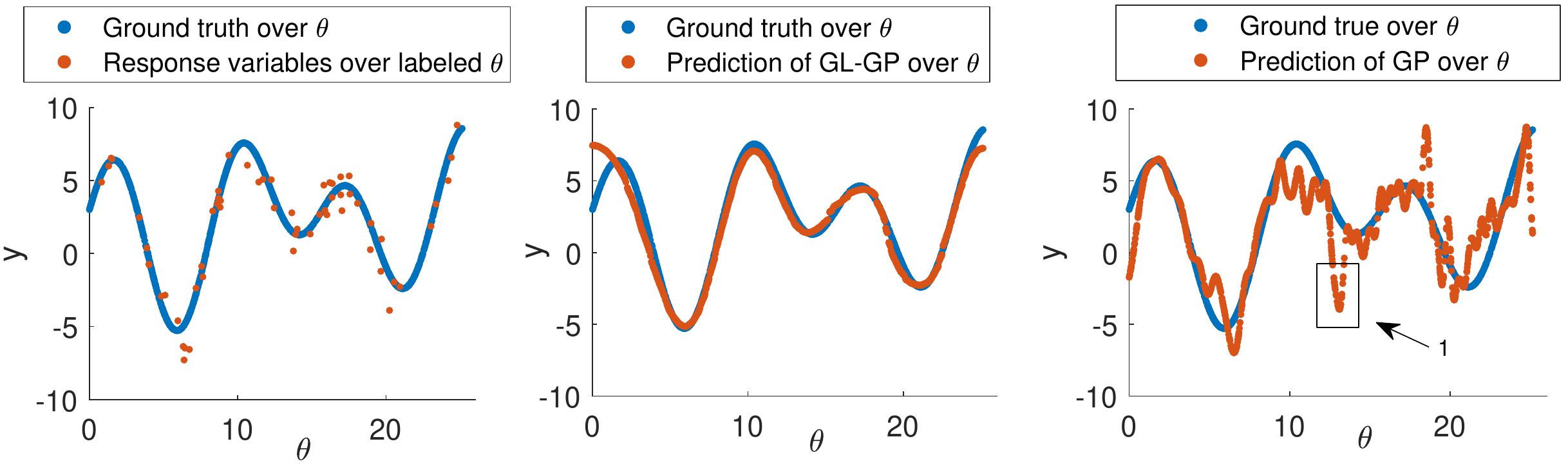}
\caption{Let $\{\theta_1, \dots, \theta_{60}\}=$ labeled inputs and $\{\theta_{61}, \dots, \theta_{1560}\}=$ unlabeled inputs. Left: red points are response variables for labeled inputs, and 
blue points are ground truth for all $\{\theta_i\}$, $i=1, \cdots, 1560$.  Middle: red points are prediction of GL-GP over $\{\theta_i\}$ for $i=61, \cdots, 1560$. Right: red points are prediction of GP over $\{\theta_i\}$ for $i=61, \cdots, 1560$.}. \label{spiral comparison theta}
\end{figure}

\subsection{Temperature distribution on the right hand of a Raynauld's disease patient}
The 3D surface scan of a human hand is a 2 dimensional compact manifold with boundary isometrically embedded in $\mathbb{R}^3$.  
The dataset in this example is a discrete version of a scan of a right hand in \cite{romero2017embodied} and consists of $1950$  points. There are $50$ labeled points $\{x_i\}_{i=1}^{50}$ among which $25$ points are on the fingers and $25$ points on the palm. The remaining points $\{x_i\}_{i=51}^{1950}$ are unlabeled. Let $\mathcal{X}=\{x_i\}_{i=1}^{1950}$. For $x_i,x_j \in \mathcal{X}$, $W_{ij}=\exp\Big(-\frac{\|x_i-x_j\|^2_{\mathbb{R}^3}}{0.012}\Big) \in \mathbb{R}^{1950 \times 1950}$. Define a diagonal matrix $D_{ii}=\sum_{j=1}^{1950} W_{ij}$. Let $L_{un}=D-W$ denote the unnormalized GL, which is different from the GL defined in the previous section.  Let $v_3$, $v_{6}$, $v_{13}$ be the eigenvectors normalized in $\ell^2$ corresponding to the $3$rd, $6$th and $13$th smallest eigenvalues of $L_{un}$.  The labels are sampled under equation~\eqref{eq:base} with $\sigma_{noise}=0.02$ and 
\begin{align}\label{equation of ground truth example 3}
f(x_i) =8.118 v_3(i)+2.46 v_6(i)+0.738 v^+_{13}(i) +32.2\,,
\end{align}  where $v^+_{13}(i)=v_{13}(i)$ if $v_{13}(i)>0$ and $v^+_{13}(i)=0$ if $v_{13}(i) \leq 0$.
The function $f$ is used to simulate the temperature distribution over the samples. We plot the labels over $x_i$ for $i=1, \ldots, 50$,  and the ground truth viewed from the back of the right hand in Figure \ref{3Dhandgroundtruth}. 

Maximizing the marginal likelihood, we obtain  $K=11$, $\epsilon =\sqrt{0.0015}$, $t=0.002$ and $\sigma_{noise}=\sqrt{0.004}$ for the GL-GP, leading to an RMSE $0.062$. For the GP with squared exponential covariance, the optimal parameters are $A=405$, $\rho=\sqrt{0.0145}$, $\sigma_{noise}=\sqrt{0.012}$, which leads to an RMSE $0.169$. Figure \ref{3Dhandcomparison} compares the prediction by GL-GP and GP over $\{x_i\}$, for $i=51, \cdots, 1950$. For better visualization, in Figure \ref{3Dhanddifference}, we compare the difference between the prediction and ground truth over $\{x_i\}$, for $i=51, \cdots, 1950$.  The GL-GP performs better over the regions indicated in the boxes. In particular, in region 3, since the squared exponential covariance in the GP tries to smooth the predictive values between the ring finger and the little finger, the prediction in region $3$ is lower than the true value. Based on the labeled information, there is a potential decrease along the index finger from the bottom part to the tip.  However, there is a faster decrease from the middle finger to the index finger. The  Euclidean distance term in the covariance of the GP reflects such a fast decrease; as a result, the predictive values of the GP are higher over region 2 and lower on the tip of the index finger. Similar explanation can be applied to regions 1 and 4.

\begin{figure}[htb!]
\centering
\includegraphics[width=1 \columnwidth]{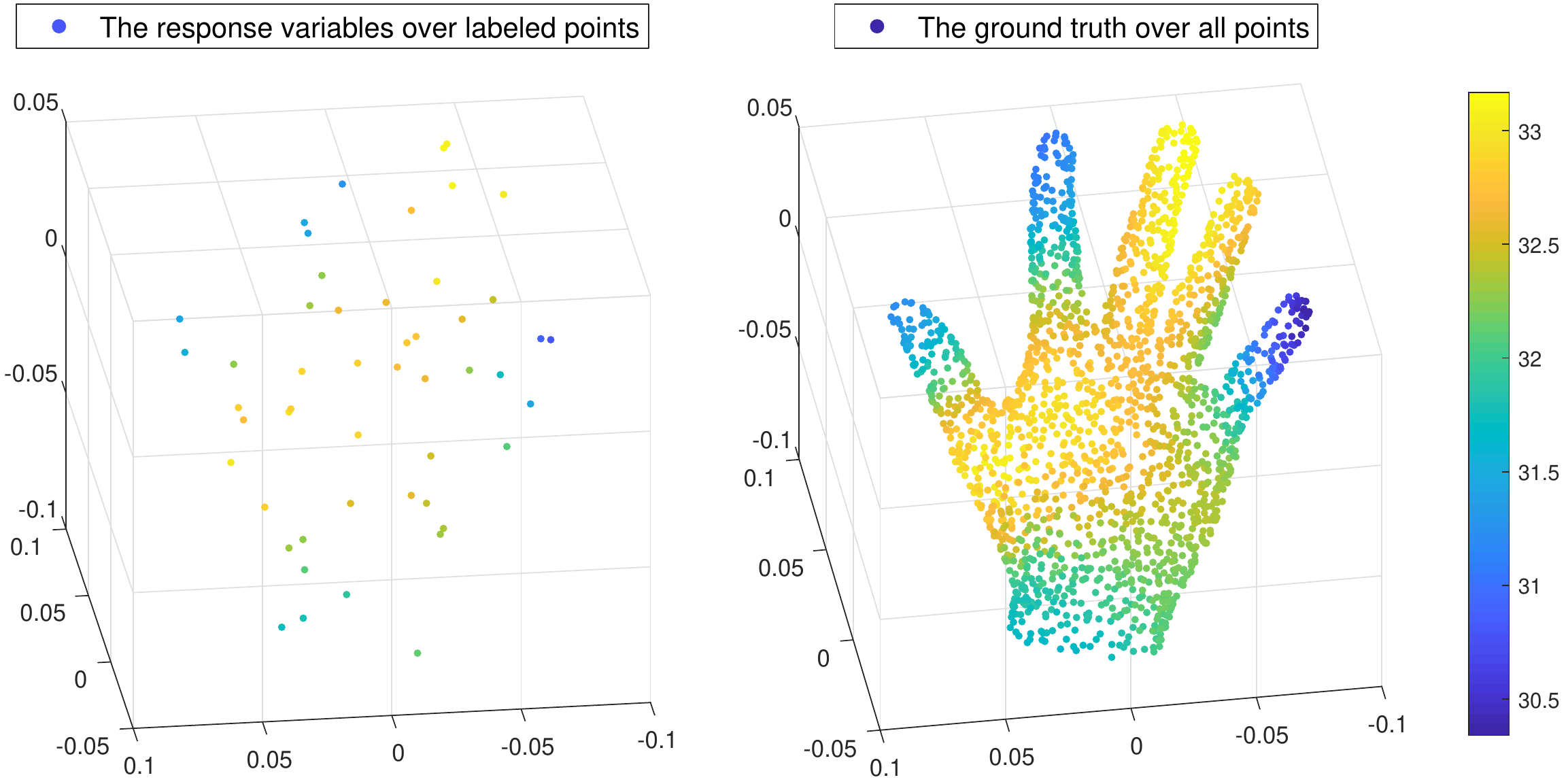}
\caption{Let $\{x_1, \dots, x_{50}\}$ be the labels points and $\{x_{51}, \dots, x_{1950}\}$ be the unlabeled points. Left: The plot of the response variables over the labeled points $\{x_i\},i=1,\cdots, 50$. Right: The plot of the ground truth  over all points $\{x_i\},i=1,\cdots, 1950$. }\label{3Dhandgroundtruth}
\end{figure}

\begin{figure}[htb!]
\centering
\includegraphics[width=1 \columnwidth]{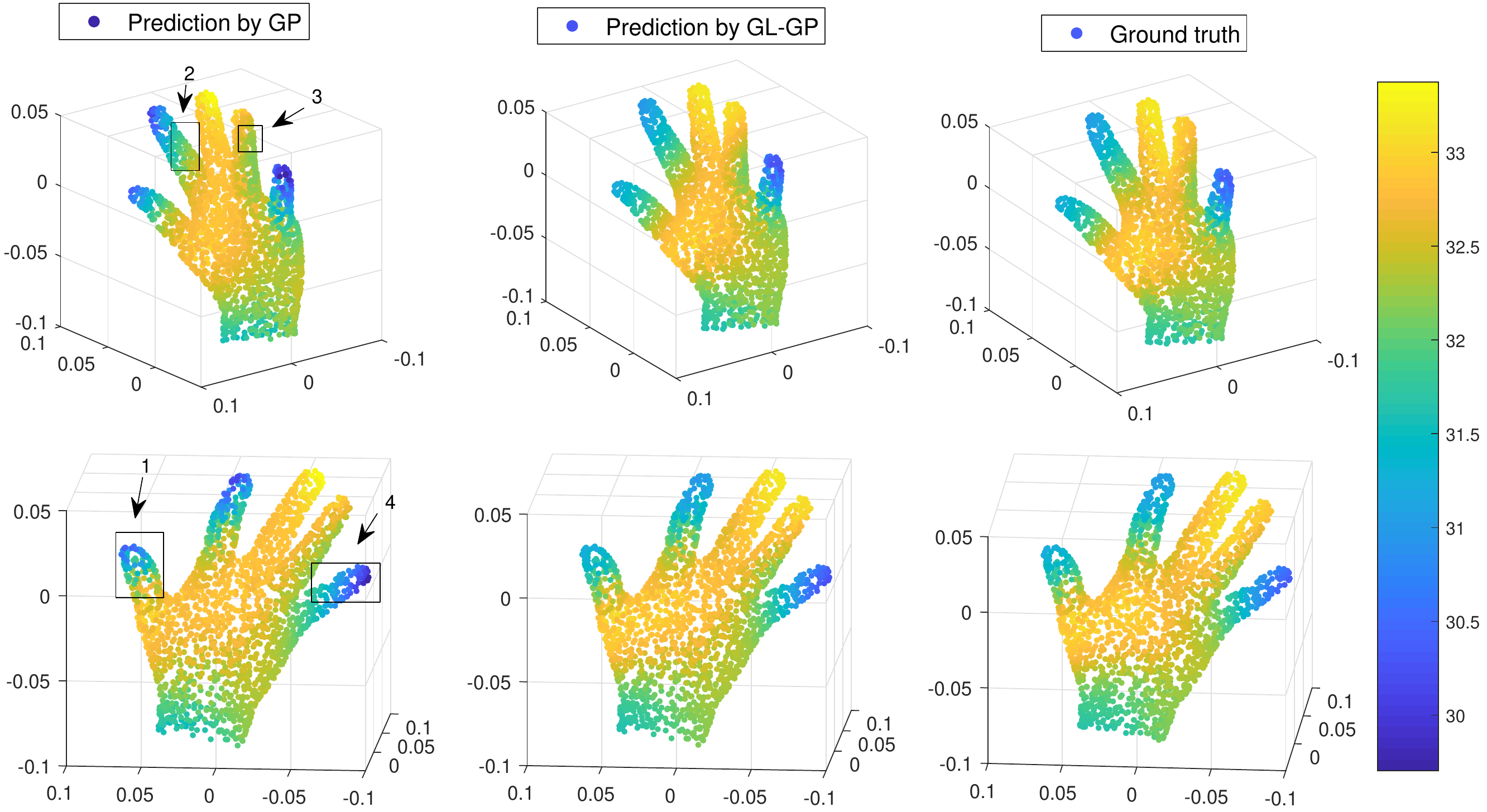}
\caption{Left 2 panels : 2 different views of the prediction by the GP with square exponential covariance over over $\{x_i\}$, for $i=51, \cdots, 1950$ with $RSME=0.169$. Middle 2 panels: 2 different views of the prediction by the GL-GP over $\{x_i\}$, for $i=51, \cdots, 1950$ with $RSME=0.062$.  Right 2 panels: 2 different views of the ground truth over $\{x_i\}$, for $i=51, \cdots, 1950$.}\label{3Dhandcomparison}
\end{figure}

\begin{figure}[htb!]
\centering
\includegraphics[width=1 \columnwidth]{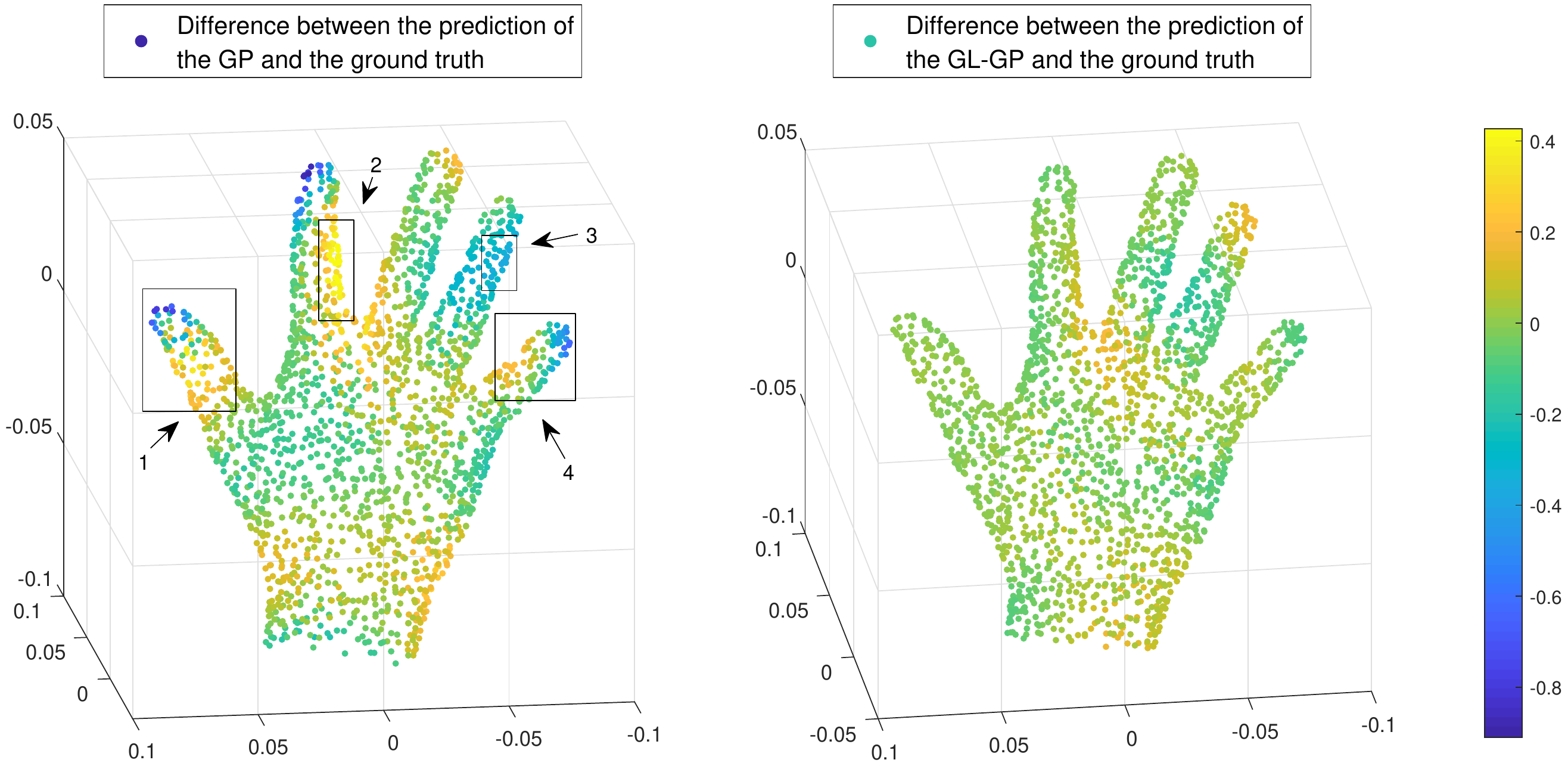}
\caption{Left: The difference between the prediction of the GP and the ground truth over $\{x_i\}$, for $i=51, \cdots, 1950$. Right: The difference between the prediction of the GL-GP and the ground truth over $\{x_i\}$, for $i=51, \cdots, 1950.$}\label{3Dhanddifference}
\end{figure}

\section{Theory of Graph based Gaussian Process}

\subsection{Graph based Gaussian process on subsets of $\mathbb{R}^D$}
In this section, we introduce theory of GLs when the dataset is sampled from a subset $S$ of $\mathbb{R}^D$. We are going to associate the GL with an integral operator on the set $S$. Then, we apply spectral theory to the integral operator to define the corresponding GL-GP covariance function. We first make the following assumption on the subset $S$ and the probability measure on $S$. 
\begin{assumption}\label{assumption on S}
Suppose $S$ is a compact connected subset of $\mathbb{R}^D$. Suppose $(S, \mathcal{F}, \mathsf{P})$ is a probability space, where $\mathsf{P}$ is a  probability measure defined over the Borel sigma algebra $\mathcal{F}$ on $S$. Suppose $\mathcal{X}=\{x_1, \cdots, x_m, x_{m+1} \cdots, x_{m+n}\}$ are $m+n$ i.i.d. samples based on the measure $\mathsf{P}$.
\end{assumption}

\begin{remark}
There are natural measures on $S$ induced by metrics in the ambient Euclidean space $\mathbb{R}^D$, for example, the $d$ dimensional Hausdorff measure. However, we do not require $\mathsf{P}$ to be absolutely continuous with respect to any of these measures.  
\end{remark}

Based on Assumption \ref{assumption on S}, we define an integral operator associated with the GL. 

\begin{definition}
For any function $f(x,x')$ on $S \times S$, we define the operator $P$  with respect to the probability measure $\mathsf{P}$  as follows,
$Pf (x):=\int_{S }f(x, x')d\mathsf P(x').$ 
For the kernel $ k_{\epsilon}(x,x')$, we define 
$d_{\epsilon}(x):= Pk_{\epsilon}(x)$ and 
$Q_{\epsilon}(x,x'):= \frac{k_{\epsilon}(x,x')}{d_{\epsilon}(x)d_{\epsilon}(x')}.$ 
For $f(x)$ on $S$, we define the operator:
\begin{align}
T_{\epsilon}f(x):= \frac{PQ_{\epsilon}f (x)}{PQ_{\epsilon}(x)}=  \frac{\int_{S }Q_{\epsilon}(x,x')f(x')d\mathsf P(x')}{\int_{S }Q_{\epsilon}(x,x')d\mathsf P(x')}. \nonumber
\end{align}
At last, we define $\mathsf{I}_\epsilon:=\frac{I-T_{\epsilon}}{\epsilon^2}.$
\end{definition}

The operator $T_{\epsilon}$ has the following important property. The property motivates our covariance function for the Gaussian process.  
 Hence, we state it as a theorem here. The proof is postponed to Section \ref{Proof of theorem Property of T epsilon on S} of the appendices.

\begin{theorem}\label{Property of T epsilon on S}
Under Assumption \ref{assumption on S}, for any fixed $\epsilon$,  $T_{\epsilon}$ is a linear, compact, self-adjoint operator from $L^2(S, \mathsf{P})$ to $L^2(S, \mathsf{P})$. Moreover, suppose $A$ is the matrix  in \eqref{L matrix} constructed from $\mathcal{X}$, then for any $f \in L^2(S, \mathsf{P})$, we have the following convergence result: for any  $x_k \in \mathcal{X}$, 
\begin{align}
\lim_{m+n \rightarrow \infty}\sum_{i=1}^{m+n} A_{k, i} f(x_i) =T_{\epsilon} f(x_k), \quad  \mbox{a.s.} \nonumber
\end{align}
\end{theorem}

We discuss the theoretical motivation for constructing a covariance function incorporating geometric information about the set $S$.
By Theorem \ref{Property of T epsilon on S}, for a fixed $\epsilon$, since $T_{\epsilon}$ is compact and self-adjoint, the eigenvalues of $\mathsf{I}_\epsilon$ satisfy $\lambda_{0,\epsilon} \leq \lambda_{1,\epsilon} \leq \cdots$. Suppose $(\lambda_{i,\epsilon}, \phi_{i,\epsilon})$ is the $i$-th eigenpair of $\mathsf{I}_\epsilon$ and each $\phi_{i,\epsilon}$ is normalized in $L^2(S, \mathsf{P})$. Then $\{\phi_{i,\epsilon}\}$ form an orthonormal basis of $L^2(S, \mathsf{P})$. For any function $f \in L^2(S, \mathsf{P})$, $f$ can be expressed as a unique linear combination of $\{\phi_{i,\epsilon}\}$, i.e. $f= \sum_{i=0}^\infty a_i \phi_{i,\epsilon}$, a.e. Note that $\sum_{i=1}^\infty a_i^2 = \|f\|^2_{L^2(S, \mathsf{P})} < \infty$. Hence, $a_i \rightarrow 0$ as $i \rightarrow \infty$ and a finite linear combination $\sum_{i=0}^l a_i \phi_{i,\epsilon}$ is an   
approximation of $f$ in an $ L^2(S, \mathsf{P})$ sense. Based on the above observation, if our unknown regression function is $f \in L^2(S, \mathsf{P})$, it is reasonable to propose a GP covariance function on $S$ by using finitely many eigenpairs of $T_{\epsilon}$. More precisely, fixing $K  \in \mathbb{N} $ and $t>0$, we define 
\begin{align}\label{DBGP covariance function}
\mathsf{C}_{\epsilon,K}(x, x', t)=\sum_{i=0}^{K-1}e^{-\lambda_{i,\epsilon}t} \phi_{i,\epsilon}(x)\phi_{i,\epsilon}(x'),
\end{align}
for $x, y \in S$. This covariance function involves three parameters, $t$, $\epsilon$ and $K$. The parameter $t$ can be regarded as controlling the bandwidth, and hence the choice of $t$ should depend on the regularity of $f$.  The parameter $\epsilon$ is used in the construction of the operator $T_{\epsilon}$ over $S$; $\epsilon$ should depend on the regularity of $S$. In contrast,  $K$ should depend on both $S$ and the regularity of $f$.

We give an intuitive justification to show that the GL-GP covariance matrix in \eqref{GL-GP matrix} approximates the GL-GP covariance function over $\mathcal{X}$. The pointwise convergence of the matrix $A$ to $T_{\epsilon}$ in Theorem \ref{Property of T epsilon on S} suggests that when $m+n$ is large enough, the $i$th smallest eigenvalue of $-L$ approximates the $i$th smallest eigenvalue of $\mathsf{I}_\epsilon$. The corresponding eigenvector approximates $\phi_{i,\epsilon}$ over $\mathcal{X}$ if the eigenvector is properly normalized.  For any continuous $f \in L^2(S, \mathsf{P})$, by the law of large numbers,
$$\lim_{m+n \rightarrow  \infty}\frac{1}{m+n}\sum_{i=1}^{m+n}f^2(x_i)=\|f\|^2_{L^2(S, \mathsf{P})}  \quad  \mbox{a.s.}$$
This suggests that if $v_i$ is an eigenvector corresponding to the $i$th smallest eigenvalue of $-L$ and $v_i$ normalized in $\ell^2$, then $(\sqrt{m+n}) v_i$ approximates $\phi_{i,\epsilon}$ over $\mathcal{X}$. Hence, $H^{K}_{\epsilon,t}$ approximates $\mathsf{C}_{\epsilon,K}(x, x', t)$ over $\mathcal{X}$. To rigorously justify this intuition, we need to impose more structure on $S$.  We will discuss this in the next section.

\subsection{Graph based Gaussian process on manifolds}
In this section, we consider the special case in which $S$ has a manifold structure and develop additional theory supporting the GL-GP in this case.
Specifically,  we make the following assumption.

\begin{assumption}\label{assumption DM}
Let $M$ be a $d$-dimensional smooth, closed and connected Riemannian manifold isometrically embedded in $\mathbb{R}^D$ through $\iota:M\to \mathbb{R}^D$. Let $S=\iota(M)$. Suppose $(M, \mathcal{F}, \mathsf{P})$ is a probability space, where $\mathsf{P}$ is a  probability measure defined over the Borel sigma algebra $\mathcal{F}$ on $M$. We assume $\mathsf{P}$ is absolutely continuous with respect to the volume measure on $M$, i.e. $d\mathsf{P}=\mathsf{p} dV$ by the Radon-Nikodym theorem, where $\mathsf{p}$ is the probability density function on $M$ and $dV$ is the volume form. We further assume $\mathsf{p}$ is smooth and is bounded from below by $\mathsf{p}_{m}>0$. $\mathcal{X}=\{x_1 \cdots, x_m, x_{m+1}, \cdots, x_{m+n}\}$ are i.i.d. sampled from $\mathsf{P}$. 
\end{assumption}

Since $\iota $ is an isometry, the function $f$ in model \eqref{eq:base} can be equivalently viewed as $f: M \rightarrow \mathbb{R}$ instead of $f:S \rightarrow \mathbb{R}$.  We start by modifying the definition of GL and the operator $T_{\epsilon}$ for datasets and functions on $M$.  
Based on Assumption \ref{assumption DM}, we define the kernel 
$$k_{\epsilon}(x,x')=\exp\Big(-\frac{\|\iota(x)-\iota(x')\|^2_{\mathbb{R}^D}}{4\epsilon^2}\Big),$$ for $x, x' \in M$. We use the kernel $k_{\epsilon}(x,x')$ to construct an $(m+n) \times (m+n)$ affinity matrix over $\mathcal{X}$ as in \eqref{W matrix}. Let $L$ be the graph Laplacian of the affinty graph $G$ with vertices $V=\mathcal{X}$. 

In the manifold case, quantities in Definition 1 have the following expansions. For any function $f(x,x')$ on $M \times M$, we have  
\begin{align}
Pf (x):=\int_{M}f(x, x') d\mathsf P(x')= \int_{M}f(x, x') \mathsf{p}(x') dV(x') \,. \nonumber
\end{align}
For $f(x)$ on $M$, we have 
\begin{align}
T_{\epsilon}f(x):= \frac{PQ_{\epsilon}f (x)}{PQ_{\epsilon}(x)}=  \frac{\int_{M}Q_{\epsilon}(x,x')f(x')\mathsf{p}(x') dV(x')}{\int_{M}Q_{\epsilon}(x,x')\mathsf{p}(x') dV(x')}. \nonumber
\end{align}
By the same argument as in Theorem \ref{Property of T epsilon on S}, $T_{\epsilon}$ is a linear, compact, self-adjoint operator from $L^2(M)$ to $L^2(M)$. Since the manifold is compact and smooth, $Q_{\epsilon}(x,x')$ is smooth. Hence, $T_{\epsilon}f \in C^\infty (M)$. We have the following proposition about the eigenvalues of $\mathsf{I}_\epsilon=\frac{I-T_{\epsilon}}{\epsilon^2}$.

\begin{proposition}\label{eigenvalue lower bound of I epsilon}
Under Assumption \ref{assumption DM}, the  eigenvalues of $\mathsf{I}_\epsilon=\frac{I-T_{\epsilon}}{\epsilon^2}$ satisfy $0 \leq \lambda_{0,\epsilon} \leq \lambda_{1,\epsilon} \leq \cdots$.
\end{proposition}
\begin{proof}
Note that $T_{\epsilon}$ is a linear, compact, and self-adjoint operator from $L^2(M)$ to $L^2(M)$. It is sufficient to prove the eigenvalues of $T_{\epsilon}$ are bounded by $1$. Suppose $(\lambda, f)$ is an eigenpair of $T_{\epsilon}$. Since the manifold is compact and smooth, $Q_{\epsilon}(x,x')$ is smooth. Hence, $f$ is smooth and $\|f\|_\infty < \infty$.  Since $Q_{\epsilon}(x,x')>0$,  $T_{\epsilon}f=\lambda f$ implies that
$$|\lambda| = \frac{\|T_{\epsilon}f\|_\infty}{\|f\|_\infty} \leq \max_{x \in M}   \frac{\int_{M}Q_{\epsilon}(x,x') (f(x')/\|f\|_\infty)\mathsf{p}(x') dV(x')}{\int_{M}Q_{\epsilon}(x,x')\mathsf{p}(x') dV(x')} \leq 1.$$ 
\qed
\end{proof}

Suppose $(\lambda_{i,\epsilon}, \phi_{i,\epsilon})$ is the $i$th eigenpair of $\mathsf{I}_\epsilon$ and each $\phi_{i,\epsilon}$ is normalized in $L^2(M)$. Then $\{\phi_{i,\epsilon}\}$ form an orthonormal basis of $L^2(M)$. The covariance function $\mathsf{C}_{\epsilon,K}(x, x', t)$ under the manifold setup is defined in \eqref{DBGP covariance function} for $x,x' \in M$, but is constructed from the GL in a slightly different manner. When $S$ is a general subset of $\mathbb{R}^D$, we do not impose any other measures on $S$ except the probability measure. However, when $M$ is a Riemannian manifold, there is natural measure, namely the volume measure, induced by the Riemannian metric of the manifold. Hence,  the eigenfunctions of $\mathsf{I}_\epsilon$ should be normalized in the $L^2(M)$ norm rather than $L^2(M, \mathsf{P})$ norm; proper normalization ensures that the eigenvectors of the $GL$ can approximate the eigenfunctions of $\mathsf{I}_\epsilon$. Hence, we introduce the following definition.

\begin{definition}Under Assumption \ref{assumption DM}, suppose $\tilde{v}$ is an eigenvector of $L$ which is normalized in the $\ell^2$ norm. Let 
$\mathbb{N}(i)=|B^{\mathbb{R}^D}_{\epsilon}(\iota(x_i)) \cap \{\iota(x_1), \cdots, \iota(x_{m+n})\}|$, the number of points in an $\epsilon$ ball in $\mathbb{R}^D$ around $x_i$. Then, we define the $\ell^2$ norm of $\tilde{v}$ with respect to the inverse estimated probability density $1/\hat{\mathsf{p}}$ as: 
\begin{align}
\|\tilde{v}\|_{\ell^2(1/\hat{\mathsf{p}})}:=\sqrt{\frac{|S^{d-1}|\epsilon^d}{d} \sum_{i=1}^{m+n} \frac{\tilde{v}^2(i)}{\mathbb{N}(i)}}, \nonumber 
\end{align}
where $|S^{d-1}|$ is the volume of the sphere $S^{d-1}$.
\end{definition}
In the above definition, we approximate $\mathsf{p}$ by using a simple $0-1$ kernel with bandwidth $\epsilon$; refer to \cite{2020spectral} for a detailed discussion. 

Denote $\mu_{i,m+n,\epsilon}$ to be the $i$th eigenvalue of $-L$ with the associated eigenvector $\tilde{v}_{i,m+n,\epsilon}$ normalized in the $\ell^2$ norm, where $i=0,\ldots,m+n-1$. We order $\mu_{i,m+n,\epsilon}$ so that $\mu_{0,m+n,\epsilon}\leq \mu_{1,m+n,\epsilon}\leq \ldots\leq \mu_{m+n-1, m+n,\epsilon}$. If we define
\begin{align} \label{normalized V i n epsilon}
v_{i,m+n,\epsilon}=\frac{\tilde{v}_{i,m+n,\epsilon}}{\|\tilde{v}_{i,m+n,\epsilon}\|_{\ell^2(1/\hat{\mathsf{p}})}},
\end{align}
then $v_{i,n,\epsilon}$ can be regarded as a discretization of some function that is normalized in the $L^2(M)$ norm.  Fixing $K \in \mathbb{N}$ and $t>0$, we define
\begin{align}\label{DBGP covariance matrix on manifold}
\tilde{\mathsf H}_{\epsilon,K,t}=\sum_{i=0}^{K-1} e^{-\mu_{i,m+n,\epsilon} t} v_{i,m+n,\epsilon} v^\top_{i,m+n,\epsilon}\in \mathbb{R}^{(m+n)\times (m+n)}
\end{align}
to be the covariance matrix for Gaussian process regression over $\mathcal{X}$ on the manifold $M$. When the probability density is uniform, \eqref{GL-GP matrix} and  \eqref{DBGP covariance matrix on manifold} are equivalent. Hence, the general GL-GP algorithm in \eqref{GL-GP matrix} can be viewed as a simplification of  \eqref{DBGP covariance matrix on manifold}. The GL-GP algorithm on manifolds is described in Section \ref{DBGP manifold algorithm summary} of the appendices.

\subsection{Convergence of $\tilde{\mathsf H}_{\epsilon,K,t}$ and its  Nystr\"{o}m extension under the manifold setup}

First, we show convergence of the covariance matrix $\tilde{\mathsf H}_{\epsilon,K,t}$ to the covariance function $\mathsf C^K_{\epsilon}(x, x', t)$ over the dataset $\mathcal{X}=\{x_i\}_{i=1}^{m+n}$ under the manifold setup.  More precisely, we provide the convergence rate of entry $i,j$ of the matrix $\tilde{\mathsf H}_{\epsilon,K,t}$ to $\mathsf{C}_{\epsilon,K}(x_i, x_j, t)$ as $m+n \rightarrow \infty$. To state our main theorem, we recall the Laplace Beltrami operator $\Delta$ of the manifold $M$.  Here, $\Delta$ is an intrinsic differential operator on the manifold generalizing the notion of the second order derivative on an interval, with the eigenpairs reflecting the geometric and topological structure of the manifold.  Let $\sigma(-\Delta)=\{\lambda_i\}_{i=0}^\infty$ be the spectrum of $-\Delta$. By the standard elliptic theory, we have $0=\lambda_0 < \lambda_1\leq \lambda_2 \leq  \cdots$ and each eigenvalue has finite multiplicity. Denote by $\phi_i$ the eigenfunction normalized in $L^2(M)$ corresponding to $\lambda_i$; that is, for each $i\in \mathbb{N}$, we have $\Delta \phi_i =-\lambda_i \phi_i$. It is well known that $\{\phi_i\}$ forms an orthonormal basis of $L^2(M)$. The main theorem of this section is as follows. The proof of the theorem is in Section \ref{proof of convergence from H K epsilon t to C  K epsilon t} of the appendices.

\begin{theorem}\label{GL-GP convergence rate main}
Under Assumption \ref{assumption DM}, let $\lambda_{i}$ be the $i$th eigenvalue of $-\Delta$.  Fixing $K\in \mathbb{N}$, we let $\mathsf \Gamma_K:=\min_{1 \leq i \leq K}\textup{dist}(\lambda_i, \sigma(-\Delta)\setminus \{\lambda_i\})$.  If $\epsilon$ is small enough,
\begin{equation}\label{relation epsilon K main theorem}
\epsilon \leq \mathcal{K}_1 \min \left(\left(\frac{\min(\mathsf\Gamma_K,1)}{\mathcal{K}_2+\lambda_K^{d/2+5}}\right)^2,\, \frac{1}{(\mathcal{K}_3+\lambda_K^{(5d+7)/4})^2}\right)\,,
\end{equation} 
and $m+n$ is sufficiently large so that $(\frac{\log(m+n)}{m+n})^{\frac{1}{4d+13}} \leq \epsilon$, then for any $t$ less than the diameter of $M$, with probability greater than $1-(m+n)^{-2}$,  
\begin{align}
\max_{x_i, x_j \in \mathcal{X}}| \tilde{\mathsf H}_{\epsilon,K,t}(i,j)-\mathsf{C}_{\epsilon,K}(x_i, x_j, t)|  \leq \mathcal{K}_4  K^2 \epsilon^{1/2}; \nonumber
\end{align}
$\mathcal{K}_1$ and $\mathcal{K}_2, \mathcal{K}_3>1$  are constants depending on $d$, $\mathsf p_m$, the $C^2$ norm of $\mathsf p$, and the volume, injectivity radius, sectional curvature and second fundamental form of the manifold. $\mathcal{K}_4$ depends on  $d$, $\mathsf p_m$,  the $C^2$ norm of $\mathsf p$, and the diameter,  volume and Ricci curvature of $M$.
\end{theorem}

Note that $\tilde{\mathsf H}_{\epsilon,K,t}$ is composed of the eigenpairs of $-L$, while $\mathsf{C}_{\epsilon,K}$ is composed of the eigenpairs of $\mathsf{I}_\epsilon$.  Therefore, at first glance, the convergence rate of $\tilde{\mathsf H}_{\epsilon,K,t}$ to  $\mathsf{C}_{\epsilon,K}$ seems to be irrelevant to the Laplace-Beltrami operator. However, to control the convergence rate, we need to bound the eigengaps, the magnitude of the eigenvalues and the $L^\infty$ norm of the eigenfunctions of $\mathsf{I}_\epsilon$, and those terms can be controlled by the eigenvalues of $-\Delta$. The relationship between $\mathsf{I}_\epsilon$ to $-\Delta$, particularly the spectral convergence when $\epsilon\to 0$, is detailed in Lemma \ref{T epsilon and Delta} in Section \ref{proof of convergence from H K epsilon t to C  K epsilon t} of the appendices. On the other hand, the eigenvalues of $-\Delta$ are completely determined by the geometry of the manifold. Hence, bounding the convergence rate by the eigenvalues of
$-\Delta$ shows that the error between $\tilde{\mathsf H}_{\epsilon,K,t}$ and $\mathsf{C}_{\epsilon,K}$ is determined by the geometry of the manifold.

The GL-GP covariance matrix $\tilde{\mathsf H}_{\epsilon,K,t}$ can also be used to recover the heat kernel of the manifold $M$. The heat kernel on a manifold $M$ is the fundamental solution to the heat equation with an appropriate boundary condition. It describes the diffusion process of the heat flow out of a point on the manifold along the intrinsic distance. Hence, the result that the GL-GP matrix approximates the heat kernel implies that the covariance structure of GL-GP also varies with respect to the intrinsic distance of the manifold as the bandwidth $t$ in $\mathsf H_{\epsilon,K,t}$ changes. We refer the readers to Theorem 3 in \cite{2020spectral} for the convergence rate from  $\mathsf H_{\epsilon,K,t}$ to the heat kernel.

\begin{remark}
Suppose we consider a fixed manifold $M$. As the sample size $m+n$ increases, the relation between $\epsilon$ and $m+n$ and the lower bound on the p.d.f implies that the sampling points become increasingly dense in $M$. Hence, Theorem \ref{GL-GP convergence rate main} matches with the fixed-domain asymptotics regime considered in spatial statistics. Let $|\textup{Rm}(M)|$ denote the norm of the curvature tensor of $M$. Consider the set of all manifolds $\mathcal{M}:=\{M:\, \textup{dim}(M)=d,\, \textup{diam}(M)<D,\, \textup{vol}(M)>v,\, |\textup{Rm}(M)|\leq \tau\}$, where $d\in \mathbb{N}$, $D>0$, $v>0$, $\tau\geq 0$. Suppose we have a sequence of manifolds $\{M_i\}_{i=1}^\infty \subset \mathcal{M}$, and sample $i=m+n$ points on each $M_i$. Then, Theorem \ref{GL-GP convergence rate main} still holds for each $M_i$, when $i$ is sufficiently large. A future direction is removing the diameter upper bound on the manifold. If  Theorem \ref{GL-GP convergence rate main} holds for manifolds without this bound, then the theorem holds in the mixed-domain asymptotics regime \cite{chang2017mixed} in spatial statistics. 
\end{remark}

If we treat $K$, $\lambda_K$ and the eigengaps as constants and focus on the relation between $\epsilon$ and $m+n$, then we have the following corollary by applying the results in \cite{2020spectral} and the same proof of the above theorem. The corollary says that if we treat $K$, $\lambda_K$ and the eigengaps as constants, then the convergence rate of $\tilde{\mathsf H}_{\epsilon,K,t}(i,j)$ to $\mathsf{C}_{\epsilon,K}(x_i, x_j, t)$ is $O(\big(\frac{\log (m+n)}{m+n}\big)^{\frac{1}{2d+6}})$.

\begin{corollary}\label{convergence from H K epsilon t to C  K epsilon t} Under Assumption \ref{assumption DM}, let $\lambda_{i}$ be the $i$th eigenvalue of $-\Delta$.  Fixing $K\in \mathbb{N}$, we let $\mathsf \Gamma_K:=\min_{1 \leq i \leq K}\textup{dist}(\lambda_i, \sigma(-\Delta)\setminus \{\lambda_i\})$.  If $\epsilon$ is small enough and $m+n$ is large enough so that $ \big(\frac{\log (m+n)}{m+n}\big)^{\frac{1}{4d+12}} \leq \epsilon$, then for any $t$ less than the diameter of $M$, with probability greater than $1-(m+n)^{-2}$,  
\begin{align}
\max_{x_i, x_j \in \mathcal{X}}| \tilde{\mathsf H}_{\epsilon,K,t}(i,j)-\mathsf{C}_{\epsilon,K}(x_i, x_j, t)|  \leq  \mathcal{K} \big(\frac{\log (m+n)}{m+n}\big)^{\frac{1}{2d+6}}, \nonumber
\end{align}
where  $\mathcal{K}$ depends on $K$, $\mathsf \Gamma_K$, $\lambda_K$, $d$, $\mathsf p_m$, the $C^2$ norm of $\mathsf p$ and the diameter, volume, injectivity radius, curvature and  second fundamental form of the manifold.
\end{corollary} 

\begin{remark}
The convergence rates in the above theorem and corollary are not optimal. We expect that the optimal convergence rate is better than what we have reported. Finding the optimal convergence rate will be explored in our future work. Even if the optimal convergence rate, and hence the asymptotic relation between $\epsilon$ and $n+m$ is known, this relationship does not provide a practical approach for choosing $\epsilon$ due to the unknown constant.
\end{remark}

Next, we come back to the Nystrom extension. Suppose we have $\ell$ additional i.i.d. samples, $\{x_{m+n+1}, \cdots, x_{m+n+\ell}\}$, from density $\mathsf{P}$. Let $\mathcal{X}^*=\{x_1, \cdots, x_{m+n}, x_{m+n+1}, \cdots, x_{m+n+\ell}\}$. Let $E \in \mathbb{R}^{(m+n+\ell) \times (m+n)}$ be the extension matrix defined in \eqref{extension matrix}. Similar to \eqref{DBGP covariance matrix extension}, the Nystr\"{o}m extension of the GL-GP covariance matrix over $\mathcal{X}^*$ on the manifold is defined as 
\begin{align}\label{DBGP covariance matrix extension on manifold}
\tilde{\mathsf H}^*_{\epsilon,K,t}=E \left(\sum_{i=0}^{K-1} \frac{e^{-\mu_{i,m+n,\epsilon} t}}{(1-\epsilon^2\mu_{i,m+n,\epsilon})^2} v_{i,m+n,\epsilon} v^\top_{i,m+n,\epsilon}\right) E^\top\in \mathbb{R}^{(m+n+\ell)\times (m+n+\ell)}\,.
\end{align}
The following theorem shows that $\tilde{\mathsf H}^*_{\epsilon,K,t}$  is an approximation to $\mathsf{C}_{\epsilon,K}(x, x', t)$ over $\mathcal{X}^*$.
\begin{theorem}\label{Nystrom rate main theorem}
Under Assumption \ref{assumption DM}, suppose we sample $\ell$ additional times i.i.d. from density $\mathsf{P}$ to obtain 
 $\{x_{m+n+1}, \cdots, x_{m+n+\ell}\}$. Let $\mathcal{X}^*=\{x_1, \cdots, x_{m+n}, x_{m+n+1}, \cdots, x_{m+n+\ell}\}$.  
\begin{enumerate}
\item Let $\tilde{\mathsf H}_{\epsilon,K,t}$ and $\tilde{\mathsf H}^*_{\epsilon,K,t}$ be defined in \eqref{DBGP covariance matrix on manifold} and \eqref{DBGP covariance matrix extension on manifold} respectively. Then $\tilde{\mathsf H}_{\epsilon,K,t}(i,j)=\tilde{\mathsf H}^*_{\epsilon,K,t}(i,j)$, for $1 \leq i,j \leq m+n$.
\item Let $\lambda_{i}$ be the $i$th eigenvalue of $-\Delta$.  Fixing $K\in \mathbb{N}$, we let $\mathsf \Gamma_K:=\min_{1 \leq i \leq K}\textup{dist}(\lambda_i, \sigma(-\Delta)\setminus \{\lambda_i\})$.  If $\epsilon$ is small enough and
$\epsilon$ satisfies \eqref{relation epsilon K main theorem}, and $m+n$ satisfies  $(\frac{\log(m+n)}{m+n})^{\frac{1}{4d+13}} \leq \epsilon$, then for any $t$ less than the diameter of $M$, with probability greater than $1-(m+n)^{-2}$,  
\begin{align}
\max_{x_i, x_j \in \mathcal{X}^*}| \tilde{\mathsf H}^*_{\epsilon,K,t}(i,j)-\mathsf{C}_{\epsilon,K}(x_i, x_j, t)|  \leq \mathcal{K}_4   K^2 \epsilon^{1/2}, \nonumber
\end{align}
where $\mathcal{K}_4$  is the same constant defined in Theorem \ref{GL-GP convergence rate main}.
\end{enumerate}
\end{theorem}

The proof is similar to the proof of Theorem \ref{GL-GP convergence rate main} and is sketched in Section \ref{proof of convergence from H K epsilon t to C  K epsilon t} of the appendices.  
Part (a) says that $\tilde{\mathsf H}^*_{\epsilon,K,t}$  is an extension of $\tilde{\mathsf H}_{\epsilon,K,t}$  to $\mathcal{X}^*$. In the trivial case when there is no additional sample points, i.e. $\ell=0$, we have $\tilde{\mathsf H}^*_{\epsilon,K,t}=\tilde{\mathsf H}_{\epsilon,K,t}$. Although we have more samples to be used to construct a covariance matrix, the eigenvectors we use are an extension of the eigenvectors of the GL constructed from $\mathcal{X}$ while the eigenvalues remain the same. Hence, part (b) says the accuracy of the approximation of $\mathsf{C}_{\epsilon,K}(x, x', t)$ by $\tilde{\mathsf H}^*_{\epsilon,K,t}$ over $\mathcal{X}^*$ is on the same level as the approximation of $\mathsf{C}_{\epsilon,K}(x', x', t)$ by $\tilde{\mathsf H}_{\epsilon,K,t}$ over $\mathcal{X}$. If we treat $K$ and $\lambda_K$ as constants and focus on the relation between $\epsilon$ and $m+n$, then by using the same argument as Corollary \ref{convergence from H K epsilon t to C  K epsilon t}, the error between $ \tilde{\mathsf H}^*_{\epsilon,K,t}(i,j)$ and $\mathsf{C}_{\epsilon,K}(x_i, x_j, t)$ is $O(\big(\frac{\log (m+n)}{m+n}\big)^{\frac{1}{2d+6}})$.

\subsection{The predictive error of the covariance matrix}
Suppose $\tilde{\mathsf H}^{K}_{\epsilon,t}$  is the GL-GP covariance matrix by the GL over $\{x_1 \cdots, x_m, x_{m+1}, \cdots, x_{m+n}\}$.   If we rewrite $\tilde{\mathsf H}^{K}_{\epsilon,t}$ as
$
\begin{bmatrix}
H_1 & H_2 \\
H_3 & H_4 \\
\end{bmatrix},
$
where $H_1$ is an $m \times m$ matrix, then the prediction at $\{x_{m+1}, \cdots, x_{m+n}\}$ using the  GL-GP covariance matrix is $\textbf{f}_*:=H_3(H_1+\sigma^2_{noise} I)^{-1}\textbf{y}$. We define an $(m+n)\times(m+n)$ matrix $\Sigma$ so that $\Sigma_{ij}=\mathsf{C}^K_{\epsilon}(x_i, x_j, t)$. Let 
$\Sigma= \begin{bmatrix}
\Sigma_{11} & \Sigma_{12}  \\
\Sigma_{21} & \Sigma_{22}
\end{bmatrix},
$
where $\Sigma_{11}$ is an $m \times m$ symmetric matrix. The predictive at $\{x_{m+1}, \cdots, x_{m+n}\}$ using the {\em exact} GL-GP covariance function is then
$\textbf{f}_{\texttt{GLGP}}:=\Sigma_{21}(\Sigma_{11}+\sigma^2_{noise} I)^{-1}\textbf{y}.$ 
The following theorem describes the difference between the predictions of the GP under the matrix $\tilde{\mathsf H}^{K}_{\epsilon,t}$ and the exact GL-GP covariance function.  
The proof of the theorem is in Section \ref{difference in prediction} of the appendices.

\begin{theorem}\label{difference theorem}
Suppose $\|\tilde{\mathsf H}^{K}_{\epsilon,t}-\Sigma\|_{\max} \leq \delta<1$. Then we have
\begin{align}
\|\textbf{f}_{*}-\textbf{f}_{\texttt{GLGP}}\|_{\max} \leq \frac{ m  \|\textbf{y}\|_{\max}\big(( c+3\lambda) \|\Sigma_{21}\|_{\max}+ \lambda \big)}{\lambda(\lambda+c \delta)} \delta \,, \nonumber 
\end{align}
where $c$ is a constant depending on $\frac{1}{\delta}(\tilde{\mathsf H}^{K}_{\epsilon,t}-\Sigma)$ and $\lambda$ is the smallest eigenvalue of  $\Sigma_{11}+\sigma^2_{noise} I$. 
\end{theorem}

 If $\Sigma$ is a discretization of the covariance function over $\{x_1, \cdots, x_{m+n}, x_{m+n+1}, \cdots, x_{m+n+\ell}\}$ and we divide $\Sigma$ so that $\Sigma_{11} \in \mathbb{R}^{m \times m}$ and $\Sigma_{21} \in \mathbb{R}^{(n+\ell) \times m}$, then the above theorem also holds for the Nystr\"{o}m extension of $\tilde{\mathsf H}_{\epsilon,K,t}$. The result of Theorem \ref{difference theorem} can be combined with Theorem \ref{GL-GP convergence rate main} or Theorem \ref{Nystrom rate main theorem} to estimate the error between the prediction of the GL-GP covariance matrix or its Nystr\"{o}m extension and the prediction of the GL-GP covariance function.

\subsection{Measurement error}

In this subsection, we discuss the stability of GL-GP when there are measurement errors, so that the data do not fall exactly on the manifold.  

\begin{assumption}\label{Assumption:noisy data}
In addition to Assumption \ref{assumption DM}, due to  measurement error, we assume that the data we observe are $\{x'_1, \cdots, x'_{m+n}\} \subset \mathbb{R}^D$ such that $\|\iota(x_i)-x'_i\|_{\mathbb{R}^D}<\delta$, for $i=1,\cdots, m+n$.
\end{assumption}

Denote $L'$ to be the GL associated with $\{x'_1, \cdots, x'_{m+n}\}$. The following theorem shows that if the measurement errors are not too large, then one can still control the eigenvalues and eigenvectors of $L'$ by those of $L$. The proof of the theorem is in Section \ref{proof measurement error} of the appendices.
\begin{theorem}\label{measurement error}
Suppose Assumption \ref{Assumption:noisy data} holds.  For any $\epsilon$ small enough, if  $\delta <C \epsilon^{d+1}$ and $m+n$ is large enough so that $(\frac{\log(m+n)}{m+n})^{\frac{1}{2d+1}} \leq \epsilon$, then with probability greater $1-(m+n)^{-2}$, $\|L-L'\|_2 \leq C' \delta$,  where $\|\cdot \|_2$ is the spectral norm of a matrix, and $C$ and $C'$ are constants depending on $\mathsf p_m$ and the $C^0$ norm of $\mathsf p$. 
\end{theorem}

It will be interesting to relax Assumption \ref{Assumption:noisy data} and allow unbounded measurement error in future work. We refer the readers to \cite{el2016graph} for more discussion and relevant results. Since the GL-GP covariance matrix is constructed by using the eigenpairs of the GL, its stability follows from the above stability theorem of the GL.

\subsection{Posterior contraction rate of GL-GP on manifold under the fixed design setup}
Let $M$ be a $d$-dimensional smooth, closed and connected Riemannian manifold with diameter bounded by $1$. Assume $M$ is isometrically embedded in $\mathbb{R}^D$ through $\iota:M\to \mathbb{R}^D$. Let $\lambda_{0,\epsilon} \leq \lambda_{1,\epsilon} \leq \cdots$ be the eigenvalues of the operator $\mathsf{I}_\epsilon$. Suppose $\phi_{i,\epsilon}$ is the $i$-th eigenfunction of $\mathsf{I}_\epsilon$ and each $\phi_{i,\epsilon}$ is normalized in $L^2(M)$. We define the following finite dimensional inner product space $\mathbb{H}_{\epsilon,K,t}$.

\begin{definition}
For fixed $\epsilon$, $K$ and $t \geq 0$, we define $\mathbb{H}_{\epsilon,K,t}$ to be an inner product space on $M$ such that
\begin{align}
\mathbb{H}_{\epsilon,K,t}=\left\{h(x)=\sum_{i=0}^{K-1} a_i e^{-\lambda_{i,\epsilon} \frac{t}{2}} \phi_{i,\epsilon}(x),\,\,  \sum_{i=0}^{K-1} a_i^2 < \infty \right\} \nonumber
\end{align}
with the inner product 
\begin{align}
\left\langle\sum_{i=0}^{K-1} a_i e^{-\lambda_{i,\epsilon} \frac{t}{2}} \phi_{i,\epsilon}(x),\,\, \sum_{i=0}^{K-1} b_i e^{-\lambda_{i,\epsilon} \frac{t}{2}} \phi_{i,\epsilon}(x)\right\rangle_{\mathbb{H}_{\epsilon,K,t}}:=\sum_{i=0}^{K-1} a_i b_i. \nonumber
\end{align}
Denote $E_{i,\epsilon, t}:=e^{-\lambda_{i,\epsilon} \frac{t}{2}} \phi_{i,\epsilon}$. Then, $\{E_{i, \epsilon, t}\}_{i=0}^{K-1}$ is an orthonormal basis of $\mathbb{H}_{\epsilon,K,t}$. Denote $B^{\mathbb{H}_{\epsilon,K,t}}_R$ to be the ball of radius $R$ in $\mathbb{H}_{\epsilon,K,t}$. 
Also, let 
$\mathbb{H}_{\epsilon,K}:=\cup_{0 \leq t \leq 1} \mathbb{H}_{\epsilon,K,t}.$
\end{definition}

Refer to Section \ref{proof posterior rate} of the appendices for basic properties of the spaces $\mathbb{H}_{\epsilon,K,t}$ and $\mathbb{H}_{\epsilon,K}$. 

We explain the reason that we focus our discussion on the space $\mathbb{H}_{\epsilon,K}$. Any function in $L^\infty (M)$ with certain regularity, for example, a function in the Besov space, can be approximated by the functions in $\cup_{\epsilon} \cup_{K} \mathbb{H}_{\epsilon,K}$. The Besov space $B^s_{\infty, \infty}$ is a subspace of  $L^\infty (M)$ with some regularity quantitatively characterized by the parameter $s$; refer to Definition \ref{def of Besov space} in the appendices. We state our result as the following  approximation proposition. The proposition implies that $\cup_{\epsilon} \cup_{K} \mathbb{H}_{\epsilon,K}$ is a dense subset of $B^s_{\infty, \infty}$. 

\begin{proposition} \label{Approximation of Besov}
Suppose  $f_0 \in B^s_{\infty, \infty}$ with $s>0$ and $\|f_0\|_{\infty} \leq 1$. For any $\gamma>0$ small enough, if $K = \lceil{ D_1 \gamma^{-\frac{d}{s}}} \rceil$,  $\epsilon$ satisfies \eqref{relation epsilon K main theorem} and $\epsilon \leq D_2 \gamma^{\frac{10(d+s)}{9s}}$, then there is a function $h\in \mathbb{H}_{\epsilon,K}$ such that $\|h-f_0\|_{\infty} \leq \gamma$ and  $\|h\|^2_{\mathbb{H}_{\epsilon,K,t}} \leq  \texttt{vol}(M)\exp\left(D_3\gamma^{-\frac{2}{s}} t\right) $.  $D_1$ is a constant depending on $s$, $d$, the diameter and Ricci curvature of the manifold, $D_2$ depends on $s$, $d$, $\mathsf p_m$, the $C^2$ norm of $\mathsf p$, and the volume, injectivity radius, curvature and second fundamental form of the manifold and $D_3$ depends on $s$, $d$, the diameter, volume and Ricci curvature of the manifold.
\end{proposition}

Through the above Proposition, we construct a stratification indexed by $\epsilon$ and $K$ of a dense subset of $B^s_{\infty, \infty}$.  We associate each layer 
$\mathbb{H}_{\epsilon,K}$ in the stratification with a Gaussian process. Suppose $\{Z_i\}$ are independent random variables with mean $0$ and variance $1$. Fix $\epsilon>0$ and $K\in \mathbb{N}$. We define a Gaussian process for $t \geq 0$ indexed by $x \in M$,
\begin{align}\label{construction of GP by basis}
W^t_{\epsilon,K}(x):=\sum_{i=0}^{K-1} Z_i E_{i, \epsilon, t}(x) =\sum_{i=0}^{K-1} e^{-\lambda_{i,\epsilon} \frac{t}{2}} Z_i \phi_{i,\epsilon}(x)\,.
\end{align}
By a straightforward calculation,  the covariance function of $W^t_{\epsilon,K}(x)$, $\mathsf{C}_{\epsilon,K}(x, x', t)$, satisfies
\begin{align}
\mathsf{C}_{\epsilon,K}(x, x', t)=\mathbb{E}[W^t_{\epsilon,K}(x)W^t_{\epsilon,K}(x')]=\sum_{i=0}^{K-1}e^{-\lambda_{i,\epsilon}t} \phi_{i,\epsilon}(x)\phi_{i,\epsilon}(x')\,. \nonumber
\end{align}
Suppose $\textbf{T}$ is a random variable with  probability density function $g$ on $[0,1]$, with
\begin{align}\label{condition on g}
\mathcal{C}_1 t^{-p}e^{-t^{-q}} \leq g(t) \leq \mathcal{C}_2 t^{-p}e^{-t^{-q}},
\end{align}
for constants $0<\mathcal{C}_1\leq \mathcal{C}_2$ and $p,q >0$. These bounds are satisfied when $\textbf{T}$ follows an inverse transformed gamma density, induced by raising a gamma random variable to the power $-1/q$; in Theorem \ref{posterior rate fixed design}, $q$ is shown to relate to the dimension of the manifold. Using $g$ as a prior for $\textbf{T}=t$ in the GL-GP corresponding to $W^\textbf{T}_{\epsilon,K}(x)$ leads to a prior $\Pi$ on $\mathbb{H}_{\epsilon,K}$.

\begin{assumption}\label{assumption on GL-GP prior}
Let $M$ be a $d$-dimensional smooth, closed and connected Riemannian manifold isometrically embedded in $\mathbb{R}^D$ through $\iota:M\to \mathbb{R}^D$. We assume that the diameter of $M$ is bounded by $1$. We consider the regression model \eqref{eq:base} with $\{x_1 \cdots, x_{n}\} \subset M$. We propose the prior for the regression function $f$ as 
\begin{align}
W^\textbf{T}_{\epsilon,K}|\textbf{T} \sim GP(0, \mathsf{C}_{\epsilon,K}(x, y, \textbf{T})), \quad \textbf{T} \sim g(t), \nonumber
\end{align}
where $g(t)$ is a density function on $[0,1]$ satisfying \eqref{condition on g}.
\end{assumption}

\begin{remark}
The assumption that the diameter of $M$ is bounded by $1$ is for notation simplicity and the below result holds without this assumption. 
\end{remark}

Following \cite{ghosal2000convergence, ghosal2007convergence}, given a true function $f_0\in  \mathbb{H}_{\epsilon, K}$ and a data set $\{(x_1, y_1), \cdots, (x_n,y_n)\}$, we say that the {\em posterior contraction rate} of the Gaussian process prior in the fixed design is at least $\gamma_n$ if
\begin{align}
\Pi\left(\frac{1}{n}\sum_{i=1}^n(f(x_i)-f_0(x_i)) \geq \gamma_n \Big|\,\, \{(x_1, y_1), \cdots, (x_n,y_n)\}\right) \rightarrow 0 \nonumber
\end{align}
when $n \rightarrow \infty$. 
In the fixed design case,  $\{x_1, \cdots, x_n\}$ are deterministic rather than sampled based on some p.d.f on $M$.

\begin{remark}
The random variable $\textbf{T}$ can be regarded as a bandwidth parameter. Our bounds on the probability density function $g$ in \eqref{condition on g} are motivated by \cite{van2009adaptive} and \cite{bhattacharya2014anisotropic}. In showing minimax rates in the fixed design setup, they choose inverse transformed gamma priors for bandwidth parameters in squared exponential GPs, with the power equal to the Euclidean space dimension.
\end{remark}

We have the following theorem about the posterior contraction rate.

\begin{theorem}\label{posterior rate fixed design}Fix $K\in\mathbb{N}$ and fix any $\epsilon$ small enough so that \eqref{relation epsilon K main theorem} holds. Under Assumption \ref{assumption on GL-GP prior}, suppose $f_0 \in \mathbb{H}_{\epsilon, K}$ with $\|f_0\|_{\infty} \leq 1$ and $p \geq 1$, $q \geq \frac{d}{2}$. If we choose $\gamma_n=(\frac{n}{\log n})^{-\frac{d}{2q+2d}}$ with $\gamma_n \leq \frac{1}{K}$, then the posterior contraction rate of the Gaussian process prior in the fixed design is at least $\gamma_n$.
\end{theorem}

\begin{remark}
For notation simplicity, we assume $\|f_0\|_{\infty} \leq 1$. A similar result holds with the coefficients in the relation between $\gamma$ and $n$ depending on $\|f_0\|_{\infty}$. 
\end{remark}

\section{Discussion}

In this article, we were motivated by the problem of nonparametric regression with predictors on an unknown subset $S$ of $\mathbb{R}^D$ that might have a complicated geometric and topological structure.  The proposed GL-GP approach is appealing in allowing the GP covariance to reflect the intrinsic geometry of $S$.  Although we have taken substantial first steps theoretically, while showing promising empirical results in illustrative examples, there are multiple areas for future research.  

A first direction relates to developing scalable implementations of the GL-GP.  There is a rich literature developing scalable GP algorithms in other contexts, such as for huge spatial and/or temporal datasets modelled via GPs with traditional squared exponential or Mat\'ern covariance functions; for example, refer to \cite{datta2016hierarchical, datta2016nonseparable, peruzzi2020highly,ambikasaran2015fast}.  It is not straightforward to extend such algorithms to our context.  One possibility is to adapt subset-of-regressor \cite{rasmussen2003gaussian} and predictive process \cite{banerjee2008gaussian} algorithms via the Nystr\"om extension idea of subsection \ref{Nystrom subsection}, or the Roseland algorithm \cite{shen2020scalability} that could be viewed as a generalization of Nystr\"om via diffusion.

Another promising direction is to consider broader classes of GL-GPs by using more flexible kernels in constructing the GL; for example, a Mat\'ern kernel could be used in place of the Gaussian or the Gaussian kernel could be modified to have bandwidth parameters for each dimension.  It is interesting to consider the theoretical and practical behavior of such kernels, and the induced smoothness and covariance behavior for a variety of geometric structures.  The Mat\'ern case may be of particular relevance in applications to spatial statistics, in which our approach may provide a competitor to the current literature on spatial barriers and the so-called coastline problem; refer, for example to \cite{bakka2019non}.

Finally, there are several natural next steps theoretically.  One general direction is to build on our rate results, attempting to obtain the optimal rate in the manifold case. Another is to study consistency in estimating the covariance parameters, a problem of particular relevance in spatial statistics; refer, for example, to \cite{li2020bayesian, tang2019identifiability}.We would also like to obtain an improved understanding for much broader classes of $S$, including extensions beyond manifolds to stratified spaces and various metric spaces.

\section*{Acknowledgement}
We sincerely thank the associate editor and the referees for their comments to improve the quality of the paper. David B Dunson and Nan Wu acknowledge the support from the European Research Council (ERC) under the European Union’s Horizon 2020 research and innovation programme (grant agreement No 856506).

\bibliographystyle{plain}
\bibliography{bib}

\clearpage

\appendix

\section{Proof of Theorem \ref{Property of T epsilon on S}}\label{Proof of theorem Property of T epsilon on S}
Clearly $T_{\epsilon}$ is a linear integral operator. The self-adjointness follows from the fact that $Q_{\epsilon}(x,y)$ is symmetric. 
To show that $T_{\epsilon}$ is compact, we first show that $ \frac{Q_{\epsilon}(x,y)}{\int_{S }Q_{\epsilon}(x,y)d\mathsf P(y)}$ is a Hilbert-Schmidt  kernel on $S \times S$. Since $ \mathsf{P}$ is a probability measure, it is $\sigma$-finite. Also, since $S$ is a compact subset of $\mathbb{R}^D$, the Borel $\sigma$-algebra $\mathcal{F}$ is countably generated. By Proposition 3.4.5 in \cite{cohn2013measure}, $L^2(S, \mathsf{P})$ is a separable Hilbert space.  Note that 
\begin{align}\label{normalized Q espilon}
\frac{Q_{\epsilon}(x,x')}{\int_{S }Q_{\epsilon}(x,x')d\mathsf P(x')}=\frac{k_{\epsilon}(x,x')}{d_{\epsilon}(x') \int_{S} \frac{k_{\epsilon}(x,x')}{d_{\epsilon}(x')}d\mathsf P(x') }\,.
\end{align}
Since $S$ is a compact subset of $\mathbb{R}^D$, the diameter of $S$  in the Euclidean distance is bounded above by $\ell$. We thus have $e^{-\frac{\ell^2}{4\epsilon}} \leq k_{\epsilon}(x,x') \leq 1$.  Since $ \mathsf{P}$ is a probability measure, $e^{-\frac{\ell^2}{4\epsilon}} \leq d_{\epsilon}(x) \leq 1$. Substituting the upper and the lower bounds for $k_{\epsilon}(x,x')$ and $d_{\epsilon}(x)$ into \eqref{normalized Q espilon}, we have $\frac{Q_{\epsilon}(x,x')}{\int_{S }Q_{\epsilon}(x,x')d\mathsf P(x')} \leq e^{\frac{\ell^2}{2\epsilon}}$. Thus,
$\int_{S} \int_{S}\Big| \frac{Q_{\epsilon}(x,y)}{\int_{S }Q_{\epsilon}(x,x')d\mathsf P(x')} \Big|^2 d\mathsf P(x) d\mathsf P(x') \leq e^{\frac{\ell^2}{\epsilon}}$.
Since $S$ is compact  and hence locally compact Hausdorff, we conclude that $T_{\epsilon}$ is a compact operator from $L^2(S, \mathsf{P})$ to $L^2(S, \mathsf{P})$.
Finally, we have  
\begin{align}
& \sum_{i=1}^{m+n} A_{ki}f(x_i)=\sum_{i=1}^{m+n} \frac{W_{ki}f(x_i)}{D_{kk}}= \frac{\sum_{i=1}^{m+n} W_{ki}f(x_i)}{\sum_{i=1}^{m+n} W_{ki}} \nonumber  \\
=&  \frac{\frac{1}{m+n}\sum_{i=1}^{m+n} W_{ki}f(x_i)}{\frac{1}{m+n}\sum_{i=1}^{m+n} W_{ki}}= \frac{\frac{1}{m+n}\sum_{i=1}^{m+n} \frac{k_{\epsilon}(x_k,x_i)}{q_{\epsilon}(x_k) q_{\epsilon}(x_i)}f(x_i)}{\frac{1}{m+n}\sum_{i=1}^{m+n} \frac{k_{\epsilon}(x_k,x_i)}{q_{\epsilon}(x_k) q_{\epsilon}(x_i)}}= \frac{\frac{1}{m+n}\sum_{i=1}^{m+n} \frac{k_{\epsilon}(x_k,x_i)}{\frac{1}{m+n}q_{\epsilon}(x_k)\frac{1}{m+n} q_{\epsilon}(x_i)}f(x_i)}{\frac{1}{m+n}\sum_{i=1}^{m+n} \frac{k_{\epsilon}(x_k,x_i)}{\frac{1}{m+n}q_{\epsilon}(x_k) \frac{1}{m+n}q_{\epsilon}(x_i)}} \nonumber \\
=& \frac{\frac{1}{m+n}\sum_{i=1}^{m+n} \frac{k_{\epsilon}(x_k,x_i)}{\frac{1}{m+n} q_{\epsilon}(x_i)}f(x_i)}{\frac{1}{m+n}\sum_{i=1}^{m+n} \frac{k_{\epsilon}(x_k,x_i)}{ \frac{1}{m+n}q_{\epsilon}(x_i)}} . \nonumber
\end{align}
By the law of large numbers, $ \frac{1}{m+n} q_{\epsilon}(x_i) \rightarrow d_{\epsilon}(x_i)$ as $m+n \rightarrow \infty$ for all $x_i$ a.s. Hence, we have a.s.
\begin{align}
\frac{\frac{1}{m+n}\sum_{i=1}^{m+n} \frac{k_{\epsilon}(x_k,x_i)}{\frac{1}{m+n} q_{\epsilon}(x_i)}f(x_i)}{\frac{1}{m+n}\sum_{i=1}^{m+n} \frac{k_{\epsilon}(x_k,x_i)}{ \frac{1}{m+n}q_{\epsilon}(x_i)}} 
\rightarrow \int_{S } \frac{k_{\epsilon}(x_k,x')}{d_{\epsilon}(x') \int_{S} \frac{k_{\epsilon}(x_k,x')}{d_{\epsilon}(x')}d\mathsf P(x') } f(x') d\mathsf P(x') =\frac{\int_{S }Q_{\epsilon}(x_k,x') f(x') d\mathsf P(x')}{\int_{S }Q_{\epsilon}(x_k,x')d\mathsf P(x')}=T_{\epsilon} f(x_k). \nonumber
\end{align}

\section{Algorithm of GL-GP under the manifold setup}\label{DBGP manifold algorithm summary}

\begin{algorithm}[bht!]
\SetAlgoLined
\Parameter{Algorithm inputs include $\epsilon$, $t$, $K$ and $\sigma_{noise}$.}
Construct the $(m+n) \times (m+n)$ matrices $W$ and $D$ as in (\ref{W matrix}) and (\ref{D matrix}) by using the kernel 
$k_{\epsilon}(x,x')=\exp\Big(-\frac{\|\iota(x)-\iota(x')\|^2_{\mathbb{R}^D}}{4\epsilon^2}\Big)$ 
and data points $\mathcal{X}=\{x_1, \cdots, x_m, x_{m+1}, \cdots, x_{m+n}\}$. Let
$\tilde{A} = D^{-1/2}W D^{-1/2}.$ 

Find the first $K$ eigenpairs of $\frac{I-\tilde{A}}{\epsilon^2}$, namely $\{\mu_{i,m+n,\epsilon}, U_{i,m+n,\epsilon}\}_{i=0}^{K-1}$. Let $\tilde{v}_{i,m+n,\epsilon}$ be the normalized vector of $D^{-1/2}U_{i,m+n,\epsilon}$ in the $l^2$ norm.

For $j=1,\cdots, m+n$, find
$\mathbb{N}(j)=|B^{\mathbb{R}^p}_{\epsilon}(\iota(x_j)) \cap \{\iota(x_1), \cdots, \iota(x_{m+n})\}|.$
Calculate
\begin{align}
\|\tilde{v}_{i,m+n,\epsilon}\|_{l^2(1/\hat{\mathsf{p}})}=\sqrt{\frac{|S^{d-1}|\epsilon^d}{d} \sum_{j=1}^{m+n} \frac{\tilde{v}_{i,m+n,\epsilon}^2(j)}{\mathbb{N}(j)}}. \nonumber 
\end{align}
For $i=0, \cdots, K-1$, construct
$v_{i,m+n,\epsilon}=\frac{\tilde{v}_{i,m+n,\epsilon}}{\|\tilde{v}_{i,m+n,\epsilon}\|_{l^2(1/\hat{\mathsf{p}})}}.$

Construct $\tilde{\mathsf H}_{\epsilon,K,t}$ as 
$\tilde{\mathsf H}_{\epsilon,K,t}=\sum_{i=0}^{K-1} e^{-\mu_{i,m+n,\epsilon} t} v_{i,m+n,\epsilon} v^\top_{i,m+n,\epsilon}\in \mathbb{R}^{(m+n)\times (m+n)}.$
Rewrite $\tilde{\mathsf H}_{\epsilon,K,t}$ as
$
\begin{bmatrix}
H_1 & H_2 \\
H_3 & H_4 \\
\end{bmatrix},
$
where $H_1\in \mathbb{R}^{m\times m}$, $H_2$, $H_3$, $H_4$ are block matrices.

Let $\textbf{y}\in \mathbb{R}^m$  be a vector with $\textbf{y}(i)=y_i$, where $y_i$ is the observed response variable corresponding to $x_i$. Then 
$\textbf{f}_*:=H_3(H_1+\sigma^2_{noise} I)^{-1}\textbf{y}$
is our proposed prediction.  

\caption{GL-GP ALGORITHM on manifold}
\end{algorithm}


\section{Proofs of Theorem \ref{GL-GP convergence rate main} and Theorem \ref{Nystrom rate main theorem}} \label{proof of convergence from H K epsilon t to C  K epsilon t}
In this section, let $\| \cdot \|_2:=\|\cdot\|_{L^2(M)}$ be the $L^2$ norm, $\|\cdot \|_\infty:=\|\cdot\|_{L^\infty(M)}$ be the $L^\infty$ norm. We have the following fact about the eigenfunctions of the Laplace-Beltrami operator, $\Delta$.

\begin{lemma}[\cite{hormander1968spectral,donnelly2006eigenfunctions}]\label{lemma hormander}
For a compact Riemannian manifold $(M,g)$ and $l>0$, we have the following bound for the $l$-th eigenvalue $\lambda_l$ and $L^2$ normalized eigenfunction $\phi_l$ of the Laplace-Beltrami operator:  
$\|\phi_l\|_\infty \leq C_1 \lambda_l^{\frac{d-1}{4}} \|\phi_l\|_2= C_1 \lambda_l^{\frac{d-1}{4}},$
where $C_1$ is a constant depending on the injectivity radius and sectional curvature of the manifold $M$.
\end{lemma}

We also have the following fact about the eigenvalues of the Laplace-Beltrami operator, $\Delta$.

\begin{lemma}[\cite{hassannezhad2016eigenvalue}]\label{lemma weyls law}
For a $d$-dimensional compact and connected Riemannian manifold $(M,g)$, the eigenvalues of the Laplace-Beltrami operator, $0=\lambda_0< \lambda_1\leq\ldots$, satisfy
\begin{align}
c_2^{1+d\sqrt{\kappa}}\texttt{diam}(M)^{-2} l^{2/d} \leq \lambda_l \leq \frac{(d-1)^2}{4}\kappa + \check C_2 V(M)^{-2/d}l^{2/d}\nonumber 
\end{align}
for all $l \geq 1$, where the $\text{Ric}_g\geq -(d-1)\kappa g$ for $\kappa\geq 0$, and $c_2>0$ and $\check C_2>0$ are constants depending on $d$ only. Hence, $c_3 l^{2/d} \leq \lambda_l \leq C_3  l^{2/d}$, where $c_3$ depends on $d$, the diameter and the Ricci curvature of $M$ and $C_3>0$ depends on $d$, the volume and the Ricci curvature of $M$.
\end{lemma}

The following Lemma shows the spectral convergence rate of the first $K$ eigenpairs of the operator $\frac{I-T_{\epsilon}}{\epsilon^2}$ to those of $\Delta$ as $\epsilon \rightarrow 0$. 

\begin{lemma}[\cite{2020spectral} Proposition 1] \label{T epsilon and Delta}
Assume that the eigenvalues of $\Delta$ are simple. Suppose $(\lambda_{i,\epsilon}, \phi_{i,\epsilon})$ is the $i$-th eigenpair of $\frac{I-T_{\epsilon}}{\epsilon^2}$ and $(\lambda_{i}, \phi_{i})$ is the $i$-th eigenpair of $-\Delta$. Assume both $\phi_{i,\epsilon}$ and $\phi_{i}$ are normalized in the $L^2(M)$ norm. For $K\in \mathbb{N}$, denote  $\mathsf \Gamma_K:=\min_{1 \leq i \leq K}\textup{dist}(\lambda_i, \sigma(-\Delta)\setminus \{\lambda_i\}).$  Suppose $\epsilon$ is small enough and
\begin{equation}\label{epsilon K relation 1}
\epsilon \leq \mathcal{K}_1 \min \left(\left(\frac{\min(\mathsf\Gamma_K,1)}{\mathcal{K}_2+\lambda_K^{d/2+5}}\right)^2,\, \frac{1}{(\mathcal{K}_3+\lambda_K^{(5d+7)/4})^2}\right)\,,
\end{equation} 
where $\mathcal{K}_1$ and $\mathcal{K}_2, \mathcal{K}_3>1$  are constants depending on $d$, $\mathsf p_m$, the $C^2$ norm of $\mathsf p$, and the volume, the injectivity radius, the curvature and the second fundamental form of the manifold. Then, there are $a_i \in \{-1, 1\}$ such that for all $0 \leq i < K$, 
$$|\lambda_{i,\epsilon}-\lambda_{i}|  \leq \epsilon^{\frac{3}{2}}, \quad \|a_i\phi_{i,\epsilon}-\phi_{i}\|_{\infty}  \leq \epsilon.$$
\end{lemma}

\begin{remark}
We assume that the eigenvalues of $\Delta$ are simple to simplify the notation. In the case when the eigenvalues are not simple, the same result still works by introducing the eigenprojection \cite{chatelin2011spectral}. 
\end{remark}

In the following lemma, we show that the first $K$ eigenpairs of $-L$ converge to those of $\frac{I-T_{\epsilon}}{\epsilon^2}$ as $m+n \rightarrow \infty$. The proof of the lemma follows from combining Corollary 1, Proposition 4, and Proposition 5 in \cite{2020spectral}, so we omit it.

\begin{lemma} \label{spectral convergence of L to I-T epsilon on closed manifold} Under Assumption \ref{assumption DM}, suppose  all the eigenvalues of $\Delta$ are simple. Let $\lambda_{i}$ be the $i$-th eigenvalue of $-\Delta$. Let $\mathsf \Gamma_K:=\min_{1 \leq i \leq K}\textup{dist}(\lambda_i, \sigma(-\Delta)\setminus \{\lambda_i\})$. Let $\mu_{i,m+n,\epsilon}$  be the $i$-th eigenvalue of $-L$ with the associated eigenvector $v_{i,m+n,\epsilon}$ normalized  as in \eqref{normalized V i n epsilon}. Let $(\lambda_{i,\epsilon}, \phi_{i,\epsilon})$ be the $i$-th eigenpair of $\frac{I-T_{\epsilon}}{\epsilon^2}$
with $\phi_{i,\epsilon}$ normalized in $L^2(M)$. If $\epsilon$ is small enough so that \eqref{epsilon K relation 1} holds and
$m+n$ is sufficiently large so that $(\frac{\log(m+n)}{m+n})^{\frac{1}{4d+13}} \leq \epsilon$, then  there are $a_i \in \{-1, 1\}$ such that for all $0 \leq i < K$, with probability greater than $1-(m+n)^{-2}$,  
$$|\mu_{i,n,\epsilon}-\lambda_{i,\epsilon}|\leq \Omega_1 \epsilon^{3/2}, \quad  
\max_{x_j \in \mathcal{X}}|a_i v_{i,m+n,\epsilon}(j)-\phi_{i,\epsilon}(x_j)|\leq  \Omega_2 \epsilon^{1/2}\,,$$ 
where 
$\Omega_1$ depends on $d$, the diameter of $M$, $\mathsf p_m$, and the $C^2$ norm of $\mathsf p$, and $\Omega_2$ depends on $d$, the diameter and the volume of $M$, $\mathsf p_m$, and the $C^2$ norm of $\mathsf p$.
\end{lemma} 

\begin{remark}
In \cite{2020spectral}, the statements of Proposition 1, Proposition 3, Proposition 4, and Proposition 5 do not include the case $i=0$.  The case $i=0$ is treated individually in the proof of the main theorem. However, those propositions are still correct for the case $i=0$.  Hence, we include the case $i=0$ in the statements of Lemmas \ref{T epsilon and Delta} and \ref{spectral convergence of L to I-T epsilon on closed manifold}.
\end{remark}

\textbf{Proof of Theorem \ref{GL-GP convergence rate main}}
By the triangular inequality, we have
\begin{eqnarray}
\lefteqn{ |\tilde{\mathsf H}_{\epsilon,K,t}(i,j)-\mathsf{C}_{\epsilon,K}(x_i, x_j, t)| 
= \Big|\sum_{l=0}^{K-1} \big(e^{-\mu_{l,m+n,\epsilon} t} v_{l,m+n,\epsilon}(i) v^\top_{l,m+n,\epsilon}(j)-e^{-\lambda_{l,\epsilon} t} \phi_{l,\epsilon}(x_i)\phi_{l,\epsilon}(x_j)\big)\Big| } \nonumber \\
& \leq  \sum_{l=0}^{K-1} \Big(|e^{-\mu_{l,m+n,\epsilon} t}-e^{-\lambda_{l,\epsilon} t}||\phi_{l,\epsilon}(x_i)\phi_{l,\epsilon}(x_j)|+|v_{l,m+n,\epsilon}(i) v^\top_{l,m+n,\epsilon}(j)-\phi_{l,\epsilon}(x_i)\phi_{l,\epsilon}(x_j)||e^{-\mu_{l,m+n,\epsilon} t}|\Big). \nonumber
\end{eqnarray}
By Lemma \ref{spectral convergence of L to I-T epsilon on closed manifold}, with probability greater than $1-(m+n)^{-2}$, for all $l < K$, 
$$|\mu_{l,m+n,\epsilon}-\lambda_{l,\epsilon}|\leq  \Omega_1\epsilon^{\frac{3}{2}},\qquad 
\max_{x_i}|a_l v_{l,m+n,\epsilon}(i)-\phi_{l,\epsilon}(x_i)|\leq  \Omega_2 \epsilon^{\frac{1}{2}}.$$ 
By Proposition \ref{positivity of eigenvalue of -L}, $\mu_{l,m+n,\epsilon} \geq 0$. Thus, $|e^{-\mu_{l,m+n,\epsilon} t}| \leq 1$ and we have
\begin{align}
|e^{-\mu_{l,m+n,\epsilon} t}-e^{-\lambda_{l,\epsilon} t}| \leq & |e^{-\mu_{l,m+n,\epsilon} t}||1-e^{(\mu_{l,m+n,\epsilon} -\lambda_{l,\epsilon} )t}| \leq  |1-e^{(\mu_{l,m+n,\epsilon} -\lambda_{l,\epsilon} )t}| \nonumber \\
 \leq & |\mu_{l,m+n,\epsilon} -\lambda_{l,\epsilon}|t e^{|\mu_{l,m+n,\epsilon} -\lambda_{l,\epsilon}|t} \leq   \Omega_1\epsilon^{\frac{3}{2}}t e^{\Omega_1\epsilon^{\frac{3}{2}}t} 
 \leq  \Omega_1 \epsilon^{\frac{3}{2}}t e^{\Omega_1 t}. \nonumber
\end{align}
Next, we bound $\|\phi_{l,\epsilon}\|_{\infty}$.  By Lemma \ref{T epsilon and Delta}, for $l<K$, $\|\phi_{l,\epsilon}\|_{\infty} \leq \|\phi_{l}\|_{\infty}+\epsilon$. Hence, 
\begin{align}
\|\phi_{0,\epsilon}\|_{\infty} \leq \|\phi_{0}\|_{\infty}+\epsilon \leq \frac{1}{\sqrt{\texttt{vol}(M)}}+\epsilon. \nonumber 
\end{align}
By Lemma \ref{lemma hormander}, for $1 \leq l <K$,
$\|\phi_{l,\epsilon}\|_{\infty} \leq \|\phi_{l}\|_{\infty}+\epsilon \leq  C_1 \lambda_l^{\frac{d-1}{4}}+\epsilon.$
We combine the above two cases and use the fact that $\epsilon < \mathcal{K}_1$, for all $l<K$,
\begin{align}\label{l infinity bound of all eignfunctions}
\|\phi_{l,\epsilon}\|_{\infty} \leq C_1 \lambda_l^{\frac{d-1}{4}}+\frac{1}{\sqrt{\texttt{vol}(M)}}+\epsilon  \leq C_1 \lambda_K^{\frac{d-1}{4}}+\frac{1}{\sqrt{\texttt{vol}(M)}}+\epsilon  \leq C_2 (\lambda_K^{\frac{d-1}{4}}+1),
\end{align}
where $C_2=C_1+\frac{1}{\sqrt{\texttt{vol}(M)}}$ which depends on $\mathsf p_m$, the $C^2$ norm of $\mathsf p$, the volume, the injectivity radius and the sectional curvature of the manifold. Therefore,
$|\phi_l(x_i)\phi_l(x_j)| \leq C^2_2 (\lambda_K^{\frac{d-1}{4}}+1)^2$ 
and
\begin{align}
|e^{-\mu_{l,n,\epsilon} t}-e^{-\lambda_l t}||\phi_l(x_i)\phi_l(x_j)| \leq \Omega_1 \epsilon^{\frac{3}{2}}t e^{\Omega_1 t} C^2_2 (\lambda_K^{\frac{d-1}{4}}+1)^2 \,. \nonumber  
\end{align}
Next, we bound the term $|v_{l,m+n,\epsilon}(i) v^\top_{l,m+n,\epsilon}(j)-\phi_{l,\epsilon}(x_i)\phi_{l,\epsilon}(x_j)||e^{-\mu_{l,m+n,\epsilon} t}|$. Since $\mu_{l,m+n,\epsilon} \geq 0$,
\begin{align}
|v_{l,m+n,\epsilon}(i) v^\top_{l,m+n,\epsilon}(j)-\phi_{l,\epsilon}(x_i)\phi_{l,\epsilon}(x_j)||e^{-\mu_{l,m+n,\epsilon} t}| \leq |v_{l,m+n,\epsilon}(i) v^\top_{l,m+n,\epsilon}(j)-\phi_{l,\epsilon}(x_i)\phi_{l,\epsilon}(x_j)|. \nonumber 
\end{align}
Also, for $l < K$, $|v_{l,m+n,\epsilon}(i) v^\top_{l,m+n,\epsilon}(j)-\phi_{l,\epsilon}(x_i)\phi_{l,\epsilon}(x_j)|$ is controlled by
\begin{align}
 &\, \max_{x_i}|a_lv_{l,m+n,\epsilon}(i)-\phi_{l,\epsilon}(x_i)|\left(\max_{x_i}|\phi_{l,\epsilon}(x_i)|+\max_{x_i}|v_{l,m+n,\epsilon}(i)|\right) \nonumber \\
\leq &\, \left(2 \max_{x_i}|\phi_{l,\epsilon}| + \Omega_2 \epsilon^{1/2}\right)\max_{x_i}|a_lv_{l,m+n,\epsilon}(i)-\phi_{l,\epsilon}(x_i)| \leq  C_2 (\lambda_K^{\frac{d-1}{4}}+1)\Omega_2 \epsilon^{1/2}.\nonumber
\end{align}
In conclusion,  $|\tilde{\mathsf H}_{\epsilon,K,t}(i,j)-\mathsf{C}_{\epsilon,K}(x_i, x_j, t)|$ is controlled by
\begin{align}
& K \Bigg(\Omega_1 \epsilon^{\frac{3}{2}}t e^{\Omega_1 t}  C^2_2 (\lambda_K^{\frac{d-1}{4}}+1)^2+ C_2 (\lambda_K^{\frac{d-1}{4}}+1)\Omega_2 \epsilon^{1/2} \Bigg) \nonumber \\
\leq & K \epsilon^{1/2}  \Bigg(\Omega_1 \epsilon t e^{\Omega_1 t}  C^2_2 (\lambda_K^{\frac{d-1}{4}}+1)^2+ C_2 (\lambda_K^{\frac{d-1}{4}}+1)\Omega_2  \Bigg) \nonumber \\
\leq & K \epsilon^{1/2}  (\Omega_1 \texttt{diam}(M) e^{\Omega_1 \texttt{diam}(M)}  C^2_2 +C_2\Omega_2) \Big(\lambda_K^{\frac{d-1}{4}}+1\Big)^2 
= K \epsilon^{1/2} \mathcal{K}_4 \Big(\lambda_K^{\frac{d-1}{4}}+1\Big)^2, \nonumber 
\end{align}
By Lemma \ref{lemma weyls law}, $$K(\lambda_K^{\frac{d-1}{4}}+1)^2 \leq K(C_3^{\frac{d-1}{4}} K^{\frac{d-1}{2d}}+1)^2 \leq (C_3^{\frac{d-1}{4}}+1)^2 K^2\,.
$$
Hence, $|\tilde{\mathsf H}_{\epsilon,K,t}(i,j)-\mathsf{C}_{\epsilon,K}(x_i, x_j, t)| \leq \mathcal{K}_4 K^2 \epsilon^{1/2}$,
where 
$$
\mathcal{K}_4=[\Omega_1 \epsilon \texttt{diam(}M) e^{\Omega_1 \texttt{diam}(M)}  C^2_2 +C_2\Omega_2](C_3^{\frac{d-1}{4}}+1)^2 $$ 
that depends on  $d$, $\mathsf p_m$,  the $C^2$ norm of $\mathsf p$, the diameter,  the volume and the Ricci curvature of $M$.

\

\textbf{Proof of Theorem \ref{Nystrom rate main theorem}}
(a) By Proposition \ref{extension of the eigenvector},  $\frac{1}{1-\epsilon^2\mu_{i,m+n,\epsilon}}E \tilde{v}_{i,m+n,\epsilon}(j)=\tilde{v}_{i,m+n,\epsilon}(j)$, where $x_j \in \mathcal{X}$. Hence, $\frac{1}{1-\epsilon^2\mu_{i,m+n,\epsilon}}E v_{i,m+n,\epsilon}(j)=v_{i,m+n,\epsilon}(j)$, where $x_j \in \mathcal{X}$.  Since the eigenvalue of $\tilde{\mathsf H}^*$ remain the same as $\mu_{i,m+n,\epsilon}$ by the construction, the result follows.

(b) By Corollary 1, Proposition 4, and Proposition 5 in \cite{2020spectral}, the same result as Lemma \ref{spectral convergence of L to I-T epsilon on closed manifold} holds for the extension vector $\frac{1}{1-\epsilon^2\mu_{i,m+n,\epsilon}}E v_{i,m+n,\epsilon}$. Specifically, we have
$$
\max_{x_j \in \mathcal{X}^*}\left| \frac{a_i}{1-\epsilon^2\mu_{i,m+n,\epsilon}}E v_{i,m+n,\epsilon}(j)-\phi_{i,\epsilon}(x_j)\right|\leq  \Omega_2 \epsilon^{1/2}
$$
with $a_i \in \{-1,1\}$. Since the eigenvalues of $\tilde{\mathsf H}^*$ remain the same as $\mu_{i,m+n,\epsilon}$ by the construction, the proof of the Theorem \ref{Nystrom rate main theorem} follows from the method as the proof of Theorem \ref{GL-GP convergence rate main}.

\section{Proof of Theorem \ref{difference theorem}} \label{difference in prediction}

Note that $\tilde{\mathsf H}_{\epsilon,K,t}$ and $\Sigma$ are both symmetric. Hence, by our assumption, there is a symmetric $(m+n) \times (m+n)$ matrix $E$ such that 
$\tilde{\mathsf H}_{\epsilon,K,t}=\Sigma +\delta E, $ 
where $E$ satisfies $\|E\|_\infty=O(1)$. Let 
$
E= \begin{bmatrix}
E_{11} & E_{12}  \\
E_{21} & E_{22}
\end{bmatrix}, 
$
where $E_{11}$ is a symmetric $m \times m$ matrix. Therefore,
we have 
$H_1=\Sigma_{11}+\delta E_{11}$ and 
$H_3=\Sigma_{21}+\delta E_{21}.$ 
Suppose $H_1+\sigma^2_{noise} I$ has the eigendecomposition $H_1+\sigma^2_{noise} I=\bar{U} \bar{\Lambda} \bar{U}^\top$ and $\Sigma_{11}+\sigma^2_{noise} I$ has the eigendecomposition $\Sigma_{11}+\sigma^2_{noise} I=U \Lambda U^\top$. We can apply perturbation theory for real symmetric matrices (Appendix A and Lemma E.4 in \cite{wu2018think}) such that we have
$\bar{\Lambda}=\Lambda+\delta \Lambda'$ and  
$\bar{U}=U \Theta +\delta U',$ 
where $\Lambda'$ is a diagonal matrix of order $O(1)$, $U'$ is of order $O(1)$ and $\Theta \in O(m)$ commutes with $\Lambda$ and $\Lambda^{-1}$. Therefore,
\begin{align} 
\textbf{f}_{*}-\textbf{f}_{\texttt{GLGP}} & =H_3(H_1+\sigma^2_{noise} I)^{-1}\textbf{y}-\Sigma_{21}(\Sigma_{11}+\sigma^2_{noise} I)^{-1}\textbf{y}, \nonumber
\end{align}
which can be expanded to
\begin{align}
&(\Sigma_{21}+\delta E_{21})(U \Theta+\delta U')(\Lambda+\delta \Lambda')^{-1}(\Theta^\top U^\top+\delta {U'}^\top) \textbf{y}  -\Sigma_{21} U \Theta  \Lambda^{-1} \Theta^\top U^\top \textbf{y} \nonumber \\
=&\,\Sigma_{21} U\Theta \big[(\Lambda+\delta \Lambda')^{-1}-\Lambda^{-1}\big]\Theta^\top U^\top \textbf{y}+\delta \Sigma_{21} U \Theta (\Lambda+\delta \Lambda')^{-1}{U'}^\top \textbf{y} \nonumber \\
&\quad+\delta\Sigma_{21} U'(\Lambda+\delta \Lambda')^{-1} \Theta^\top{U'}^\top \textbf{y}+\delta^2 \Sigma_{21} U' (\Lambda+\delta \Lambda')^{-1}{U'}^\top \textbf{y} \nonumber \\
&\quad+ \delta E_{21}(U \Theta+\delta U')(\Lambda+\delta \Lambda')^{-1}(\Theta^\top U^\top+\delta {U'}^\top) \textbf{y}\,. \nonumber 
\end{align}
Note that $\|(\Lambda+\delta \Lambda')^{-1}\|_{\infty} \leq \frac{1}{\lambda+c\delta}$ and $\|(\Lambda+\delta \Lambda')^{-1}-\Lambda^{-1}\|_{\infty} \leq \frac{c\delta}{\lambda(\lambda+c \delta)}$, where $c$ is a constant depending on $E_{11}$.
Therefore, we have that 
\begin{align}
\|\textbf{f}_{{\texttt{GLGP}}}-\textbf{f}_{*}\|_{\infty} \leq & \frac{c\delta \|\Sigma_{21}\textbf{y}\|_{\infty}}{\lambda(\lambda+c\delta)} +\frac{2\delta\|\Sigma_{21}\textbf{y}\|_{\infty}}{\lambda+c\delta} +\frac{\delta^2\|\Sigma_{21}\textbf{y}\|_{\infty}}{\lambda+c\delta} +\frac{m \delta \|\textbf{y}\|_{\infty}}{\lambda+c\delta}\nonumber  \\
\leq &  \frac{\delta m  \|\textbf{y}\|_{\infty}\big(( c+3\lambda) \|\Sigma_{21}\|_{\infty}+ \lambda \big)}{\lambda(\lambda+c\delta)}. \nonumber 
\end{align}

\section{Proof of Theorem \ref{measurement error}}\label{proof measurement error}

For any pairs $x_i, x_j$ and $x'_i, x'_j$, without loss of generality, we assume
$\|\iota(x_i)-\iota(x_j)\|_{\mathbb{R}^D}<\|x'_i-x'_j\|_{\mathbb{R}^D}$. Then, a sequence of trivial bounds leads to
\begin{align}
& |k_{\epsilon}(x_i,x_j)-k_{\epsilon}(x'_i,x'_j)| \nonumber \\
= & \Big|\exp\Big(-\frac{\|\iota(x_i)-\iota(x_j)\|^2_{\mathbb{R}^D}}{4\epsilon^2}\Big)-\exp\Big(-\frac{\|x'_i-x'_j\|^2_{\mathbb{R}^D}}{4\epsilon^2}\Big)\Big|
\nonumber \\
 \leq & \exp\Big(-\frac{\|\iota(x_i)-\iota(x_j)\|^2_{\mathbb{R}^D}}{4\epsilon^2}\Big) \Big|\frac{\|\iota(x_i)-\iota(x_j)\|^2_{\mathbb{R}^D}-\|x'_i-x'_j\|^2_{\mathbb{R}^D}}{4\epsilon^2}\Big| \nonumber \\
 \leq & \exp\Big(-\frac{\|\iota(x_i)-\iota(x_j)\|^2_{\mathbb{R}^D}}{4\epsilon^2}\Big) |\|\iota(x_i)-\iota(x_j)\|_{\mathbb{R}^D}-\|x'_i-x'_j\|_{\mathbb{R}^D}|\Big|\frac{\|\iota(x_i)-\iota(x_j)\|_{\mathbb{R}^D}+\|x'_i-x'_j\|_{\mathbb{R}^D}}{4\epsilon^2}\Big| \nonumber \\
 \leq & \frac{2\delta}{\epsilon}\exp\Big(-\frac{\|\iota(x_i)-\iota(x_j)\|^2_{\mathbb{R}^D}}{4\epsilon^2}\Big)\Big(\frac{\|\iota(x_i)-\iota(x_j)\|_{\mathbb{R}^D}}{2\epsilon}+\frac{\delta}{4\epsilon}\Big) 
 \leq  \frac{2\delta}{\epsilon}\Big(\frac{1}{2}+\frac{\delta}{4\epsilon}\Big) \leq \frac{2\delta}{\epsilon}\,. \nonumber
\end{align}
Note that we use the fact that $a\exp(-a^2)<\frac{1}{2}$ for any $a>0$ in the second to last step. Hence, for any $i$,
$|q_{\epsilon}(x_i)-q_{\epsilon}(x'_i)|<\frac{2(m+n)\delta}{\epsilon}.$
Similarly, we can show that 
\begin{align}
\exp\Big(-\frac{2\texttt{diam}(M) \delta}{\epsilon^2}\Big) \leq \frac{k_{\epsilon}(x_i,x_j)}{k_{\epsilon}(x'_i,x'_j)} \leq \exp\Big(\frac{2\texttt{diam}(M) \delta}{\epsilon^2}\Big)\,, \nonumber 
\end{align}
where $\texttt{diam}(M)$ is the diameter of the manifold.
Thus, for any $i$,
$$\exp\Big(-\frac{2\texttt{diam}(M) \delta}{2\epsilon^2}\Big) q_{\epsilon}(x_i) < q_{\epsilon}(x'_i), \quad 
k_{\epsilon}(x'_i,x'_j) \leq  \exp\Big(\frac{2\texttt{diam}(M) \delta}{\epsilon^2}\Big).$$
By Lemma 10 in \cite{2020spectral}, if {$\epsilon>0$} is small enough and $m+n$ is large enough so that $\frac{1}{\sqrt{m+n}\epsilon^d} (\sqrt{-\log\epsilon}+\sqrt{\log (m+n)})< C'_1$, where $C'_1$ depends on $d$, the diameter of $M$, $\mathsf p_m$, and the $C^0$ norm of $\mathsf p$, then with probability greater $1-(m+n)^{-2}$, $C_4(m+n)\epsilon^d< q_{\epsilon}(x_i)<C_5 (m+n)\epsilon^d$, where $C_4$ and $C_5$ depend on $\mathsf p_m$ and the $C^0$ norm of $\mathsf p$. Note that $\frac{1}{\sqrt{m+n}\epsilon^d} (\sqrt{-\log\epsilon}+\sqrt{\log (m+n)})< C'_1$ is satisfied when $(\frac{\log (m+n)}{m+n})^{\frac{1}{2d+1}} \leq \epsilon$. Also note that $(\frac{\log (m+n)}{m+n})^{\frac{1}{2d+1}} \leq \epsilon$ implies $\frac{1}{m+n} \leq \epsilon^{2d+1}$.  So, we have $$C_4\exp\Big(-\frac{2\texttt{diam}(M) \delta}{\epsilon^2}\Big) (m+n)\epsilon^d< q_{\epsilon}(x'_i),$$ and hence 
\begin{align}
|W_{ij}-W'_{ij}| \leq & \frac{|k_{\epsilon}(x_i,x_j)-k_{\epsilon}(x'_i,x'_j)|}{q_{\epsilon}(x_i)q_{\epsilon}(x_j)} +k_{\epsilon}(x'_i,x'_j)\left(\frac{|q_{\epsilon}(x_i)-q_{\epsilon}(x'_i)|}{q_{\epsilon}(x_i)q_{\epsilon}(x'_i)q_{\epsilon}(x'_j)}+\frac{|q_{\epsilon}(x_j)-q_{\epsilon}(x'_j)|}{q_{\epsilon}(x_i)q_{\epsilon}(x_j)q_{\epsilon}(x'_j)}\right)\nonumber \\
\leq & \frac{2\delta}{C^2_1 (m+n)^2\epsilon^{2d+1}} + \frac{2 \exp(\frac{6\texttt{diam}(M) \delta}{\epsilon^2})\delta}{C^2_1 (m+n)^3\epsilon^{3d+1}} 
< \frac{5 \delta}{C_4^2 (m+n)^2 \epsilon^{2d+1}}\,, \nonumber
\end{align}
where we use the fact that  $\frac{1}{(m+n)} \leq \epsilon^{2d+1}$ in the last step. 
Also,
$\gamma_1 :=\max_{i,j}W_{ij} \leq \frac{1}{C^2_1 (m+n)^2\epsilon^{2d}}$
and
\begin{align}
\gamma_2 :=&\, \min_i \sum_{j \not= i} \frac{W_{ij}}{(m+n)}= \min_i \sum_{j \not= i} \frac{k_{\epsilon}(x_i,x_j)}{(m+n)q_\epsilon(x_i)q_\epsilon(x_j)}>\frac{1}{C^2_2 (m+n)^3\epsilon^{2d}}\min_i \sum_{j \not= i}k_{\epsilon}(x_i,x_j) \nonumber \\
=&\, \frac{1}{C^2_2 (m+n)^3\epsilon^{2d}} \min_{i} (q_{\epsilon}(x_i)-1)>\frac{C_4 (m+n)\epsilon^{d}-1}{C^2_2 (m+n)^3\epsilon^{2d}}=\frac{C_4-\frac{1}{(m+n)\epsilon^d}}{C^2_2 (m+n)^2 \epsilon^d}. \nonumber
\end{align}
Suppose $\delta<\frac{C_4}{20C^2_2}\epsilon^{d+1} $. Then $\gamma_2 >\frac{C_4-\frac{1}{n\epsilon^d}}{C^2_2 (m+n)^2 \epsilon^d} >2  \times \frac{5 \delta}{C_4^2 (m+n)^2 \epsilon^{2d+1}}$. Therefore, by Lemma 2.1 in  \cite{el2016graph},
\begin{align}
\|L-L'\|_2 \leq \frac{\gamma_1}{\gamma_2}\frac{5 \delta}{C_4^2 (m+n)^2 \epsilon^{2d+1}} +\frac{\gamma^2_1}{\gamma_2\left(\gamma_2-\frac{5\delta }{C_4^2 (m+n)^2 \epsilon^{2d+1}}\right)}\frac{5 \delta}{C_4^2 (m+n)^2 \epsilon^{2d+1}} \leq \frac{C' \delta}{(m+n)^2\epsilon^{4d+1}}, \nonumber 
\end{align}
where $C'$ depends on $\mathsf p_m$ and the $C^0$ norm of $\mathsf p$. We substitute the relation $\frac{1}{(m+n)} \leq \epsilon^{2d+1}$ again and the conclusion follows.

\section{Proofs of Theorem \ref{posterior rate fixed design}}\label{proof posterior rate}

We discuss some properties of the inner product space  $\mathbb{H}_{\epsilon,K,t}$. First, we have the following bounds controlling $\| \cdot \|_{\mathbb{H}_{\epsilon,K,t}}$. The proof follows directly from the fact that $M$ is compact and $\{\phi_{i,\epsilon}\}$ forms an orthonormal basis of $L^2(M)$ and we omit details. 

\begin{lemma} \label{H norm bounds}
Fix $\epsilon>0$, $K\in\mathbb{N}$ and $t \geq 0$.  For $f\in \mathbb{H}_{\epsilon,K,t}$, 
$
e^{-\lambda_{K-1,\epsilon} t }\|f \|^2_{\mathbb{H}_{\epsilon,K,t}} \leq \|f\|^2_{2} \leq \texttt{vol}(M) \|f\|^2_{\infty}.
$
\end{lemma}

Note that for fixed $\epsilon$ and $K$, the spaces $\mathbb{H}_{\epsilon,K,t}$ are isomorphic for all $ t \geq 0$. Moreover, we have the following nested sequence of unit balls. 

\begin{lemma} \label{nested unit balls}
Fix $\epsilon>0$ and $K\in\mathbb{N}$. If $0 \leq t_1 \leq t_2$, then $B^{\mathbb{H}_{\epsilon,K,t_2}}_1 \subset B^{\mathbb{H}_{\epsilon,K,t_1}}_1$.
\end{lemma}

\begin{proof}
$f(x) = \sum_{i=0}^{K-1} b_i e^{-\lambda_{i,\epsilon} \frac{t_2}{2}}\phi_{i,\epsilon}(x)  \in  B^{\mathbb{H}_{\epsilon,K,t_2}}_1$
can be rewritten as $f(x)=\sum_{i=0}^{K-1} a_i e^{-\lambda_{i,\epsilon} \frac{t_1}{2}}\phi_{i,\epsilon}(x)$ for $t_1\leq t_2$. 
Then we have $a_i e^{-\lambda_{i,\epsilon} \frac{t_1}{2}}=b_i e^{-\lambda_{i,\epsilon} \frac{t_2}{2}}$. Since $t_1 \leq t_2$, we have $a_i \leq b_i$. Hence, $\sum_{i=0}^{K-1} a^2_i \leq \sum_{i=0}^{K-1} b^2_i \leq 1$. In other words , we have $f(x) \in  B^{\mathbb{H}_{\epsilon,K,t_1}}_1$. Hence, $ B^{\mathbb{H}_{\epsilon,K,t_2}}_1 \subset  B^{\mathbb{H}_{\epsilon,K,t_1}}_1$.
\end{proof}

Next, we recall some definitions about the Littlewood-Paley theory in the smooth functional calculus. We first construct a bump function $\Psi$, or father wavelet, satisfying the following conditions: 
(1) $\Psi \in C^\infty(\mathbb{R}_{\geq 0})$ with support in $[0,1]$,
(2) $0 \leq  \Psi(x) \leq 1$ when $x \geq 0.5$,
(3) $\Psi(x)=1$ for $0 \leq x \leq 0.5$,
(4) $\Psi'(x) \leq \mathcal{C}$ for some $\mathcal{C}>0$.

\begin{definition}
Fix $\delta>0$ and $\epsilon>0$. We define the following integral kernels and operators: 
$$ \Psi(\delta \sqrt{\Delta})(x,x'):=\sum_{i=0}^\infty \Psi(\delta\sqrt{\lambda_i})\phi_i(x)\phi_i(x'),\quad 
 \Psi_{K}(\delta \sqrt{\mathsf{I}_\epsilon})(x,x'):=\sum_{i=0}^{K-1} \Psi(\delta\sqrt{\lambda_{i,\epsilon}})\phi_{i,\epsilon}(x)\phi_{i,\epsilon}(x').$$ 
Note that by Proposition \ref{eigenvalue lower bound of I epsilon}, $\lambda_{i,\epsilon} \geq 0$. Hence, $ \Psi_{K}(\delta \sqrt{\mathsf{I}_\epsilon})(x,x')$ is well-defined.
The corresponding integral operators are defined as 
$$ \Psi(\delta \sqrt{\Delta})f(x)=\int_M \Psi(\delta \sqrt{\Delta})(x,x')f(x') dV(x'),\quad 
\Psi_{K}(\delta \sqrt{\mathsf{I}_\epsilon})f(x)=\int_M \Psi_{K}(\delta \sqrt{\mathsf{I}_\epsilon})(x,x') f(x') dV(x').$$
\end{definition}

Set $\Psi_j(x)=\Psi(2^{-j}x)$ for $j \geq 0$.  We can thus define the Besov space we have interest in this paper.

\begin{definition}\label{def of Besov space}(Besov space)
Fix $s>0$. $f \in B^s_{\infty, \infty}$ if $f \in L^\infty(M)$ and 
\begin{align}
\|f\|_{B^s_{\infty, \infty}}:=\sup_{j} 2^{sj} \|\Psi_j(\sqrt{\Delta})f-f\|_\infty < \infty. \nonumber  
\end{align}
\end{definition}
The above definition is independent of the choice of bump function $\Psi$.

\subsection{Proof of Proposition \ref {Approximation of Besov}}
\begin{proof}
If $f_0 \in B^s_{\infty, \infty}$, it follows from the definition of the Besov space that we have 
$\|\Psi(\delta \sqrt{\Delta})f_0-f_0\|_\infty \leq C \delta^s$
for some $C>0$. Denote $\gamma:=2C \delta^s$.
For a small $\delta$, by the definition of $\Psi$, we have 
\begin{align}
\|\Psi(\delta \sqrt{\Delta})f_0-f_0\|_\infty =&\, \left\|\int_M \sum_{\lambda_i \leq \delta^{-2}} \Psi(\delta\sqrt{\lambda_i})\phi_i(x)\phi_i(x') f_0(x') dV(x')-f_0\right\|_\infty\nonumber\\
\leq&\, C\delta^s =\gamma/2\,. 
\label{Proof:proposition4:result1}
\end{align}
By Lemma \ref{lemma weyls law}, we have $\lambda_{K} \geq c_3 K^{\frac{2}{d}}$.
Suppose  $K$ satisfies $ c_3 K^{\frac{2}{d}} >  \delta^{-2}=  (\frac{\gamma}{2 C})^{\frac{-2}{s}}$, i.e. $K > c_3^{-\frac{d}{2}}(2C)^{\frac{d}{s}}\gamma^{-\frac{d}{s}}$. Then, for any $i \geq K$, $\lambda_i > \delta^{-2}$ and $\Psi(\delta\sqrt{\lambda_i})=0$. Therefore, if we choose $K=\Big \lceil{c_3^{-\frac{d}{2}}(2C)^{\frac{d}{s}}\gamma^{-\frac{d}{s}}}\Big \rceil= \Big \lceil{D_1 \gamma^{-\frac{d}{s}}}\Big \rceil$, where $D_1:= c_3^{-\frac{d}{2}}(2C)^{\frac{d}{s}}$, we have 
\[
\Psi(\delta \sqrt{\Delta})f_0=\int_M \sum_{i =0}^{K-1}\Psi(\delta\sqrt{\lambda_i})\phi_i(x)\phi_i(x') f_0(x') dV(x'). 
\]
Note that $D_1$ depends on $s$, $d$, the diameter and the Ricci curvature of the manifold.

Next, we  find a suitable $\epsilon$ so that $\int_M \sum_{i =0}^{K-1}\Psi(\delta\sqrt{\lambda_i})\phi_i(x)\phi_i(x') f_0(x') dV(x')$ can be well controlled by $\Psi_{K}(\delta \sqrt{\mathsf{I}_\epsilon})f_0(x)$. It comes from a sequence of bounds.
\begin{align}\label{approximation 1 proposition 3}
& \left\|\Psi_{K}(\delta \sqrt{\mathsf{I}_\epsilon})f_0(x) - \int_M \sum_{i =0}^{K-1}\Psi\Big(\delta\sqrt{\lambda_i}\Big)\phi_i(x)\phi_i(x') f_0(x') dV(x')\right\|_\infty \\
\leq & \texttt{vol}(M) K \max_{0 \leq i \leq K-1}\left\|\Big[ \Psi\Big(\delta\sqrt{\lambda_{i,\epsilon}}\Big)- \Psi\Big(\delta\sqrt{\lambda_{i}}\Big)\Big] \phi_i(x)\phi_i(x')+\Psi\Big(\delta\sqrt{\lambda_{i,\epsilon}}\Big)\big[\phi_{i,\epsilon}(x)\phi_{i,\epsilon}(x')- \phi_i(x)\phi_i(x')\big]\right\|_{\infty} \nonumber \\
\leq & \texttt{vol}(M) K \left(\max_{0 \leq i \leq K-1}\Big|\Psi\Big(\delta\sqrt{\lambda_{i,\epsilon}}\Big)- \Psi\Big(\delta\sqrt{\lambda_{i}}\Big)\Big|  \max_{0 \leq i \leq K-1}\|\phi_i\|^2_{\infty}+\max_{0 \leq i \leq K-1}\|\phi_{i,\epsilon}(x)\phi_{i,\epsilon}(x')- \phi_i(x)\phi_i(x')\|_{\infty}\right)\,, \nonumber 
\end{align}
where we bound $\Psi(\delta\sqrt{\lambda_{i,\epsilon}})$ by $1$. We need some quantities to further control the right hand side of \eqref{approximation 1 proposition 3}. If $\delta$ is small enough and
\begin{equation}\label{epsilon K relation 2}
\epsilon \leq \mathcal{K}_1 \min \left(\left(\frac{\min(\mathsf\Gamma_K,1)}{\mathcal{K}_2+\lambda_K^{d/2+5}}\right)^2,\, \frac{1}{(\mathcal{K}_3+\lambda_K^{(5d+7)/4})^2}\right)\,,
\end{equation} 
then by applying Lemma \ref{T epsilon and Delta},  $\lambda_{0,\epsilon} \leq \epsilon^{3/2}$. Hence, $\delta \sqrt{\lambda_{0,\epsilon}}<0.5$ and $\Psi(\delta\sqrt{\lambda_{0,\epsilon}})=\Psi(\delta\sqrt{\lambda_{0}})=1$.
Note that \eqref{epsilon K relation 2} implies that $\lambda_{K-1}^{\frac{d-1}{4}} \leq (\frac{\epsilon}{\mathcal{K}_1})^{-\frac{1}{10}}$.
Thus,
the right hand side of \eqref{approximation 1 proposition 3} is  further bounded by 
\begin{align}
 & \texttt{vol}(M) K \left(\max_{1 \leq i \leq K-1} \mathcal{C}\delta|\sqrt{\lambda_{i,\epsilon}}-\sqrt{\lambda_{i}}|  C^2_1 \lambda_{K-1}^{\frac{d-1}{2}}+\max_{0 \leq i \leq K-1} \|\phi_{i,\epsilon}-\phi_{i}\|_{\infty}(\|\phi_{i,\epsilon}\|_{\infty}+\|\phi_{i}\|_{\infty}) \right) \nonumber \\
\leq & \texttt{vol}(M) K \left(\mathcal{C}  C^2_1  \delta \max_{1 \leq i \leq K-1} \frac{|\lambda_{i,\epsilon}-\lambda_{i}|}{\sqrt{\lambda_{i,\epsilon}}+\sqrt{\lambda_{i}}} \lambda_{K-1}^{\frac{d-1}{2}}+\max_{0 \leq i \leq K-1} \|\phi_{i,\epsilon}-\phi_{i}\|_{\infty}(\|\phi_{i,\epsilon}\|_{\infty}+\|\phi_{i}\|_{\infty}) \right) \nonumber \\
\leq & \texttt{vol}(M) K \left(\mathcal{C}  C^2_1  \delta \frac{\epsilon^{\frac{3}{2}}}{\sqrt{\lambda_1}}  (\frac{\epsilon}{\mathcal{K}_1})^{-\frac{1}{5}}+\max_{0 \leq i \leq K-1} \epsilon (2\|\phi_{i}\|_{\infty}+\epsilon) \right)  \nonumber \\
\leq & \texttt{vol}(M) K \left(\frac{\mathcal{C}  C^2_1\mathcal{K}_1^{\frac{1}{5}}}{\sqrt{\lambda_1}}  \delta \epsilon^{1.3} + \epsilon\Big (2(C_1+\frac{1}{\sqrt{\texttt{vol}(M)}})(\lambda_K^{\frac{d-1}{4}}+1)+\epsilon\Big) \right) \nonumber \\
\leq & \texttt{vol}(M) K \left(\frac{\mathcal{C}  C^2_1\mathcal{K}_1^{\frac{1}{5}}}{\sqrt{\lambda_1}}  \delta \epsilon^{1.3} + 4(C_1+\frac{1}{\sqrt{\texttt{vol}(M)}}) \mathcal{K}_1^{\frac{1}{10}}\epsilon^{0.9}\right) \nonumber \\
\leq  & \texttt{vol}(M) \left(\frac{\mathcal{C}  C^2_1\mathcal{K}_1^{\frac{1}{5}}}{\sqrt{\lambda_1}}+ 4(C_1+\frac{1}{\sqrt{\texttt{vol}(M)}}) \mathcal{K}_1^{\frac{1}{10}}\right)K \epsilon^{0.9}\,.\label{Bound:Proof:proposition4:1} 
\end{align}
Note that $K= \Big \lceil{D_1 \gamma^{-\frac{d}{s}}}\Big \rceil \leq (D_1+1) \gamma^{-\frac{d}{s}}$, when $\gamma<1$. Hence, \eqref{Bound:Proof:proposition4:1} is further controlled by
$$\texttt{vol}(M) \left(\frac{\mathcal{C}  C^2_1\mathcal{K}_1^{\frac{1}{5}}}{\sqrt{\lambda_1}}+ 4(C_1+\frac{1}{\sqrt{\texttt{vol}(M)}}) \mathcal{K}_1^{\frac{1}{10}}\right)(D_1+1) \gamma^{-\frac{d}{s}}\epsilon^{0.9}.$$
Clearly, if we have 
\begin{align}
 \epsilon  \leq  \left(2\texttt{vol}(M) \left(\frac{\mathcal{C}  C^2_1\mathcal{K}_1^{\frac{1}{5}}}{\sqrt{\lambda_1}}+ 4(C_1+\frac{1}{\sqrt{\texttt{vol}(M)}}) \mathcal{K}_1^{\frac{1}{10}}\right)(D_1+1) \right)^{-\frac{10}{9}}  \gamma^{\frac{10(d+s)}{9s}}= D_2  \gamma^{\frac{10(d+s)}{9s}}\,,
\nonumber  
\end{align} 
then we have
\begin{align}
\left\|\Psi_{K}(\delta \sqrt{\mathsf{I}_\epsilon})f_0(x) - \int_M \sum_{i =0}^{K-1}\Psi\Big(\delta\sqrt{\lambda_i}\Big)\phi_i(x)\phi_i(x') f_0(x') dV(x')\right\|_\infty \leq \frac{\gamma}{2}\,.\label{Proof:proposition4:result2}
\end{align}
Note that $D_2$ depends on $s$, $d$, $\mathsf p_m$, the $C^2$ norm of $\mathsf p$, and the volume, the injectivity radius, the curvature and the second fundamental form of the manifold.
By \eqref{Proof:proposition4:result1} and \eqref{Proof:proposition4:result2}, we have $\|\Psi_{K}(\delta \sqrt{\mathsf{I}_\epsilon})f_0-f_0\|_\infty \leq \gamma.$

Let $h(x) := \Psi_{K}(\delta \sqrt{\mathsf{I}_\epsilon})f_0(x).$ To finish the proof, we claim that $h\in \mathbb{H}_{\epsilon,K}$. Since $\{\phi_{i,\epsilon}\}$ form a basis of $L^2(M)$, we have $\|f_0\|^2_2=\sum_{i=0}^\infty |\int_M \phi_{i,\epsilon}(y) f_0(y) dV(y)|^2$. Therefore, by the assumption of $f_0$, we have
\begin{equation}
\sum_{i=0}^{K-1}\left|\int_M \phi_{i,\epsilon}(y) f_0(y) dV(y)\right|^2 \leq \|f_0\|^2_2 \leq \texttt{vol}(M)\,.\label{bound K terms sum for f0} 
\end{equation}
As a result, by a direct bound, we have
\begin{align}
\|h\|^2_{\mathbb{H}_{\epsilon,K,t}} =& \left\|\int_M \sum_{i=0}^{K-1}  \Psi\left(\delta\sqrt{\lambda_{i,\epsilon}}\right)\phi_{i,\epsilon}(x)\phi_{i,\epsilon}(x') f_0(x') dV(x')\right\|^2_{\mathbb{H}_{\epsilon,K,t}} \nonumber \\
\leq & \left\|\sum_{i=0}^{K-1}   \Psi\left(\delta\sqrt{\lambda_{i,\epsilon}}\right) \int_M \phi_{i,\epsilon}(x') f_0(x') dV(x') \phi_{i,\epsilon}(x) \right\|^2_{\mathbb{H}_{\epsilon,K,t}} \nonumber \\
\leq & \sum_{i=0}^{K-1} \left| \Psi\left(\delta\sqrt{\lambda_{i,\epsilon}}\right)\right|^2 \left|\int_M \phi_{i,\epsilon}(x') f_0(x') dV(x')\right|^2 \exp\left(\lambda_{i,\epsilon}t\right) \nonumber \\
\leq & \sum_{i=0}^{K-1} \left| \Psi\left(\delta\sqrt{\lambda_{i,\epsilon}}\right)\right|^2 \left|\int_M \phi_{i,\epsilon}(x') f_0(x') dV(x')\right|^2 \exp\left(2\lambda_{i}t\right) \nonumber \\
\leq & \exp\left(2 \lambda_K t\right) \sum_{i=0}^{K-1}\left|\int_M \phi_{i,\epsilon}(x') f_0(x') dV(x')\right|^2 \nonumber\\
\leq & \texttt{vol}(M)\exp\left(2 \lambda_K t\right) \,,\nonumber
\end{align}
where the last bound comes from \eqref{bound K terms sum for f0}. Recall the fact that $K\leq (D_1+1) \gamma^{-\frac{d}{s}}$, when $\gamma<1$. Hence by Lemma \ref{lemma weyls law}, 
$$2\lambda_K \leq 2 C_3 (D_1+1)^{\frac{2}{d}} \gamma^{-\frac{2}{s}}=D_3\gamma^{-\frac{2}{s}},$$
where $D_3$ depends on $s$, $d$, the diameter, the volume and the Ricci curvature of $M$. Hence, $\|h\|^2_{\mathbb{H}_{\epsilon,K,t}} \leq  \texttt{vol}(M)\exp\left(D_3\gamma^{-\frac{2}{s}} t\right) $.
\end{proof}

\subsection{Proof of Theorem \ref{posterior rate fixed design}}

We introduce the definition of the concentration function.

\begin{definition}
Fix $\epsilon>0$, $K\in\mathbb{N}$ and $t\in [0,1]$.  Consider the Gaussian Process $W^t_{\epsilon,K}$ defined in \eqref{construction of GP by basis}. Let
\[
S^{t}(\gamma):=-\log (\mathbb{P}(\|W^t_{\epsilon,K}\|_{\infty} <\gamma))\,.
\]
For any $f_0\in \mathbb{H}_{\epsilon,K}$, the concentration function of the Gaussian process $W^t_{\epsilon,K}$ is defined as 
\[
\phi^t_{f_0}(\gamma)=\|f_0\|^2_{\mathbb{H}_{\epsilon,K,t}}+S^{t}(\gamma).
\]
\end{definition}

We have the following bound for the concentration function. 

\begin{lemma}\label{upper bounds on concentration function}
Fix $K\in\mathbb{N}$ and fix any $\epsilon$ small enough so that \eqref{epsilon K relation 1} holds. For $f_0 \in \mathbb{H}_{\epsilon,K}$ with $\|f_0\|_\infty \leq 1$, we have 
$$\|f_0\|^2_{\mathbb{H}_{\epsilon,K,t}} \leq \texttt{vol}(M) \exp\left(2 C_3 K^{\frac{2}{d}} t \right),\quad 
S^{t}(\gamma) \leq C_6 K \log \left(\frac{C_7}{\gamma}\right),$$
where $C_3$ is the constant defined in Lemma \ref{lemma weyls law}, $C_6$ is a constant and $C_7$ depends on the volume of $M$.
\end{lemma}

\begin{proof}
By Lemma \ref{H norm bounds}, we have
$\|f_0 \|^2_{\mathbb{H}_{\epsilon,K,t}} \leq \texttt{vol}(M) \exp\left(\lambda_{K-1 ,\epsilon} t \right) \|f_0 \|^2_{\infty}.$
Based on Lemma \ref{T epsilon and Delta} and Lemma \ref{lemma weyls law}, 
we have
$\lambda_{K-1 ,\epsilon} \leq 2\lambda_{K-1} \leq 2 C_3 K^{\frac{2}{d}}.$
Hence, 
\begin{align} \label{bounds of the norm in K}
\|f_0 \|^2_{\mathbb{H}_{\epsilon,K,t}} \leq \texttt{vol}(M) \exp\left(  2 C_3 K^{\frac{2}{d}} t \right) \|f_0 \|^2_{\infty}\leq \texttt{vol}(M) \exp\left(  2 C_3 K^{\frac{2}{d}} t \right)\,.
\end{align}

For the second term, let 
$B^{\mathbb{H}_{\epsilon,K,t}}_1$ be a unit ball centered at $0$ in the $\mathbb{H}_{\epsilon,K,t}$ norm.
Let
$N(\gamma, B^{\mathbb{H}_{\epsilon,K,t}}_1, L^\infty)$ be the covering number of $ B^{\mathbb{H}_{\epsilon,K,t}}_1$ by balls of radius $\gamma>0$ in the $L^\infty(M)$ norm. Clearly, $B^{\mathbb{H}_{\epsilon,K,t}}_1$ can be covered by no more than $(\frac{4}{R})^K$ balls of radius $R>0$ in the $\mathbb{H}_{\epsilon,K,t}$ norm. By \eqref{bounds of the norm in K}, if $B^\infty_{\gamma}$ is a ball of radius $\gamma>0$ centered at $0$ in the $L^\infty(M)$ norm and if we set $R:=\texttt{vol}(M) \exp\left(C_3 K^{\frac{2}{d}} t \right) \gamma$, then 
$B^{\mathbb{H}_{\epsilon,K,t}}_{R} \subset  B^{\infty}_{\gamma}.$
Hence, 
$N(\gamma, B^{\mathbb{H}_{\epsilon,K,t}}_1, L^\infty) \leq \left(\frac{4}{R}\right)^K \leq \left(\frac{C_5}{\gamma}\right)^K \exp\left(-C_3 K^{\frac{2}{d}+1} t\right) \leq \left(\frac{C_5}{\gamma}\right)^K,$
where $C_5$ depends on $\texttt{vol}(M)$. Hence, $\log N(\gamma, B^{\mathbb{H}_{\epsilon,K,t}}_1, L^\infty)  \leq K \log (\frac{C_5}{\gamma})$.
With this bound, based on the same argument as that for Proposition 3 in \cite{castillo2014thomas}, we have
\begin{align*}
-\log (\mathbb{P}(\|W^t_{\epsilon,K}\|_{\infty} <\gamma)) \leq C_6 K \log \left(\frac{C_7}{\gamma}\right),
\end{align*}
where $C_6$ is a constant and $C_7$ depends on $\texttt{vol}(M)$. The conclusion follows.
\end{proof}

Let 
$\Phi(u)=\int_{-\infty}^u \frac{1}{\sqrt{2\pi}}e^{-\frac{w^2}{2}}dw$
be the cumulative density function (cdf) of the standard normal distribution. The following Lemma about $\Phi(u)$ can be found in \cite{van2009adaptive}

\begin{lemma}\label{gaussian cdf}(Lemma 4.10 in \cite{van2009adaptive})
$\Phi(u)$ satisfies $\Phi(u) \leq \exp\left(-\frac{u^2}{2}\right)$ for $u<0$. Moreover, we have $-\sqrt{2\log(\frac{1}{v})} \leq \Phi^{-1}(v)$ for $0 < v < 1$.
\end{lemma}

We introduce the following technical lemma to prove the main result of this section.

\begin{lemma}\label{xln1/x}
If $\gamma_n=(\frac{n}{\log n})^{-\frac{d}{2q+2d}}$ and $n$ is sufficiently large so that $\log n \geq C_7$, then $\gamma_n^{-\frac{2q}{d}} \log \big(\frac{C_7} {\gamma_n}\big) \leq  n \gamma^2_n$, where $C_7$ is defined in Lemma \ref{upper bounds on concentration function}.
\end{lemma}
\begin{proof}
If $\log n \geq C_7$, then we have $(\frac{n}{\log n})^{-\frac{d}{2q+2d}} \geq (\frac{n}{C_7})^{-\frac{d}{2q+2d}}$. Hence, $\gamma_n \geq (\frac{n}{C_7})^{-\frac{d}{2q+2d}}$, which is equivalent to $n \geq C_7 \gamma_n^{-\frac{2q+2d}{d}}$.  
Moreover, a straightforward calculation shows that $\gamma_n=(\frac{n}{\log n})^{-\frac{d}{2q+2d}}$  is equivalent to $\gamma_n^{-\frac{2q}{d}} \log n = n \gamma^2_n$. Hence, 
\begin{align*}
n \gamma^2_n= \gamma_n^{-\frac{2q}{d}} \log n \geq  \gamma_n^{-\frac{2q}{d}} \log \left(C_7 \gamma_n^{-\frac{2q+2d}{d}}\right) \geq \gamma_n^{-\frac{2q}{d}} \log \left(\frac{C_7} {\gamma_n}\right). 
\end{align*}
\end{proof}

The main result of this section is summarized in the following proposition. Theorem \ref{posterior rate fixed design} follows from the above proposition and the same argument as Theorem 3.1 in \cite{van2008rates}.

\begin{proposition}
Fix $K\in\mathbb{N}$ and $\epsilon$ small enough so that \eqref{epsilon K relation 1} holds. For the probability density function of $\textbf{T}$, assume $p>1$ and $q \geq \frac{d}{2}$. Take $f_0 \in \mathbb{H}_{\epsilon, K}$ with $\|f_0\|_{\infty} \leq 1$. Denote $\gamma_n=(\frac{n}{\log n})^{-\frac{d}{2q+2d}}$. Then, when $n$ is sufficiently large and $\gamma_n \leq \frac{1}{K}$, we have 
\begin{eqnarray}
\mathbb{P}(\|W^\textbf{T}_{\epsilon,K}-f_0\|_{\infty} \leq \gamma_n) \geq \mathcal{A}_1 \exp\left(-\mathcal{A}_2 n \gamma_n^2\right),  \label{Prior mass}
\end{eqnarray}
where $\mathcal{A}_1>0$ and $\mathcal{A}_2>0$ are constants depending on  $\texttt{vol}(M)$,
and there is a Borel measurable subset $\mathsf B_n \subset \mathbb{H}_{\epsilon, K}$ such that
\begin{align}
& \mathbb{P}(W^\textbf{T}_{\epsilon,K} \not \in \mathsf B_n) \leq \exp\left(-C_6 n \gamma^2_n\right)\,, \label{Sieve}\\
& \log N(\tilde{\gamma_n}, \mathsf B_n, L^\infty) \leq \frac{1}{4}n\tilde{\gamma}^2_n\,, \label{Entropy}
\end{align}
where $\tilde{\gamma}_n:=2\gamma_n $.
\end{proposition}

\begin{proof}
We start by proving \eqref{Prior mass}.
Based on \cite[Lemma 5.3]{van2008reproducing}, for a fixed $t\in(0,1]$, we have 
\begin{align}
 e^{-\phi^t_{f_0}(\gamma_n)} \leq \mathbb{P}(\|W^t_{\epsilon,K}-f_0\|_{\infty} \leq 2\gamma_n) \leq  e^{-\phi^t_{f_0}(2\gamma_n)}. \nonumber 
\end{align}
Hence, 
\begin{align}
&\mathbb{P}(\|W^\textbf{T}_{\epsilon,K}-f_0\|_{\infty} \leq 2\gamma_n)=\int_{0}^1 \mathbb{P}(\|W^t_{\epsilon,K}-f_0\|_{\infty} \leq 2\gamma_n) g(t) dt\nonumber\\
\geq&\, \int_{0}^1 e^{-\phi^t_{f_0}(\gamma_n)}g(t)dt \geq \int_{t^*}^{2t^*} e^{-\phi^t_{f_0}(\gamma_n)}g(t)dt, \nonumber 
\end{align}
where $t^*\in (0,1/2)$ will be determined later.
By Lemma \ref{upper bounds on concentration function}, 
$$\phi^t_{f_0}(\gamma_n) \leq  \texttt{vol}(M) \exp\left(  2 C_3 K^{\frac{2}{d}} t \right)+C_6 K \log \left(\frac{C_7}{\gamma_n}\right) 
\leq \texttt{vol}(M) \exp\left(  2 C_3 \gamma_n^{-\frac{2}{d}} t \right)+C_6 \gamma_n^{-1} \log \left(\frac{C_7}{\gamma_n}\right),$$
where the second inequality holds by the assumption that $K \leq \frac{1}{\gamma_n}$. When $n$ is sufficiently large, 
choose $t^*=\frac{1}{4C_3}(\gamma_n)^{\frac{2}{d}}\log(\frac{1} {\gamma_n})\in (0,1/2)$. Then, for $t\in [t^*, 2t^*]$, we have 
\begin{align}
\phi^t_{f_0}(\gamma_n) \leq  \texttt{vol}(M) \exp\left( 2 C_3 \gamma_n^{-\frac{2}{d}} t \right)+C_6 \gamma_n^{-1} \log \left(\frac{C_7}{\gamma_n}\right) \leq \texttt{vol}(M)  \gamma_n^{-1} +C_6 \gamma_n^{-1} \log \left(\frac{C_7}{\gamma_n}\right)\, .  \nonumber
\end{align}
Hence, when $n$ is sufficiently large, there is a constant $C_8>0$ depending on $\texttt{vol}(M)$ such that
$\phi^t_{f_0}(\gamma_n) \leq C_8 \gamma_n^{-1} \log \left(\frac{C_7} {\gamma_n}\right).$ 
Observe that $t^* \rightarrow 0$ as $n\to\infty$ and the lower bound for $g(t)$, which is $\mathcal{C}_1 t^{-p}e^{-t^{-q}}$, goes to $0$ as $t\to 0$. Therefore, 
$
\inf_{t\in [t^*, 2t^*]} g(t)=\mathcal{C}_1 {t^*}^{-p}e^{-{t^*}^{-q}}
$ 
when $n$ is sufficiently large.
As a result,
\begin{eqnarray}
\lefteqn{
\mathbb{P}(\|W^\textbf{T}_{\epsilon,K}-f_0\|_{\infty} \leq 2\gamma_n) 
 \geq  \int_{t^*}^{2t^*} e^{-\phi^t_{f_0}(\gamma_n)}g(t)dt
\geq  t^* \exp\left(- C_8\gamma_n^{-1} \log \left(\frac{C_7} {\gamma_n}\right)\right) \mathcal{C}_1 {t^*}^{-p}e^{-{t^*}^{-q}} } \nonumber \\
= &\, \left(\frac{1}{4C_3}\right)^{1-p}\gamma_n^{\frac{2-2p}{d}}\log^{1-p}\left(\frac{1} {\gamma_n}\right) \exp\left(- C_8\gamma_n^{-1} \log \left(\frac{C_7} {\gamma_n}\right)-(4C_3)^{-q}\gamma_n^{-\frac{2q}{d}}\log^{-q}\left(\frac{1} {\gamma_n}\right)\right)\,. \nonumber
\end{eqnarray}
When $n$ is sufficiently large, we have $(\gamma_n)^{-\frac{2}{d}}\log^{-1}(\frac{1} {\gamma_n})>1$. Using the assumptions that $p \geq 1$ and $q \geq \frac{d}{2}$,  the above term is bounded below by
$C_9 \exp\left(- C_{10} \gamma_n^{-\frac{2q}{d}} \log \left(\frac{C_7} {\gamma_n}\right)\right),$
where $C_9$ depends on $\texttt{vol}(M)$ and $p$ and $C_{10}$ depends on $\texttt{vol}(M)$ and $q$.
Finally, by Lemma \ref{xln1/x}, when $n$ is sufficiently large so that $\log n \geq C_7$, 
$
\mathbb{P}(\|W^\textbf{T}_{\epsilon,K}-f_0\|_{\infty} \leq 2\gamma_n)  \geq C_9 \exp\left(- C_{10} n \gamma_n^2\right)\,. 
$
Hence, we obtain our claim \eqref{Prior mass} that 
$$
\mathbb{P}(\|W^\textbf{T}_{\epsilon,K}-f_0\|_{\infty} \leq \gamma_n)  \geq C_9 \exp\left(-4 C_{10} n \gamma_n^2\right)=\mathcal{A}_1 \exp\left(-\mathcal{A}_2 n \gamma_n^2\right)
$$
for $\mathcal{A}_1,\mathcal{A}_2>0$.

Second, we prove \eqref{Sieve}.
Since $K \leq \frac{1}{\gamma_n}$, by Lemma \ref{upper bounds on concentration function}, for any $t\in [0,1]$, 
\begin{align}
S^{t}(\gamma_n)=-\log (\mathbb{P}(\|W^t_{\epsilon,K}\|_{\infty} <\gamma_n)) \leq C_6 \gamma_n^{-1} \log \left(\frac{C_7}{\gamma_n}\right)\,. \nonumber 
\end{align}
Let $B^{\mathbb{H}_{\epsilon, K, 0}}_1$ be the ball of radius $1$ centered at $0$ in $\mathbb{H}_{\epsilon, K, 0}$ and let $B_1$ be the unit ball centered at $0$ in $L^\infty(M)$.  Let $M=2 \sqrt{2C_6} \sqrt{n} \gamma_n$ and set
\[
\mathsf B_n:= (M B^{\mathbb{H}_{\epsilon, K, 0}}_1+\gamma_n B_1) \cap \mathbb{H}_{\epsilon, K}.  
\]
By Lemma \ref{nested unit balls}, we have for $t \geq 0$
\begin{align}
(M B^{\mathbb{H}_{\epsilon, K, t}}_1+\gamma_n B_1) \cap \mathbb{H}_{\epsilon, K} \subset (M B^{\mathbb{H}_{\epsilon, K, 0}}_1+\gamma_n B_1) \cap \mathbb{H}_{\epsilon, K} =\mathsf B_n\,.
\nonumber 
\end{align}
Hence, we have 
\begin{align}
\mathbb{P}(W^t_{\epsilon,K} \not \in \mathsf B_n) & =\mathbb{P}(W^t_{\epsilon,K} \not \in   (M B^{\mathbb{H}_{\epsilon, K, 0}}_1+\gamma_n B_1) \cap \mathbb{H}_{\epsilon, K} )  \leq  \mathbb{P}(W^t_{\epsilon,K} \not \in   (M B^{\mathbb{H}_{\epsilon, K, t}}_1+\gamma_n B_1) \cap \mathbb{H}_{\epsilon, K} ) \nonumber 
\end{align}

Note that $B_1$ is the unit ball centered at $0$ in $L^\infty(M)$ rather than in $ \mathbb{H}_{\epsilon, K,t}$. However, based on the definition of $S^{t}(\gamma_n)$, $\mathbb{P}(\|W^t_{\epsilon,K}\|_{\infty} <\gamma_n))=\exp\left(-S^{t}(\gamma_n)\right))$.
Hence, if we apply Borell's inequality, we have 
$$ \mathbb{P}(W^t_{\epsilon,K} \in   (M B^{\mathbb{H}_{\epsilon, K, t}}_1+\gamma_n B_1) \cap \mathbb{H}_{\epsilon, K} ) \geq \Phi(\Phi^{-1}(\exp\left(-S^{t}(\gamma_n)\right))+M).$$
In conclusion,
\begin{align}
\mathbb{P}(W^t_{\epsilon,K} \not \in \mathsf B_n) & \leq 1-\Phi(\Phi^{-1}(\exp\left(-S^{t}(\gamma_n)\right))+M)=\Phi(-\Phi^{-1}(\exp\left(-S^{t}(\gamma_n)\right))-M)\,. \nonumber 
\end{align}

Since $\mathbb{P}(\|W^t_{\epsilon,K}\|_{\infty} <\gamma_n) \rightarrow 0$ as $\gamma_n \rightarrow 0$, if $\gamma_n$ is small enough, we have $\mathbb{P}(\|W^t_{\epsilon,K}\|_{\infty} <\gamma_n) <\frac{1}{4}$, and hence $\exp\left(-S^{t}(\gamma_n)\right) < \frac{1}{4}$. 
By Lemma \ref{gaussian cdf}, 
$\Phi^{-1}(\exp\left(-S^{t}(\gamma_n)\right)) \geq -\sqrt{2 S^{t}(\gamma_n)}.$ 
So, 
\begin{align}
-\Phi^{-1}(\exp\left(-S^{t}(\gamma_n)\right))-M \leq  \sqrt{2 C_6 \gamma_n^{-1} \log \left(\frac{C_7}{\gamma_n}\right)}- 2 \sqrt{2 C_6} \sqrt{n} \gamma_n\,.   \nonumber
\end{align}
By Lemma \ref{xln1/x} and the fact that $q \geq \frac{d}{2}$, if $\gamma_n=(\frac{n}{\log n})^{-\frac{d}{2q+2d}}$ and $\log n \geq C_7$, then we have 
$$\gamma_n^{-1} \log \left(\frac{C_7} {\gamma_n}\right) \leq \gamma_n^{-\frac{2q}{d}} \log \left(\frac{C_7} {\gamma_n}\right) \leq  n \gamma^2_n.$$ 
Hence,
\[
-\Phi^{-1}(\exp\left(-S^{t}(\gamma_n)\right))-M \leq  \sqrt{2 C_6 n \gamma^2_n }- 2 \sqrt{2 C_6} \sqrt{n} \gamma_n \leq - \sqrt{2  C_6} \sqrt{n} \gamma_n.
\]
By Lemma \ref{gaussian cdf}, $\mathbb{P}(W^t_{\epsilon,K} \not \in \mathsf B_n) \leq \exp\left(-C_6 n \gamma^2_n\right)$. Therefore,
\[
\mathbb{P}(W^\textbf{T}_{\epsilon,K} \not \in \mathsf B_n) \leq \int_{0}^1 \exp\left(-C_6 n \gamma^2_n\right) g(t)dt \leq \exp\left(-C_6 n \gamma^2_n\right).
\]

Finally, we prove \eqref{Entropy}.
Choose $\tilde{\gamma}_n=2\gamma_n$. Then,
\begin{align}
\log N(\tilde{\gamma}_n,\,  (M B^{\mathbb{H}_{\epsilon, K, 0}}_1+\gamma_n B_1) \cap \mathbb{H}_{\epsilon, K}, \, L^\infty) \leq & \log N(\gamma_n, \, M B^{\mathbb{H}_{\epsilon, K, 0}}_1,\, L^\infty)\,. \nonumber
\end{align}
By Lemma \ref{upper bounds on concentration function} and the assumptions that $K \leq \frac{1}{\gamma_n}$ and $n \geq 8 C_6$,
\begin{align}
\log N(\gamma_n, M B^{\mathbb{H}_{\epsilon, K, 0}}_1, L^\infty) \leq K \log\left(\frac{C_7 M}{\gamma_n}\right) \leq \frac{1}{\gamma_n} \log (2 \sqrt{2C_6} \sqrt{n})= \frac{1}{2\gamma_n} \log (8 C_6 n) \leq \frac{1}{\gamma_n} \log n\,. \nonumber 
\end{align}
To show that $ \frac{1}{\gamma_n} \log n \leq n \gamma^2_n$, it is equivalent to show that $\log n \leq n \gamma_n ^3$.  Since $\gamma_n=(\frac{n}{\log n})^{-\frac{d}{2q+2d}}$, $n \gamma_n ^3 =n^{1-\frac{3d}{2q+2d}}(\log n)^{\frac{3d}{2q+2d}}$.  Hence, 
$\log n \leq n \gamma_n ^3$ is equivalent to $(\log n)^{\frac{2q-d}{2q+2d}} \leq n ^{\frac{2q-d}{2q+2d}}$, which follows from $q \geq \frac{d}{2}$ and $n \geq \log n$.
In conclusion, we have
\[
\log N(\tilde{\gamma_n}, B^{\mathbb{H}_{t}}_M+\gamma_n B_1, L^\infty) \leq n \gamma^2_n=\frac{1}{4}n\tilde{\gamma}^2_n.
\]
\end{proof}

\section{Choosing the GL-GP parameters in the two balloons example}\label{parameter choices two balloons}

We demonstrate the process of choosing the GL-GP parameters using the two balloons dataset in section \ref{simulations}.  This example is designed to improve understanding of relationships between the parameters and the geometry of the underlying set. 
To illustrate relationships among the parameters, we evaluate the functional in \eqref{marginal likelihood} over a grid with $K=1,\cdots, 35$, $\epsilon^2=0.001+0.001 \times i$, $ i=1,\cdots, 29$, $t=0.01 \times j$, $j=1,\cdots, 100$ and $\sigma_{noise}^2=0.1 \times k$, $k=1,\cdots, 20$. 
Over the grid, the maximum is achieved when $K=30$, $\epsilon^2=0.012$, $t=0.33$ and $\sigma_{noise}^2=0.9$.  
In Figure \ref{plot marginal likelihood}, we plot the log of the marginal likelihood when $t=1$ and $\sigma_{noise}^2=1$. We can observe two local optima at $K=22$ and $\epsilon^2=0.011$, and $K=30$ and $\epsilon^2=0.012$, respectively. Our simple grid search algorithm can find the global optima even with multimodality. 
We emphasize that both $\epsilon$ and $K$ are related to the geometry of $S$, and their contribution to the covariance matrix is via a series of nonlinear operators; Hence, their roles in the marginal likelihood are complicated. 
Compared with $\epsilon$, which is a scale defining the intrinsic neighborhood and is closely related to the geometry of $S$, $t$ and $\sigma_{noise}$ are related to the regression function. Thus the dependence of the likelihood on $t$ and $\sigma_{noise}$ should be different. In Figure \ref{plot marginal likelihood}, we plot the log of the marginal likelihood over the grid with $\epsilon^2=0.001+0.001 \times i$, $i=1, \cdots, 29$ and $t=0.01 \times i$, $i=1,\cdots 60$, when we fix $K=24$ and $\sigma_{noise}^2=0.5$. The marginal likelihood is unimodal with an absolute maximum at $\epsilon^2=0.012$ and $t=0.28$. 

\begin{figure}
\centering
\includegraphics[width=1 \columnwidth]{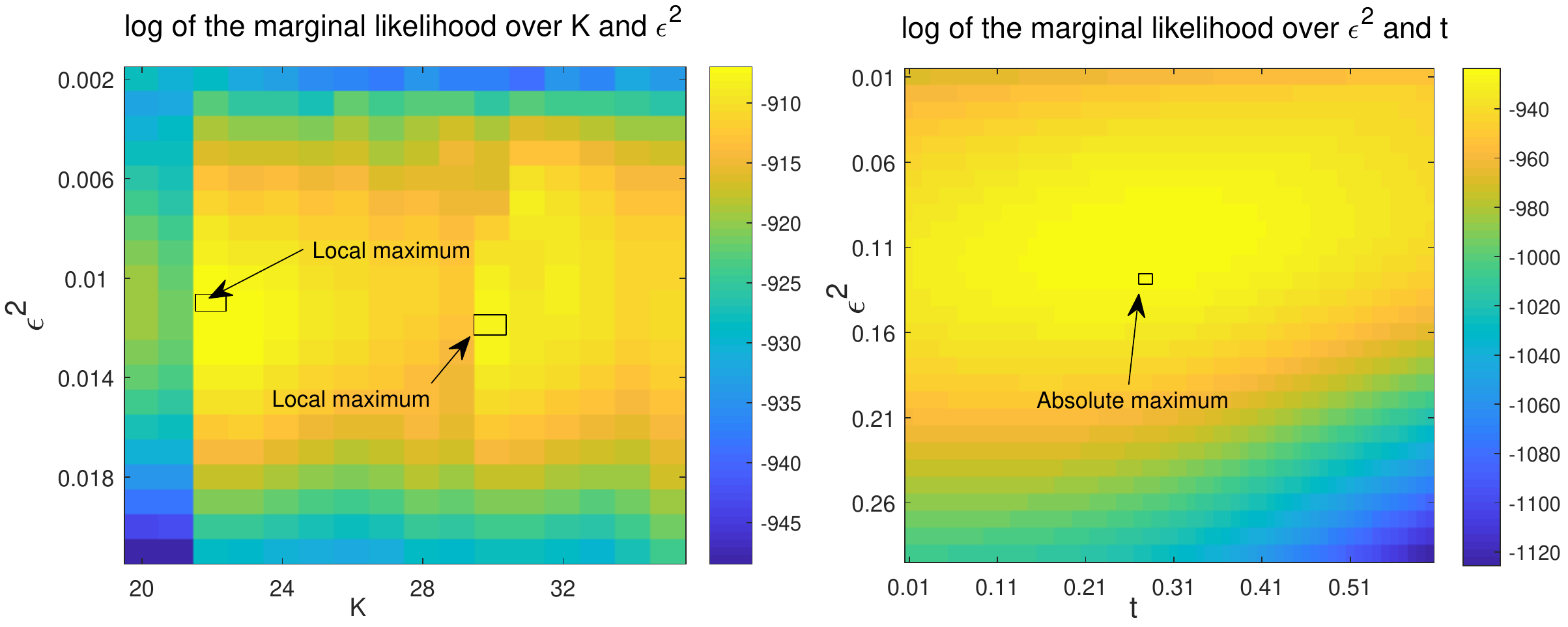}
\caption{Left: log marginal likelihood for $K$ from $20$ to $35$ and $\epsilon^2$ from $0.002$ to $0.021$, with $t=1$ and $\sigma_{noise}^2=1$. There are two local optima at  $K=22, \epsilon^2=0.011$ and $K=30, \epsilon^2=0.012$, respectively. Right: log marginal likelihood for  $\epsilon^2$ from $0.002$ to $0.03$ and $t$ from $0.01$ to $0.6$, with $K=24$ and $\sigma_{noise}^2=0.5$. There is a global optima at $\epsilon^2=0.012$ and $t=0.28$. }\label{plot marginal likelihood}
\end{figure}

\section{Performance of the Nystr\"{o}m extension on a spiral}\label{simulation nystrom}

We consider a spiral, $M$, embedded in $\mathbb{R}^{2}$ parametrized by 
$$\gamma(\theta)=((\theta+4)^{0.7}\cos(\theta), (\theta+4)^{0.7}\sin(\theta)) \in \mathbb{R}^{2},\mbox{ where }\theta \in [0, 8\pi).$$ 
We randomly sample $60$ labeled points $\{\theta_1, \ldots,\theta_{60}\}$ and $1500$ unlabeled points $\{\theta_{61}, \ldots, \theta_{1560}\}$ on $[0 ,8\pi)$ based on the uniform probability density function.  Let $x_i=\gamma(\theta_i)$ for $i=1,\ldots,1560$. The labels are sampled under equation~\eqref{eq:base} with $\sigma_{noise}^2=1$ and 
\begin{align}
f(\gamma(\theta)) =3\sin\big(\frac{\theta}{10}\big) +3\cos\big(\frac{\theta}{2}\big)+4\sin\big(\frac{4\theta}{5}\big)\,.
\end{align}
We sample $299$ indices $i_1 ,\cdots, i_{299}$ among $61, \cdots, 1560$ based on a uniform distribution. Let $\beta_j=\theta_{i_j}$ and $x'_j=\gamma(\beta_j)$ for $j=1, \cdots, 299$. We plot labels over $\{x_i\}$ for $i=1, \ldots, 60$,  and the ground truth over $\{x_1 , \cdots, x_{60}, x'_{1}, \cdots, x'_{299}\}$ in Figure \ref{spiralsmallsample}.  First, we show the performance of GL-GP by using the GL constructed from the dataset with $359$ sample points, $\{x_1, \cdots x_{60}, x'_{1} ,\cdots, x'_{299}\}$. By maximizing the marginal likelihood to choose the covariance parameters, we obtain  $K=28$, $\epsilon =\sqrt{0.13}$, $t=9$ and $\sigma_{noise}=1$ for the GL-GP, leading to an RMSE of $1.178$. For the GP with squared exponential covariance, we obtain $A=17$, $\rho=\sqrt{2.2}$, $\sigma_{noise}=\sqrt{0.7}$  leading to an RSME of $2.049$. We plot the predictions by the GL-GP and the GP over $\beta_i$ for $i=1 , \cdots, 299$ respectively in Figure \ref{extensionspiral}. 

\begin{figure}
\centering
\includegraphics[width=.8\columnwidth]{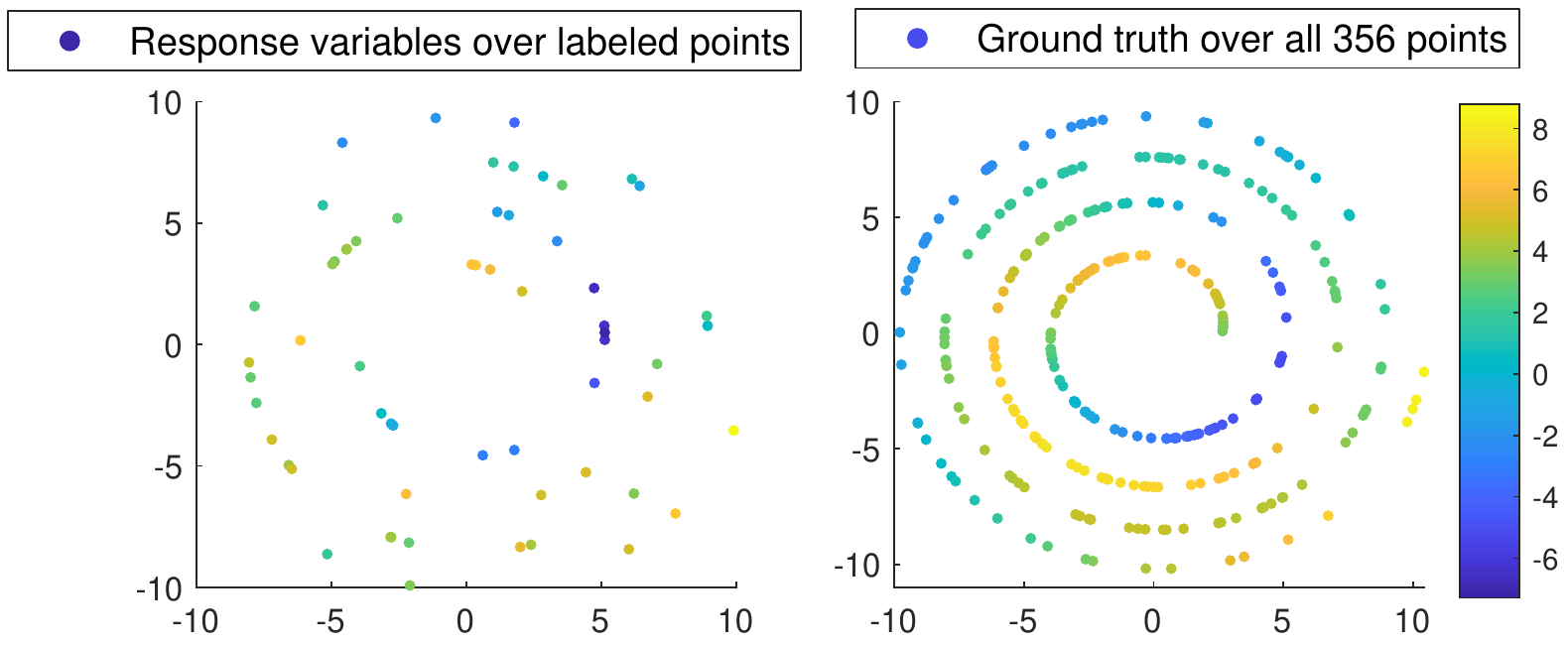}
\caption{Let $\{x_1, \dots, x_{60}\}$ be the labels points and $\{x'_{1}, \dots, x'_{299}\}$ be the unlabeled points. Left: The response variables over the labeled points. Rright: ground truth over all $359$ points. }. \label{spiralsmallsample}
\end{figure}

\begin{figure}
\centering
\includegraphics[width=.8 \columnwidth]{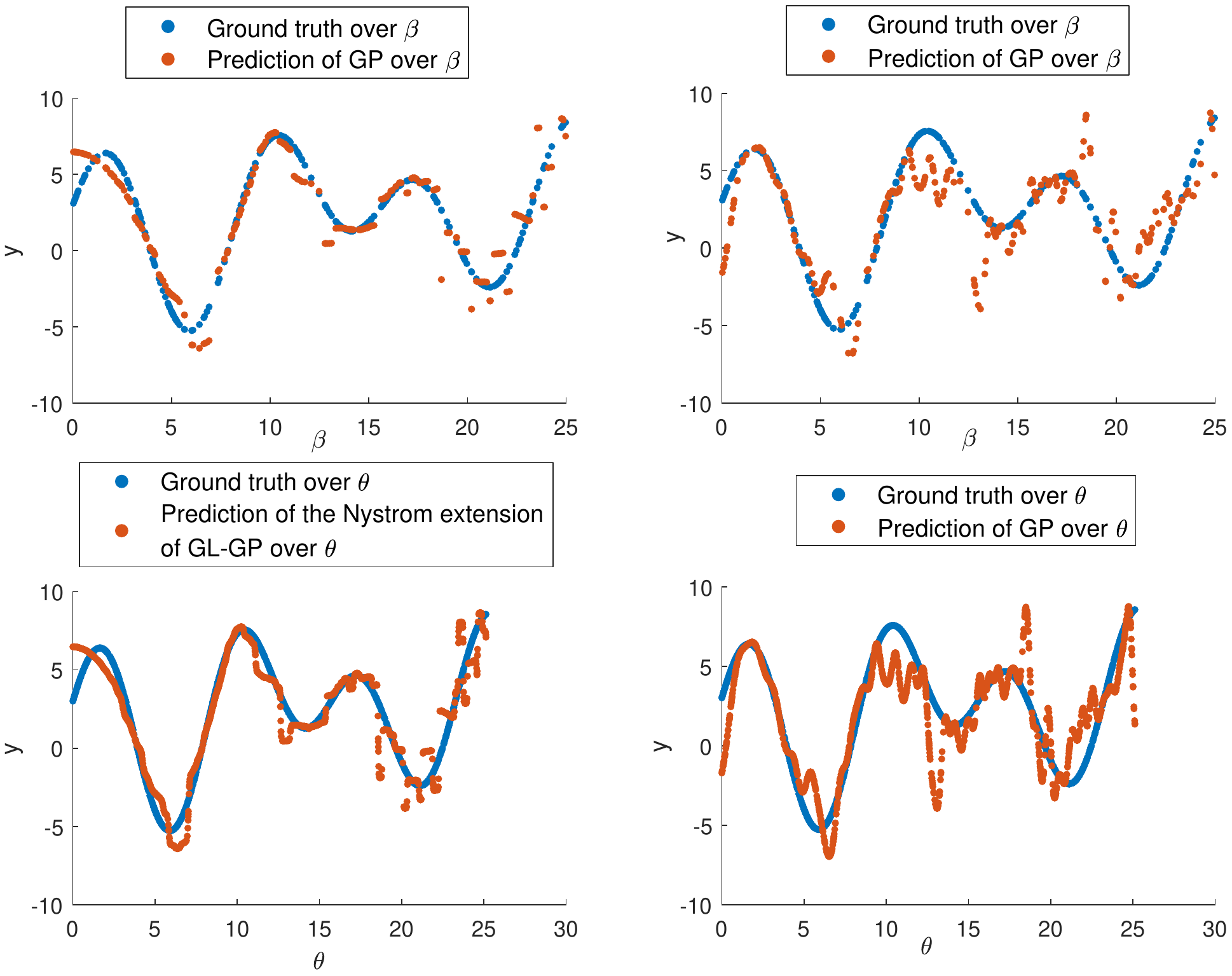}
\caption{Top left: The red points represent the prediction of GL-GP over $\{\beta_i\}$ for $i=1, \cdots, 299$ with $RSME=1.178$. The blue points correspond to the ground truth over all $\{\beta_i\}$ for $i=1, \cdots, 299$.  Top right: The red points represent the prediction of the GP over $\{\beta_i\}$ for $i=1, \cdots, 299$ with $RSME=2.049$. Bottom left: The red points represent the prediction of the Nystr\"{o}m extension method over $\{\theta_i\}$ for $i=61, \cdots, 1560$ with $RSME=1.235$. Bottom right: The red points represent the prediction of the GP over $\{\theta_i\}$ for $i=61, \cdots, 1560$ with $RMSE=1.967$.}. \label{extensionspiral}
\end{figure}

Next, we apply the Nystr\"{o}m extension. We extend the covariance matrix constructed from the $359$ points with paramters $K=28$, $\epsilon =\sqrt{0.13}$, $t=9$ and $\sigma_{noise}=1$ to a covariance matrix over $\{x_i\}_{i=1}^{1560}$ by using \eqref{DBGP covariance matrix extension}. The RSME between the predicton over $\{x_i\}_{i=61}^{1560}$ by this extended covariance matrix and the true values of the regression function is $1.235$. For GP with squared exponential covariance, the RSME between the predicton over $\{x_i\}_{i=61}^{1560}$ by this extended covariance matrix and the true values of the regression function is $1.967$. We plot the prediction by  the Nystr\"{o}m extension over $\{\theta_i\}_{i=61}^{1560}$ in Figure \ref{extensionspiral} and compare it with the prediction of squared exponential GP. Since the covariance matrix is an extension of the GL-GP  covariance matrix constructed over the small-size samples, the prediction is not as good as the prediction by using GL-GP covariance matrix constructed directly from the large-size samples in section \ref{simulations}. But, the Nystr\"{o}m extension method still improves the performance from the GP with the squared exponential covariance.

\section{Performance of Graph based Gaussian Processes on a square}

\begin{figure}
\centering
\includegraphics[width=.65\columnwidth]{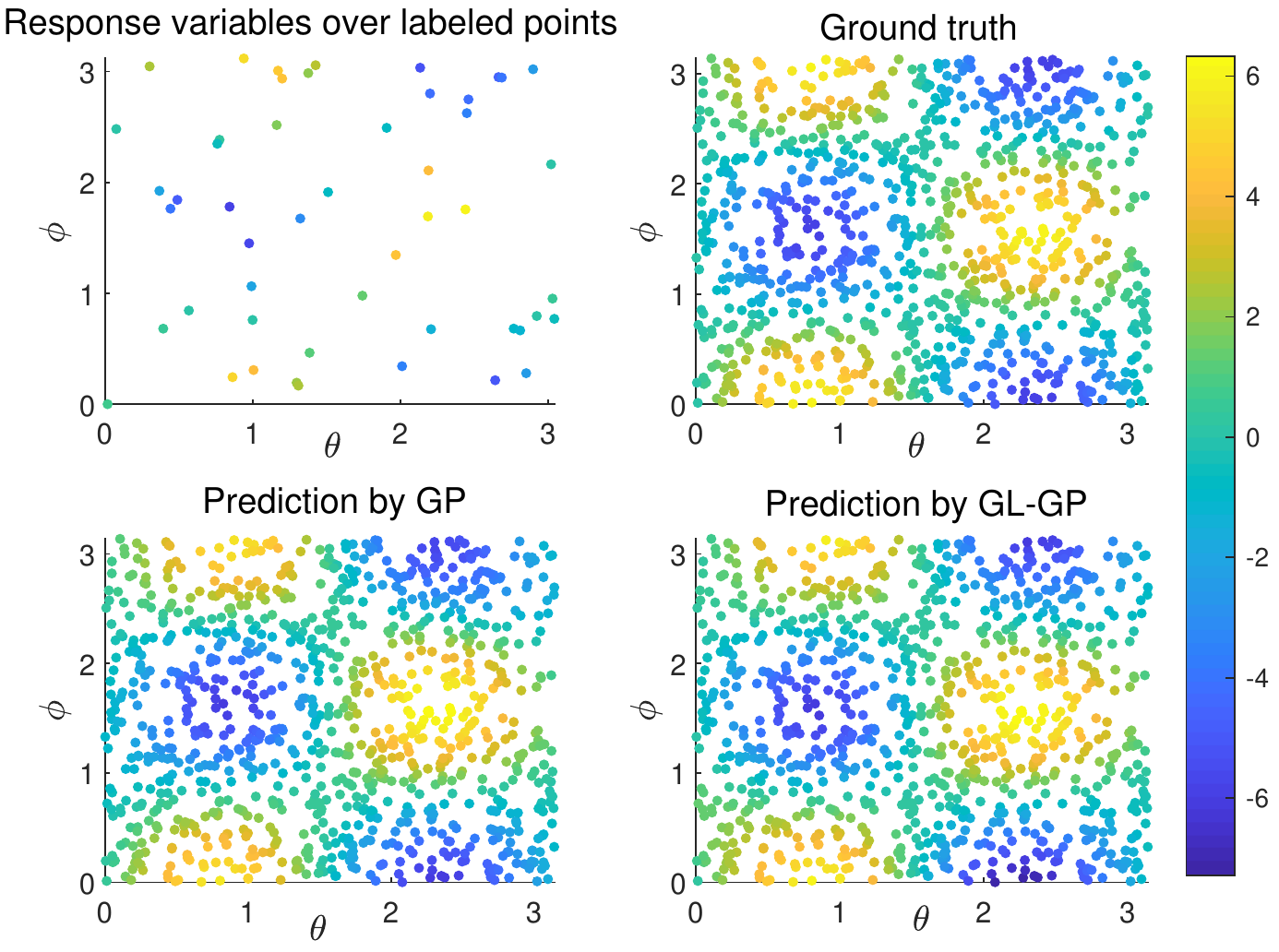}
\caption{ Let $\{x_1, \dots, x_{50}\}$ be the labels points and $\{x_{51}, \dots, x_{1050}\}$ be the unlabeled points. Top left: response variables over the labeled points. Top right: ground truth over all points. Bottom left: prediction of GL-GP over the unlabeled points with $RMSE=0.435$.  Bottom right: prediction of the GP over the unlabeled points with $RMSE=0.387$.}. \label{rectangle comparison}
\end{figure}

In this section, we apply the GL-GP to samples on a (flat) square in  $\mathbb{R}^{2}$. A square has trivial geometry, in particular, the intrinsic distance between two points on the square is the same as the extrinsic Euclidean distance, and there is no conflict between the global and local geometries. Hence, the GP with the square exponential covariance also reflects the intrinsic geometry of the set. We compare those two methods in this section. Let $S$ be a square of side length $\pi$ in $\mathbb{R}^{2}$. Each point $x$ in $S$ has coordinates $x=(\theta, \phi)$, where $0 \leq \theta, \phi \leq \pi$. We randomly sample $50$ labeled points $\{x_1, \ldots,x_{50}\}$ and $1000$ unlabeled points $\{x_{51}, \ldots, x_{1050}\}$ in $S$ based on the uniform probability density function. The labels are sampled under equation~\eqref{eq:base} with $\sigma_{noise}=0.5$ and 
\begin{align}\label{regression function square}
f(x)=f(\theta, \phi) =6\sin(2\theta) \cos(2\phi)\,.
\end{align}
We plot the labels over $x_i$ for $i=1, \ldots, 60$,  and the ground truth in the top two panels in Figure \ref{rectangle comparison}. By maximizing the marginal likelihoods to choose the covariance parameters, we obtain  $K=19$, $\epsilon =\sqrt{0.2}$, $t=0.2$ and $\sigma_{noise}=\sqrt{0.3}$ for the GL-GP, leading to an RMSE of $0.435$. For the GP with the square exponential covariance, we obtain $A=18$, $\rho=\sqrt{0.3}$, $\sigma_{noise}=\sqrt{0.2}$ and an RMSE of $0.387$. The bottom two panels in Figure \ref{rectangle comparison} show the predictions of the two different approaches.

Next, we sample $100$ groups of points on the square, $\{x^j_{1}, \ldots, x^j_{1050}\}$, $j=1,\cdots, 100$. For each group, the first $50$ points are labeled, while the remaining points are unlabeled.  The labels are sampled under equation~\eqref{eq:base} and the regression function \eqref{regression function square} with $\sigma_{noise}=0.5$. For each group, we maximize the marginal likelihoods to choose the covariance parameters and calculate the RSME between the predictions of the two methods and the ground truth over the unlabel points. The average RSME for GL-GP is $0.601$ with variance $0.027$. The average RSME for the GP with the squared exponential covariance is $0.584$ with variance $0.032$.

\end{document}